\documentclass[runningheads,a4paper]{llncs}
\pdfoutput=1
\usepackage{etex}
\usepackage{amssymb,amsmath,enumerate,stmaryrd}
\usepackage[all]{xy}
\usepackage{pslatex}
\usepackage{pdfsync}
\usepackage{bussproofs}
\usepackage{color,graphicx,xcolor}
\usepackage{lscape}
\usepackage{pifont}
\usepackage{wrapfig}
\usepackage{lipsum}
\usepackage{tikz}
\usepackage{rotating}
\usetikzlibrary{arrows,shadows,matrix,shapes,positioning,chains,calc}
\usepackage[underline=false,rounded corners=false]{pgf-umlsd}
\usepackage{wasysym}

\newif\ifcscomp
  \cscomptrue
  \cscompfalse

\newif\ifcost
  \costfalse
  \costtrue

\newif\ifnat
  \natfalse
  \nattrue

\newif\ifemi
  \emifalse
  \emitrue

\newif\ifthm
  \thmfalse
  \thmtrue

\newif\ifllncs
  \llncsfalse
  \llncstrue

\newif\ifarticle
  \articletrue
  \articlefalse

\newif\iftesi
  \tesitrue
  \tesifalse

\newif\ifmod
  \modtrue
  \modfalse

\newif\ifqapl
  \qapltrue
  \qaplfalse

\newif\ifqapl
  \qaplfalse
  \qapltrue

\newif\ifpar
  \partrue
  \parfalse

\newcommand{\cda}{\mbox{CDA$^\sharp$}}
\newcommand{\inleq}[3]{{#1} \leq {#2} \leq {#3}}
\newcommand{\defref}[1]{Def.~\ref{#1}}
\newcommand{\figref}[1]{Fig.~\ref{#1}}
\newcommand{\secref}[1]{\S~\ref{#1}}
\newcommand{\exref}[1]{Ex.~\ref{#1}}

\newcommand{\thmref}[1]{Th.~\ref{#1}}
\newcommand{\el}[2]{{#1}\llbracket{#2}\rrbracket}
\newcommand{\hcvel}[2]{\hcv{#1}\llbracket{#2}\rrbracket}

\newcommand{\qlab}[1]{{\mathfrak L}({#1})}

\newcommand{\hdof}[1]{\mathcal H_{\nretohds{#1}}}

\renewcommand{\ne}{\mathsf{ne}}

\newcommand{\letters}{\mathcal{S}}

\newcommand{\longversion}[1]{}

\newcommand{\nretohds}[1]{\llparenthesis #1 \rrparenthesis}

\usepackage{fge}

\newcommand{\step}[3]{#2 \stackrel{#1}{\to} #3}

\newcommand{\trans}{\mathit{tr}}

\newcommand{\rec}{\mathit{reach}}
\newcommand{\emptystr}{\epsilon}

\newcommand{\pop}[1]{\mathtt{pop}}
\newcommand{\push}[1]{\mathtt{push}}

\newcommand{\weight}[1]{|#1|}

 \newcommand{\Real}[1]{\mathrm{Real}}

  \newcommand{\compile}[2]{\ifthenelse{\equal{#1}{yes}}{#2}{}}

  \newcommand{\cf}[2]{
    \fontsize{#1}{#1}{\selectfont{#2}}
  }
  \ifemi
  \usepackage{ulem}
    \newcommand{\emi}[1]{{\marginpar{\cf{6}{{#1}}}}}
    \newcommand{\emic}[2]{\ \\[.1em]
      \definecolor{shadecolor}{rgb}{1,0.99,0.9}
      \fcolorbox{red}{shadecolor}{\parbox{\linewidth}{ 
            \color{gray}
            {\color{blue} #2}{\sf #1}
            }}
      \ \\[.1em]}
    
    \else
    \newcommand{\emi}[1]{}
    \newcommand{\emic}[2]{}
    
    \fi

  \ifemi

  \else

  \fi

  \newcommand{\st}{\ \ \big| \ \ }

     \ifmod
       
     \else
       
     \fi

  \ifcscomp
    
  \else
    
  \fi

  \newcommand{\proofend}{\mbox{$\Box$}}

  \newcommand{\mmdef}{\mbox{$\;\;\stackrel{\mathrm{def}}{=}\;\;$}}

  \ifnat
  \fi

  \newcommand{\upd}[2]{[{#1} \mapsto {#2}]}

  \ifpar
    
  \else
    
  \fi

  \newcommand{\comment}[1]{}
  \iftesi

  \else

  \fi

\newcommand{\conf}[1]{\langle {#1} \rangle}

\newcommand{\sep}{\;\;\mid\;\;}

\iftesi
  
\else
  
\fi

\newcommand{\names}{\mbox{$\mathcal{N}$}}

\ifarticle
  \newtheorem{theorem}{Theorem}[section]
  \newtheorem{definition}{Theorem}[section] 
  \newtheorem{proposition}{Theorem}[section] 
  \newtheorem{lemma}{Theorem}[section]
  \newtheorem{corollary}{Theorem}[lemma] 
  \newtheorem{remark}{Theorem}[section]
  \newtheorem{observation}{Theorem}[section]
  \newtheorem{notation}{Theorem}[section]
  \newtheorem{example}{Theorem}[section]
\else
  \ifthm
  \else
    \newtheorem{theorem}{Theorem}[section]
    \newtheorem{definition}{Definition}[section] 
    \newtheorem{proposition}{Proposition}[section] 
    \newtheorem{lemma}{Lemma}[section]
    \newtheorem{corollary}{Corollary}[lemma] 
    \newtheorem{remark}{Remark}[section]

    \newtheorem{example}{Example}[section]
  \fi
\fi

\ifpar
\else

\fi

\iftesi
  
\else
  
\fi

\newcommand{\Nat}{\mathbb{N}}

\newcommand{\card}[1]{\left\|#1\right\|}

\newcommand{\dom}[1]{\mathit{dom}(#1)}

\newcommand{\tuple}[1]{\langle#1\rangle}

\definecolor{darkgreen}{rgb}{0,0.5,0}
\definecolor{darkblue}{rgb}{0,0,0.8}
\definecolor{darkred}{rgb}{0.9,0,0}

\usepackage{url}
\urldef{\mailsa}\path|{kurz, emilio}@mcs.le.ac.uk|
\urldef{\mailsb}\path|tomoyuki.suzuki@cs.cas.cz|

\begin{document}

\mainmatter

\title{Nominal Regular Expressions \\ for Languages over Infinite Alphabets}
\subtitle{Extended Abstract}

\author{Alexander Kurz\inst{1} \and Tomoyuki Suzuki\inst{2} \and Emilio Tuosto\inst{1}}

\authorrunning{A.~Kurz, T.~Suzuki and E.~Tuosto}

\institute{Department of Computer Science, University of Leicester, UK\\
\and Institute of Computer Science, Academy of Sciences of the Czech Republic,
Czech Republic
}

\toctitle{Lecture Notes in Computer Science}

\maketitle

\begin{abstract}
  We propose regular expressions to abstractly model and
  study properties of resource-aware computations.
  Inspired by nominal techniques -- as those popular in process calculi --
  we extend classical regular expressions with names (to model computational resources) and suitable operators (for allocation,
  deallocation, scoping of, and freshness conditions on resources).
  We discuss classes of such nominal regular expressions, show
  how such expressions have natural interpretations in terms of
  languages over infinite alphabets, and give Kleene theorems to
  characterise their formal languages in terms of nominal automata.
\end{abstract}

\newcommand{\nex}[1]{\mathit{#1}}
\newcommand{\pexpleft}[1]{\langle_{\nex{#1}}}
\newcommand{\pexpright}[2]{\rangle_\nex{#1}^\nex{#2}}
\newcommand{\pexp}[3]{{\pexpleft{#1}\nex{#2}\pexpright{#1}{#3}}}
\newcommand{\npexp}[2]{{\pexpleft{#1}\nex{#2}\pexpright{#1}{}}}
\newcommand{\pclose}[1]{\circlearrowleft_{#1}}
\newcommand{\etw}[3]{#1\ \ddagger\ #2 \ \ddagger\ #3}
\newcommand{\length}[1]{\mathit{lth(#1)}}
\newcommand{\lstlength}[1]{\mathit{lstlth(#1)}}
\newcommand{\lstnames}[1]{\underline{#1}}
\newcommand{\cplus}{\hat+}
\newcommand{\ccirc}{\hat\circ}
\newcommand{\cks}{\hat\ast}
\newcommand{\cdmd}{\hat\Diamond}
\newcommand{\lplus}{\check+}
\newcommand{\lcirc}{\check\circ}
\newcommand{\lks}{\check\ast}
\newcommand{\ldmd}{\check\Diamond}
\newcommand{\llhd}{\mathaccent\lhd{\lhd}}
\newcommand{\langof}{\mathbf{L}}
\newcommand{\namesexp}{\mathsf{N}}
\newcommand{\lstofn}[2]{{#1}^{\left({#2}\right)}}
\newcommand{\gft}[1]{\underline{#1}}
\newcommand{\alang}[1]{\mathcal{L}_{\textit{#1}}}
\newcommand{\ltw}{\alang{two}}
\newcommand{\lpd}{\alang{all}}
\newcommand{\lfst}{\alang{fst}}
\newcommand{\llst}{\alang{lst}}
\newcommand{\ltze}{\alang{tze}}
\newcommand{\ltwl}{\alang{2lst}}
\newcommand{\lwlst}{\alang{wlst}}
\newcommand{\simples}{\alang{ses}}
\newcommand{\simplet}{\alang{onet}}
\newcommand{\moret}{\alang{ths}}
\newcommand{\cof}[1]{\underline{#1}}
\newcommand{\namesofw}[1]{|{#1}|}
\newcommand{\ihsv}[2]{{#2}\natural{#1}}
\newcommand{\hcv}[1]{\mathit{cv}\left({#1}\right)}
\newcommand{\hstry}[1]{{#1}^\text{\mbox{\ding{173}}}}
\newcommand{\hsofcvn}[2]{{#1}^{\left({#2}\right)}}
\newcommand{\fne}{\mathsf{fne}}
\newcommand{\reg}[1]{\mathsf{reg}\left({#1}\right)}
\newcommand{\app}[2]{#2 + [\nex{#1}]}
\newcommand{\ilist}[2]{\left(#1 \right)_{#2}}
\newcommand{\trx}[2]{\left(\nex{#1}\ \nex{#2}\right)}
\newcommand{\lfresh}{\,\#\,}
\newcommand{\gfresh}{\,\underline{\#}\,}
\newcommand{\lgfresh}{\mbox{\textbf{\#}$^\partial$}}
\newcommand{\nec}[1]{\mathbb{C}}
\newcommand{\seqd}[1]{\textsf{#1}}
\newcommand{\start}{\seqd{\footnotesize{start}}}
\newcommand{\aR}{\seqd{A}}
\newcommand{\bR}{\seqd{B}}
\newcommand{\select}{\seqd{select}}
\newcommand{\selection}{\seqd{selection}}
\newcommand{\request}{\seqd{req}}
\newcommand{\result}{\seqd{res}}
\renewcommand{\arraystretch}{1.5}
\newcommand{\lalloc}[1]{{\lfloor #1 \rfloor}}
\newcommand{\galloc}[1]{{\underline{\lfloor #1 \rfloor}}}

\section{Introduction}\label{sec:intro}

  We equip regular expressions with different types of name binders
  in order to define a theoretical framework to model and study
  computations involving resource-handling.
  In particular, we are interested in computations where resources
  can be freshly generated, used, and then deallocated.
  We use names to abstract away from the actual nature of resources;
  in fact, we adopt a very general notion of computational resources
  that encompass e.g., memory cells, communication ports,
  cryptographic keys, threads' identifiers etc.
  This allows us to use an infinite set $\names$ of
  \emph{names} to denote resources while binders and freshness
  conditions formalise the life-cycle of resources.
  Freshness conditions are taken from the theory of \emph{nominal sets} where
  $n \lfresh X$ states that $n \in \names$ does not appear (free) in a
  structure $X$, which can be a set of names or, more generally, a term
  built from names and set-theoretic constructions
  \cite{GabbayP99}.
  To this end, $\names$ is equipped with the
  action of the finitely generated permutations, which then extends to
  words and languages.
  Moreover, in the spirit of nominal sets, all languages of interest
  to us will be closed under the action of permutations.

  Together with the use of classical operators of regular languages we
  then define languages over infinite alphabets including $\names$.
  As we will see, there are different ways of extending regular
  expressions with binders or freshness conditions.
  We consider some natural definitions of nominal regular expressions
  and give Kleene theorems to characterise their languages in terms
  of automata.

Besides having interesting theoretical aspects, automata and languages
over infinite alphabets can also be adopted to specify and verify
properties of systems. This is very much in the spirit of HD-automata
which were invented to check equivalence of $\pi$-calculus processes.
We use a scenario based on distributed choreographies (as those envisaged by
W3C~\cite{w3c:cho}) and show how to specify correct
executions of realisations of choreographies, which can be described as message
sequence charts. The figure below
\begin{wrapfigure}[8]{c}[100pt]{6cm}\vspace{-.7cm}
\resizebox{2.3cm}{!}{
  \begin{sequencediagram}
    \newthread[yellow]{ra}{\aR}
    \newinst[1]{rb}{\bR}

    \mess{ra}{\start}{rb}{}
    \mess{rb}{\select}{ra}
    \mess{ra}{\selection}{rb}
    \begin{sdblock}[green]{Loop}{}
      \mess{ra}{\request}{rb}
      \mess{rb}{\result}{ra}
    \end{sdblock}
  \end{sequencediagram}
}
\end{wrapfigure}
\noindent
describes a protocol between two distributed components
\aR\ and \bR.
After \start ing the protocol, \aR\ waits for a list of services offered by \bR\
to \select\ from (for simplicity, data is not represented).
Upon request, \bR\ replies to \aR\ with a list of options to \select\ from.
Then \aR\ makes her \selection\ and loops to send a number of requests (with a
\request\ message) to each of which \bR\ replies with a result (with the
\result\ message).
We describe a few possible realisations of the above choreography.

As a first implementation, think of \aR\ and \bR\ as repeatedly executing the
protocol.
This can be conveniently captured using session types~\cite{hvk98}, where each
run of the protocol is uniquely identified with \emph{session names}.
Languages over infinite alphabets can suitably specify such runs; for instance,
consider
\begin{equation*}
  \simples = \big\{ a\, b\, r_0 \cdots r_k \st \forall k \in \Nat, \forall
  \inleq 0 {i \not= j} k. r_i \not= r_j \big\}
\end{equation*}
where $a$ and $b$ are two distinct letters representing the two components
executing \aR\ and \bR\ (and $\Nat$ is the set of natural numbers).
A word $a\, b\, r_0 \cdots r_k \in \simples$ corresponds to a trace where $a$ and
$b$ engage in $k$ runs of the protocol and $r_i$ identifies the $i$-th run.

Another suitable implementation would be one where \bR\ is multi-threaded and
for instance activates a new thread for each request in the loop.
A simplification usually adopted in session-based frameworks is that a
thread serving a request cannot be involved in other requests.
Assuming that in each loop \aR\ makes two requests:
\begin{eqnarray*}
  \simplet & = & \big\{ a \, b \, r_0 \, p_0 \, p'_0 \cdots r_k \, p_k \, p'_k
  \st \forall k \in \Nat, \forall \inleq 0 i k. p_i \lfresh \{ p'_i, r_i\} \land
  p'_i \lfresh \{ p_i, r_i \}
  \\ \nonumber & &
  \qquad \qquad \qquad \qquad \qquad \qquad \land r_i \lfresh
  \{r_0,\ldots,r_{i-1}, p_0,\ldots,p_{i-1}, p'_0,\ldots,p'_{i-1}\}
  \big\}
\end{eqnarray*}
where $p_i$ and $p'_i$ are the names of the two processes that serve the first
and the second request from \aR\ in the $i$-th run.
Note that in $\simplet$, $p_i$ is not required to be distinct
from $p_j$ (or from $p'_j$).

In yet another realisation, the threads of $\bR$ would have to activate other
threads to serve the requests from $\aR$.
(In session-based frameworks this is known as \emph{delegation}.)
As a simplified model of  these traces, let
\begin{eqnarray*}
  \alang{thr}(r_i) & = & \bigcup_{h \in \Nat}\big\{ v_0d \cdots v_hd \st \forall
   \inleq 0 j h, \exists p \not= p'\in \names\setminus\{r_i\}. v_j \in
   \{p,p'\}^\ast \big\}
\end{eqnarray*}
(where $d$ marks when a thread wants to delegate its computation) and consider:
\begin{eqnarray*}
  \moret & = & \bigcup_{k \in \Nat}\big\{ ab r_0 w_0 r_0 \cdots r_k w_k r_k \st
  \forall \inleq 0 i k. w_i \in \alang{thr}(r_i) \land r_i \lfresh r_0w_0 \cdots
  r_{i-1}w_{i-1} \big\}
\end{eqnarray*}

  An original contribution of this paper is the introduction of 
  \emph{relative global freshness}, a notion of freshness that enables
  us to control how to \emph{forget} (i.e. deallocate) names.
    For example, we will see that the languages $\simplet$
  and $\moret$ can be accepted by automata using
  relative global freshness.
  In fact, crucially, the freshness condition on the names for threads
  allows names to be re-used once the run is finished.
  This is possible due to the peculiar ability of relative global
  freshness to ``forget'' names.

Related to this, we point out that relative global freshness is
different from global freshness as defined
in~\cite{Tzevelekos11}. Indeed, the classes of languages we consider
are all closed under concatenation, in contrast to
\cite{Tzevelekos11}.


\section{Nominal regular expressions}
We fix a finite set $\letters$ of `letters'  and a countably infinite set $\names$ of `names'
 and consider languages over infinite alphabets as sets
of finite words over $\letters \cup \names$.

We define \emph{nominal regular expressions} (NREs, for short) by
extending classical regular expressions with names $n \in \names$ and
name binders.
We use different types of angled brackets $\npexp {} \_$ (decorated
with sub- and/or super-scripts) to denote binders.
The brackets identify the scope of the binder and are indexed by the
name they bind.
The interpretation of NREs is defined formally in \secref{sec:exp2lang}, here
we discuss the basic ideas.

\medskip\noindent\textbf{Basic nominal regular expressions} (b-NREs) are 
defined by 
\begin{equation}
 \label{eq:bnre}
  \ne \ ::= \ \nex{1} \sep \nex{0} \sep n \sep s \sep \ne + \ne \sep
  \ne \circ \ne \sep \ne^\ast \sep \npexp{n}{\ne}
\end{equation}
where $\nex{1}$ and $\nex{0}$ are constants to denote the language
consisting of the empty word $\emptystr$ only and the empty language,
$n$ and $s$ range over $\names$ and $\letters$, resp., the operators
$+$, $\circ$, and $\_^\ast$ are familiar from regular expressions, and
$\npexp{n}{\ne}$ is our notation for name binding: $\ne$ is the scope
of the binder and the occurrences of $n$ in $\ne$ are bound.
For example, in $n\,\npexp{n}{n}\,n$, the first and the last
occurrences of $n$ are free, whilst the occurrence in the bracket is
bound and is interpreted as a locally fresh name, that is a name
distinct from the occurrences of $n$ outside the scope of the binder.
So the language of $n \, \npexp{n}{n} \, n$ is
\[
 \big\{n m n \in \names^3 \st m \in \names \setminus \{n\} \big\}
\]
The (de-)allocation mechanism featured by b-NREs is very simple: a
fresh name is allocated when entering the scope of a binder and
deallocated when leaving.
The freshness conditions on the allocated name require that it is
distinct from the other currently allocated names.
The next class of NREs has a more sophisticated deallocation mechanism.

\medskip\noindent\textbf{NREs with permutations} (p-NREs)
extend b-NREs by permutation actions:
\begin{equation}
 \label{eq:pnre}
  \ne ::= \nex{1} \sep \nex{0} \sep n \sep s \sep \ne + \ne \sep \ne \circ \ne
  \sep \ne^\ast \sep \pexp{n}{\ne}{m}
\end{equation}
Novel with respect to to b-NREs is the notation $\pexpright{n}{m}$ which evokes the name transposition
$\left(m\ n\right)$ to be applied when leaving the scope of the
binder.
In other words, when $\pexpright{n}{m}$ closes the scope of $n$, the name $m$ is deallocated while we leak $n$
by replacing all free occurrences of $m$  after
$\pexpright{n}{m}$ with $n$.
For example, in the p-NRE $\pexp{m}{m \pexp{n}{n}{m}m}{m}$, the 
occurrences of $m$ after $\pexpright{n}{m}$ actually mean $n$ since
$m$ is replaced by $n$ when the scope of $n$ is closed. 
For example, $\pexp{m}{\ (\pexp{n}{\,n}{m})^{\,\ast\,}\ }{m}$ is the language of all words where any two successive names must be different, see \cite{KurzST12tcs} for more details.

\begin{remark}
  We consider b-NREs as special p-NREs, identifying $\npexp n \ne$  with $\pexp{n}{\ne}{n}$.
\end{remark}

The deallocation device of p-NREs requires some care; intuitively,
a name $m$ can be deallocated only \emph{after} it has been allocated.
This condition is formalised by requiring that in a p-NRE $\ne$ any
subexpression $\pexp n {\ne'} m$ with $n \neq m$ occurs within the scope
of a binder $\pexp m {\_} {m'}$.
We will consider only p-NREs satisfying this condition.

\medskip\noindent\textbf{NREs with underlines} (u-NREs) extend b-NREs by \emph{relative global freshness}:
\[
 \ne ::= \nex{1} \sep \nex{0} \sep n \sep \cof{n} \sep s \sep \ne + \ne \sep \ne
 \circ \ne \sep \ne^\ast \sep \npexp{n}{\ne}
\]
Novel with respect to to b-NREs is the notation $\cof n$, which denotes relative global freshness. It requires that the name represented by $\cof n$ is
  distinct from any name allocated \emph{after} $n$.
  For instance, the language of the u-NRE $\npexp n { n \npexp m {mn} \cof n
    \npexp m m}$ is
  \[
  \big\{ n m n n' m' \in \names^5 \st n \neq m \text{ and } n' \neq m, n
  \text{ and } m' \neq n'\big\}
  \]
  where $n'$ corresponds to the name denoted by $\cof n$ in the u-NRE
  has to be different from both $n$ and $m$ even if the latter has
  been deallocated.
  Note the difference with the freshness condition on $m'$
  (corresponding to the second binder on $m$ in the u-NRE) which is
  required to be different only from $n'$.  

\medskip\noindent\textbf{NREs with underlines and permutations}
(up-NREs) combine p-NREs and u-NREs:
\[
 \ne ::= \nex{1} \sep \nex{0} \sep n \sep \cof{n} \sep s \sep \ne + \ne \sep \ne
 \circ \ne \sep \ne^\ast \sep \pexp{n}{\ne}{m}
\]
where the conditions on p-NREs also hold for up-NREs.

\medskip\noindent\textbf{Examples} of NREs for the languages of
\secref{sec:intro} are
 \begin{table}[htbp]
  \begin{minipage}{\linewidth}\centering
   \begin{tabular}{|c|c|c|}\hline
    Language & Corresponding NRE & Type of NRE \\\hline
    $\simples$ & $ab\npexp{n}{\cof{n}^\ast}$ & u-NRE \\\hline
    $\simplet$ &
	$ab\npexp{n}{\left(\cof{n}\npexp{m}{m\npexp{l}{l}}\right)^\ast}$ & u-NRE
	    \\\hline
    $\moret$ &
	$ab\npexp{n}{\left(\cof{n}\npexp{m}{\left(\pexp{l}{\left(m+l\right)^\ast
	d}{m} + \npexp{l}{\left(m+l\right)^\ast d}\right)^\ast}\right)^\ast}$ &
	    up-NRE \\\hline
   \end{tabular}
  \end{minipage}
 \end{table}

\noindent This correspondences are obtained by applying the method presented in
\secref{sec:exp2lang}.
The details of a more complex example are given in Appendix \ref{sec:complex}.


\section{Chronicle deallocating automata}
\label{sec:asharp}
The class of \emph{chronicle deallocating automata}
characterises languages over $\letters \cup \names$.
\exref{ex:automata:simple} below gives an intuition of our automata
which are defined in \defref{def:cda}.
\begin{example}
 \label{ex:automata:simple}\it
 The language $\simples$ in \secref{sec:intro} is accepted by the automaton in
 the figure below

 \noindent
 \begin{minipage}{\linewidth}
  \begin{wrapfigure}[5]{r}[40pt]{.4\linewidth}\vspace{-.8cm}
   \resizebox{3.5cm}{!}{
   \includegraphics{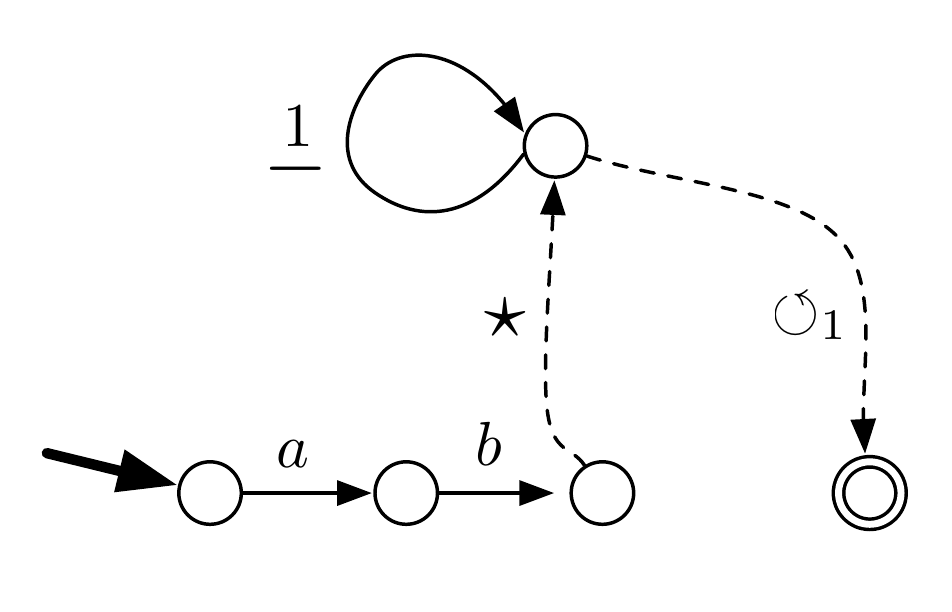}
   }
  \end{wrapfigure}
  that first consumes the two letters $a$ and $b$, then allocates a
  locally fresh name in register $1$ with the $\star$-transition
  (without consuming any letter).
  By means of repeated $\cof{1}$-transitions, it can consume
  $r_0,\ldots,r_k$ guaranteeing their freshness with respect to the
  chronicle for the first register.
  Finally, the name in the first register is deallocated and the
  final state is reached.
 \end{minipage}
\end{example}

Given a natural number $k$, in the following, $\reg k \mmdef
\{1,\ldots,k\}$ denotes a set of $k$ registers.
The empty word is denoted by $\emptystr$.
\begin{definition}
 \label{def:cda}
 A \emph{chronicle deallocating automaton (over
   $\letters$)} is a finite-state automaton $\tuple{Q,q_0,F,\trans}$
 with states $Q$, initial state $q_0 \in Q$, final states $F \subseteq
 Q$ and transition relation $\trans$ such that
 \begin{itemize}
  \item $Q$ is equipped with a map $\card{\_} \colon Q \to \Nat$ such that
	$\card q = 0$ for all $q \in F \cup \{q_0\}$;
  \item writing $\reg{q}$ for $\reg{\card{q}}$, each $q \in Q$ has a
    set of possible labels
    \[
    \qlab{q} \ \mmdef\  \letters \ \cup \ \reg{q} \ \cup \ \{\star\} \ \cup
    \ \{\gft{i} \st i \in \reg{q}\} \cup \{\pclose{i} \st i \in \reg{q}\}
    \]
  \item for $q \in Q$ and $\alpha \in \qlab{q} \cup \{\emptystr\}$,
    the set $\trans(q,\alpha) \subseteq Q$ contains the
    \emph{$\alpha$-successor states} of $q$ satisfying the conditions
    below for all $q' \in \trans(q,\alpha)$:
    \begin{eqnarray*}
      \card{q'} & = & \card q + 1, \qquad \text{if } \alpha = \star
      \\
      \card{q'} & = & \card q - 1, \qquad \text{if } \alpha = \pclose{i}
      \text{ for } i \in \reg{q}
      \\
      \card{q'} & = & \card q, \qquad\quad\ \, \text{if } \alpha =
      \emptystr \text{ or } \alpha \in \letters
      \cup \reg{q} \cup \{\gft{i} \st i \in \reg{q}\}
    \end{eqnarray*}
 \end{itemize}
 We let \cda\ denote the class of chronicle deallocating automata.
 A \cda\ is (i) \emph{deterministic} if, for each $q \in Q$,
 $\weight{\trans(q,\emptystr)} = 0$ and $\weight{\trans(q,\alpha)} = 1$, if
 $\alpha \in \qlab{q}$, (ii) a \emph{chronicle} automaton
 (CA$^\sharp$) if for all $q \in Q$ there is no $\pclose{i}$-transition  for
  $i \in \{1,\ldots,\card{q}-1\}$, (iii) a \emph{deallocating} automaton
 (DA$^\sharp$) if for all $q \in Q$ there is no $\gft{i}$-transition
 for $i \in \{1,\ldots,\card{q}-1\}$, and (iv) an automaton \emph{with
 freshness} (A$^\sharp$) if it is both CA$^\sharp$ and DA$^\sharp$.
\end{definition}

Configurations of \cda\ are defined in terms of \emph{chronicles} that 
keep track of names assigned to registers.
The chronicle $s_i$ of a register $i$ is a non-empty word on $\names$ together
with one of the names in $s_i$ that pinpoints the \emph{current value} of $i$,
denoted as $\hcv{s_i}$.
Let $s$ and $t$ be chronicles.
The \emph{extension} $s@t$ of $s$ with $t$ is the concatenation of the
words $s$ and $t$; while $s \setminus t$ is the word obtained by deleting from
$s$ the names in $t$.

We write $L = [e_1,\ldots,e_k]$ for a list of elements $e_1,\ldots,e_k$ (with
$[]$ being the empty list) and define $\el L i = e_i$.
An \emph{extant chronicle} is a (possibly empty) finite list $E =
[s_1,\ldots,s_k]$ of chronicles; we define $\hcv{E} =
[\hcv{s_1},\ldots,\hcv{s_k}]$,
which is always a list of pairwise
distinct names.
We extend $@$ and $\setminus$ to extant chronicles
element-wise: $E@s = [s_1@s,\ldots,s_k@s]$ and $E \setminus t
= [s_1\setminus t,\ldots,s_k\setminus t]$.
We may identify a list
with the underlying set of its elements (e.g.~writing $e \in L$ when there is
$i$ such that $e = \el L i$ and  $n \in s$ when
$n$ occurs in the chronicle $s$).
Also, for an extant chronicle $E$, $\el{\hcv{E}}i$ and $\el E i$
indicate the current value and the chronicle of a register $i$,
respectively.
Given two extant chronicles $E$ and $E'$, we let $E + E'$ be the list
obtained by appending $E'$ to $E$.

\begin{definition}
 \label{def:configuration}
 A \emph{configuration} of a \cda\ $\mathcal{H}=\tuple{Q,q_0,F,\trans}$ is a
 triple $\tuple{q,w,E}$ where $q \in Q$, $w$ is a word, and $E$ is an extant
 chronicle.
 A configuration $\tuple{q,w,E}$ is \emph{initial} if $q = q_0$ and $E = []$,
 and it is \emph{accepting} if $q \in F$, $w = \emptystr$ and $E = []$.
 Given two configurations $t = \tuple{q,w,E}$ and $t' = \tuple{q',w',E'}$,
 $\mathcal{H}$ \emph{moves from $t$ to $t'$} (written as $\step{\mathcal
 H}{t}{t'}$) if there is $\alpha \in \qlab{q}\cup\{\emptystr\}$ such that $q'
 \in \trans(q,\alpha)$ and
 \[
 \left\{
 \begin{array}{l@{\ }l}
  \alpha \in \reg{q}, & \text{if}\quad w = (\hcvel E \alpha)w' \quad \text{and}
   \quad E' = E
   \\
   \alpha \in \letters\cup\{\emptystr\}, & \text{if}\quad w = \alpha w'
   \quad\text{and}\quad E' = E
   \\
   \alpha = \star, & \text{if}\quad w = w', \quad n \in \names \setminus
   \hcv{E}, \quad E' = (E@n) + [n], \quad
   \\
  & \quad
   \hcvel{E'}{\card{q'}} = n \quad \text{and} \quad \forall i \in
   \reg{q}. \hcvel{E'} i = \hcvel E i
   \\
   \alpha = \gft{i}, & \text{if}\quad w = n w', \quad n \in \names \setminus
   \left(\hcv{E} \cup \el{E}{i}\right), \quad E' = E@n, \quad
   \\
  & \quad
   \hcvel{E'} i = n \quad \text{and} \quad \forall j \in
   \reg{q'}\setminus\{i\}.\hcvel{E'} j = \hcvel{E} j
   \\
   \alpha = \pclose{i}, & \text{if}\quad w = w', \quad E' + \el{E} {\reg{q}} =
   E, \quad \hcvel{E'} i = \hcvel{E} {\card{q}}
   \\
  &
   \qquad \qquad \qquad \quad \quad\text{and} \quad \forall j \in
   \reg{q'}\setminus\{i\}. \hcvel{E'} j = \hcvel{E} j
 \end{array}
 \right.
    \]
 The set $\rec_{\mathcal H}(t)$ of states reached by $\mathcal H$ from the
 configuration $t$ is given by
 \[
 \rec_{\mathcal H}(t) \mmdef
 \begin{cases}
  \{q\}
  & \text{if } t = \conf{q, \emptystr,E}
  \\
  \bigcup_{\step{\mathcal H}{t}{t'}}{\rec_{\mathcal H}(t')}
  & \text{otherwise}
 \end{cases}
 \]
 A \emph{run of $\mathcal H$ on a word $w$} is a sequence of moves of $\mathcal
 H$ from $\tuple{q_0, w,[]}$.
\end{definition}

\comment{
  \begin{wrapfigure}[10]{l}{5cm}\vspace*{-.5cm}
   \begin{minipage}{\linewidth}\centering
    \includegraphics[scale=.4]{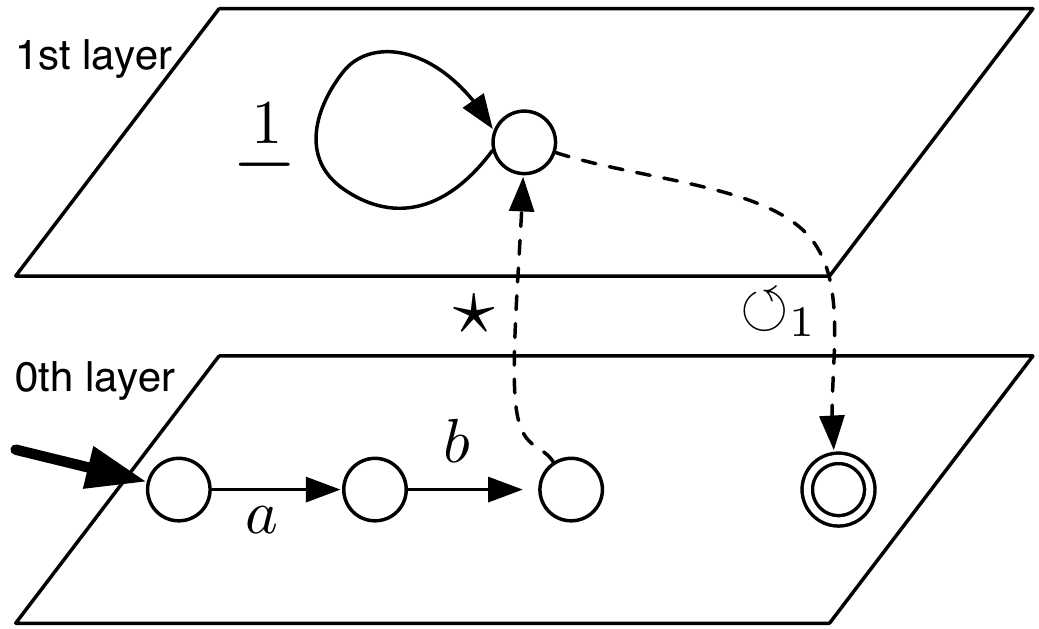}
   \end{minipage}
   \caption{CA$^\sharp$ for the language $\simples$.\label{fig:ca4lses}}
  \end{wrapfigure}
  }

  Intuitively, $\star$-transitions allocate a new register as the
  highest register of the target state (given $N \subseteq \names$ and $n \in \names$,
  we write $\nex n \lfresh N$,
  read `$\nex n$ is fresh for $N$', when $\nex n \not\in N$).
Such a register is initially assigned with a \emph{locally fresh} name $n$,
that is a name $n$ fresh for the current content of the registers;
accordingly the chronicle of the new register is created and initialised
(together with the updates of the other chronicles) to record the use of $n$.
Relative global freshness is implemented by $\gft{i}$-transitions that (as
for local freshness)  assign a name $n$ fresh with respect to the current
values of registers and (unlike in local freshness) also fresh with respect to
the $i$-th register's chronicle; contextually, $n$ is assigned to the $i$-th
register and all the other chronicles are updated to record the use of the new
name $n$.
Transitions labelled by $\pclose{i}$ permute the content, but not the 
chronicle, of the $i$-th register with the highest register of the current
state $q$ and dispose the chronicle of the $\reg q$-th register of $q$.

\begin{definition}
 Given a word $w$ on $\letters \cup \names$, a \cda\ $\mathcal{H}$
 \emph{accepts} (or \emph{recognises}) $w$ when $F \cap
 \rec_{\mathcal{H}}(\tuple{q_0,w,[]}) \neq \emptyset$. The set
 $\mathcal{L}_\mathcal{H}$ of words accepted by $\mathcal{H}$ is \emph{the
 language of $\mathcal{H}$}.
\end{definition}

\begin{example}
 \label{ex:automata:onethread}\it
 The automaton for the language $\simplet$ in \secref{sec:intro} is
 more complex than the one in \exref{ex:automata:simple}.
 Initially the automaton behaves as the one in \exref{ex:automata:simple}.
 After the $\cof{1}$-transition, it allocates the second and the third
 registers to consume   $p_0$ and $p'_0$.
  Note that $p_0$ and $p'_0$ are also recorded into the chronicle of the first
register;

\noindent
 \begin{minipage}{\linewidth}
  \begin{wrapfigure}[10]{r}[20pt]{.4\linewidth}\vspace*{-1.2cm}
   \resizebox{4cm}{!}{
   \includegraphics{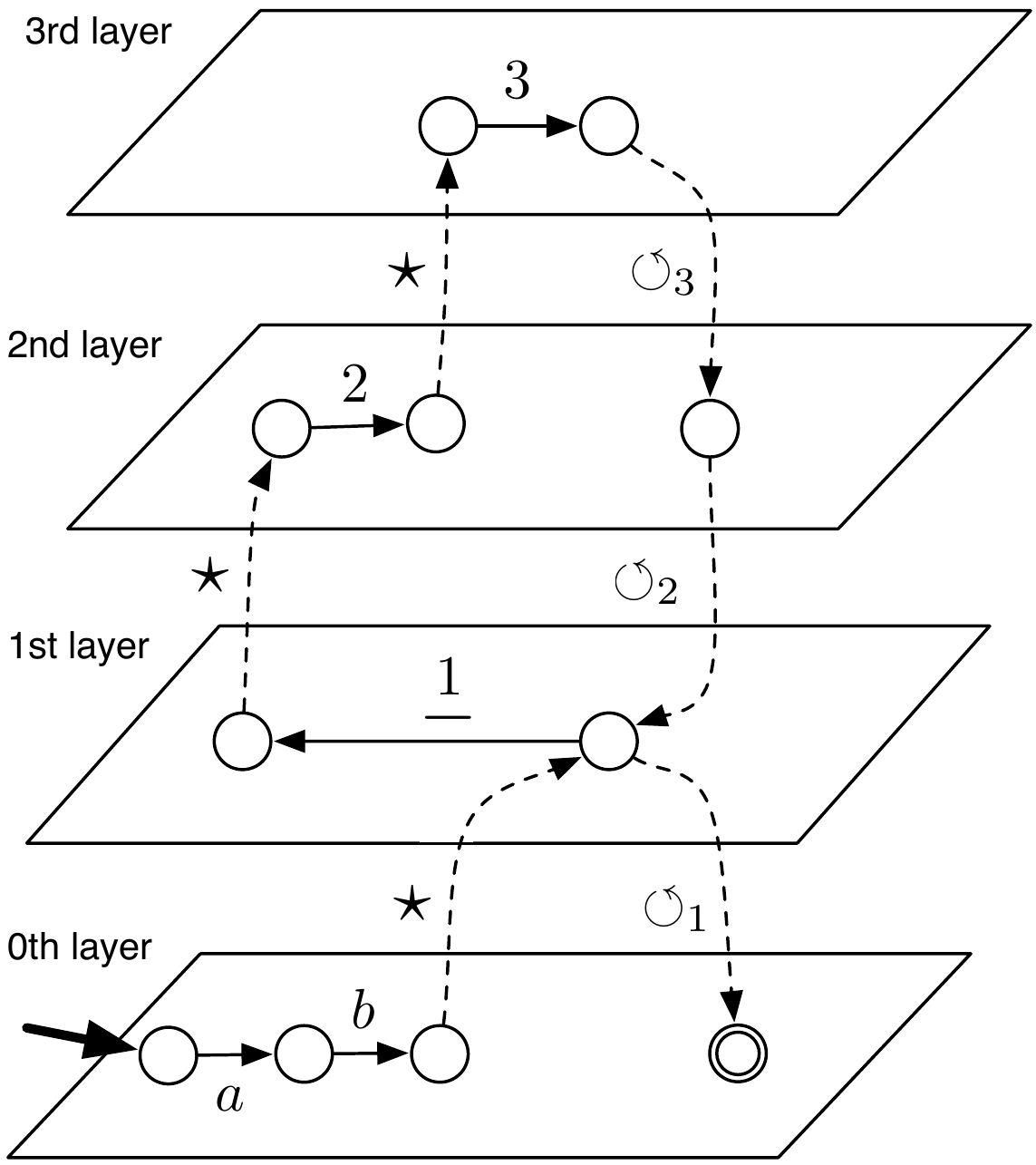}
   }
  \end{wrapfigure}
  hence, when the automaton comes back to the first layer by the
  $\pclose{3}$ and $\pclose{2}$ transitions, the next consumed name
  $r_1$ is guaranteed to be globally fresh (namely $r_1$ is different
  from $r_0$, $p_0$ and $p'_0$).
  On the other
 hand, the other registers cannot remember anything after they
  have been deallocated.
  So when the automaton goes up again to the second and the third layers by
  $\star$-transitions, it has two registers for \emph{locally} fresh names for
  $p_1$ and $p'_1$ with respect to the current session name $r_1$.  And, for
  instance, one of them can be $r_0$.
  
 \end{minipage}
\end{example}

 \begin{example}
  \label{ex:automata:morethreads}\it
  The automaton for the language $\moret$ in \secref{sec:intro} is the most
  complex and is given below.
  The automaton allocates a new register to consume

  \noindent
  \begin{minipage}{\linewidth}
   \begin{wrapfigure}[9]{r}[100pt]{.6\linewidth}\vspace*{-1.2cm}
    \resizebox{3.5cm}{!}{
    \includegraphics{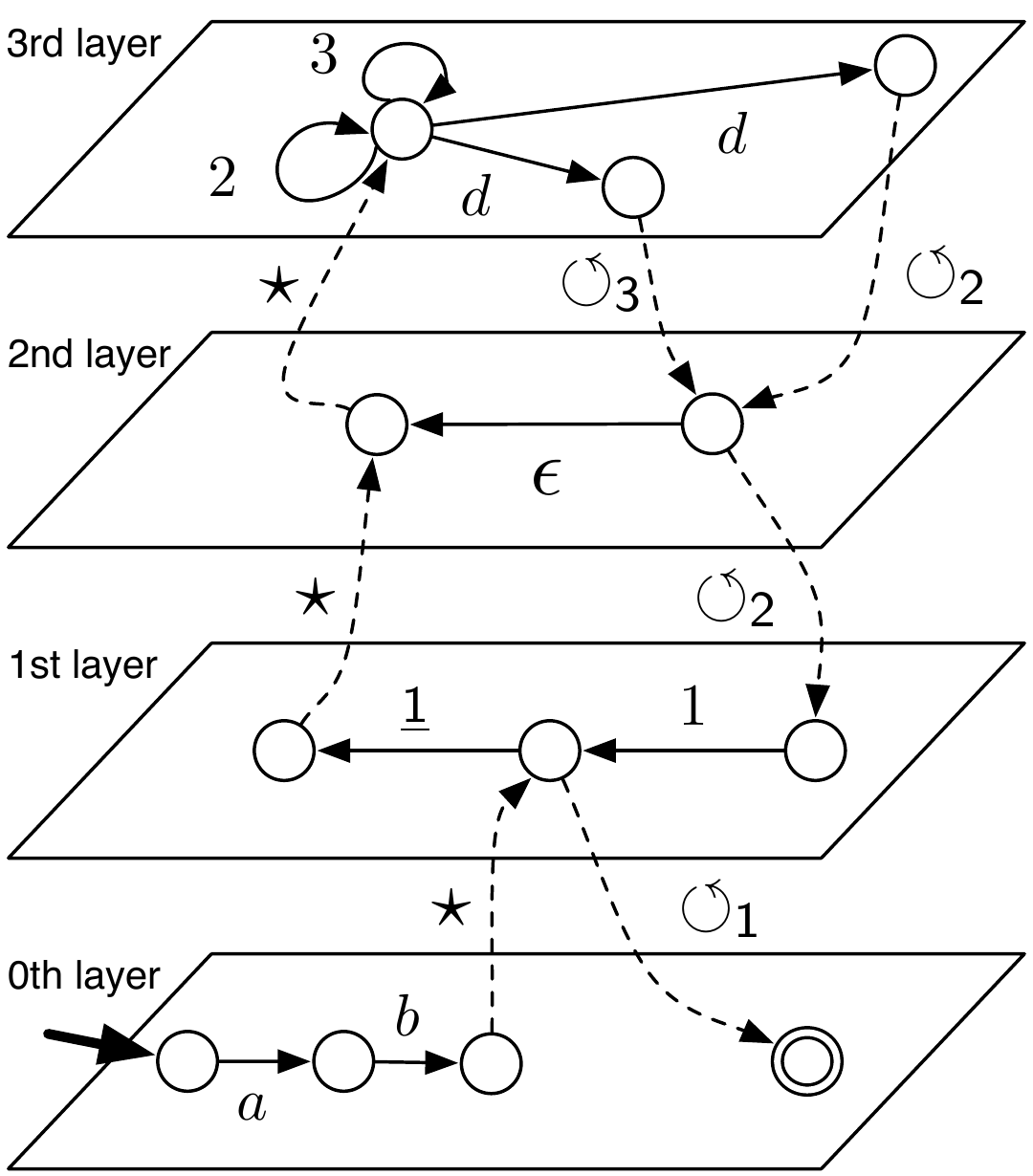}
    }
   \end{wrapfigure}
   a name $r_i$ for each run
   $i$ by means of the $\cof{1}$-transition on the first layer.
   From there the automaton has to accept the $w_i \in \alang{thr}(r_i)$.
   Similarly to the previous cases, two $\star$-transitions allocates the
   registers to accept two locally fresh names $p$ and $p'$.
   Now the automaton can loop in the left-most state of the third layer
   consuming the occurrences of $p$ and $p'$.  When $d$ appears on the input,
   the automaton non-deterministically decides which process to deallocate with
   transitions $\pclose{2}$ or $\pclose{3}$.
   Now the $\emptystr$-transition lets the automaton accept a new
   thread or deallocate the other register and start a new session
   after consuming the name of the current run stored in the first
   register.
   After delegating two threads, we finish the session $r_i$ by the
   $1$-transition on the first layer.
  \end{minipage}
 \end{example}

 By adding to the automaton of \exref{ex:automata:morethreads} an
 $\emptystr$-transition from the final state to the initial state, we
 can repeat the whole process by reusing also the names of previous
 runs (since the first chronicle is released when the automaton moves
 back to the lowest layer).


\newcommand{\phw}[2]{[\hspace{-2.5pt}(#1 \mid #2 )\hspace{-2.5pt}]}
\newcommand{\phl}{\davidsstar}

\section{Interpreting NREs}
\label{sec:exp2lang}

We formally define the language associated with a nominal regular expression.
Technically, we adapt the \emph{context} and the \emph{language calculus}~\cite{KurzST12tcs} for the new classes of NREs.

\newcommand*\circled[1]{\tikz[baseline=(char.base)]{
            \node[shape=circle,draw,inner sep=.5pt] (char) {#1};}}
\medskip\noindent\textbf{Schematic words.}
    To assign languages over infinite alphabets to NREs, it is natural
    to introduce the notion of \emph{schematic words}.
    We consider a countably infinite collection of \emph{placeholders $\phl_i$}
    A \emph{schematic word $\phw{\phl_1 \cdots \phl_k}{\phi}$}
    consists of a finite word of placeholders and a condition $\phi$ of
    the form
    \[
     \phi ::= \ \phl_i \not= \phl_j \st \phi , \phi \st \phi \lor \phi
    \]
    with the intention that $\phl_i \not= \phl_j$ means that $\phl_i$ and $\phl_j$ are not
    identical, `,' means `and' and $\vee$ means `or'.
    So, for example, a schematic word $\phw{\phl_1 \phl_2 \phl_3
    \phl_1}{\phl_1 \not= \phl_3}$ expresses the collection of words whose third
    letter is different from the first and the last letters (but the
    second letter can be any name): $\{abca \st a,b,c \in \names, a \not= c\}$.

    Languages over infinite alphabets recognised by ``nominal automata'' are
    typically closed under permutations, see e.g.~Proposition \ref{prop:alpha}.
    Schematic words describe such languages by means of inequations
    (freshness) of names.
    In Fig.~\ref{fig:languagecalculus}, we shall use schematic
    words which are extended to contain names from
    $\names$ (as well as placeholders).

\medskip\noindent\textbf{From maps to permutations.}
    Given a function $f$ with domain $\dom f$, the update $f_{\upd a b}$ has domain $\dom f
    \cup \{a\}$ with $f_{\upd a b}(a) = b$; $\bot$ is the empty map.
    We consider lists over $\names$ with no repeated elements (ranged over by
    $N$, $M$).
    Let $\length{N}$ be the length of $N$ and, for $\nex n \in \names$.
    The transposition of $\nex n$ and $\nex m$, denoted by $\trx m n$, is the
    bijection that swaps $\nex m$ and $\nex n$ and is the identity on any other
    names.
    Given two lists $N$ and $M$ of length $k$, let $N{\rhd}M$ be the map from
    $N$ to $M$ such that
    \[
    N{\rhd}M: \el N i \mapsto \el M i \qquad \text{for each } i \in
    \{1,\ldots,k\}
    \]
    which we extend it to a permutation 
    $\pi_{[N{\rhd}M]}$ on $\names$ that restricts to a bijection on $N\cup M$ and to the identity on $\names\setminus(N\cup M)$, see also
    \cite{KurzST12tcs}.
    In Fig.~\ref{fig:languagecalculus}, to transfer name (placeholder)
    information, we consider the above permutation for a list of current values
    $C = [n_1, \ldots, n_k]$ and an extant chronicle $E = [s_1, \ldots, s_k]$,
    with respect to the natural bijection: $n_i \mapsto \hcv{s_i}$ for each $i$
    (abusing notation, we write $\pi_{\left[{C}{\rhd}{E}\right]}$ for
    $\pi_{\left[{C}{\rhd}{\hcv{E}}\right]}$).
    Also, we may include some placeholders for $\pi_{\left[{C}{\rhd}{E}\right]}$.

\medskip\noindent\textbf{Permutations on expressions and extant chronicles.}
For an up-NRE $\ne$ and a bijection $\pi$ on $\names$ , the
permutation action of $\pi$ on $\ne$, denoted as $\pi\cdot\ne$, is
    \begin{enumerate}
     \item $\pi\cdot\nex{1} = \nex{1}$ \hfil $\pi\cdot\nex{0} =
       \nex{0}$ \hfil $\pi\cdot\nex{n} = \pi(\nex{n})$ \hfil $\pi
       \cdot \cof n = \cof{\pi(n)}$ \hfil $\pi\cdot\nex{s} =
       \nex{s}$ 
     \vspace{5pt}
     \item $\pi\cdot\left(\ne_1+\ne_2\right) = \left(\pi\cdot\ne_1\right) +
	   \left(\pi\cdot\ne_2\right)$
     \hfill $\pi\cdot\left(\ne_1\circ\ne_2\right) = \left(\pi\cdot\ne_1\right)
	   \circ \left(\pi\cdot\ne_2\right)$  \vspace{4pt}
     \item $\pi\cdot\left(\ne^\ast\right) = \left(\pi\cdot\ne\right)^\ast$
     \hfill $\pi\cdot\left(\pexp{n}{\ne}{m}\right) =
	   \pexp{\pi(n)}{\left(\pi\cdot\ne\right)}{\pi(m)}$
    \end{enumerate}
    The permutation action of $\pi$ on a chronicle $s_i= n_{i_1}\ldots n_{i_k}$,
    denoted as $\pi\cdot s_i$, is $t =
    \pi\left(n_{i_1}\right)\ldots\pi\left(n_{i_k}\right)$ with $\hcv t =
    \pi\cdot\hcv{s_i}$.
    Finally, the permutation action of $\pi$ on an extant chronicle $E = [s_1,
    \ldots, s_k]$ is $\pi \cdot E = [\pi\cdot s_1,\ldots,\pi\cdot s_k]$.
    Note that we may include placeholders in contexts in
    Fig.~\ref{fig:languagecalculus}.

    \label{subsec:langcomp}
    \renewcommand{\lfresh}{\,\mbox{\bf\#}\,}
    \renewcommand{\gfresh}{\,\mbox{\bf\underline{\#}}\,}
    \newcommand{\ctxc}{CTXC} \newcommand{\lngc}{LNGC}

    \medskip\noindent\textbf{Contextualised expressions} are triples
    $\etw{C}{\ne}{E}$ where $\ne$ is a nominal regular expression, $C$
    is a finite list of pairwise distinct (including indices) names
    and placeholders, i.e.~$C \in \left(\names \cup \{\phl, \phl_1,
      \ldots\}\right)^\ast$ (called \emph{pre-context}) and $E$ is an
    extant chronicle (\emph{post-context}).
To compute languages
from NREs, we may
    include placeholders in (extant) chronicles.
Placeholders appear in contexts in the language calculus
    Fig.~\ref{fig:languagecalculus} only to abstract some names.
    Intuitively, $C$ is the list of the names and placeholders ``used before''
    $\ne$ and $E$ is the extant chronicle ``established after'' $\ne$.
    More precisely, the post-context $E$ possesses two important data:
    for \ctxc\ and the construction of corresponding automata (on the
    inductive step for $\pexp{n}{\ne}{m}$), it tells ``which registers
    are permuted when brackets are closes'' and, for
    \lngc, it reserves numbers of relative-global fresh names for each
    register, which is necessary to consider ($\lcirc$) in
    Fig.~\ref{fig:languagecalculus}.
    It is useful to explicitly express the current values of (extant) chronicles
    and write $\ihsv{s}{n}$ for the chronicle $s$ with $\hcv{s} = n$ and
    $\left[\ihsv{s_1}{n_1},\ldots,\ihsv{s_k}{n_k}\right]$ for an extant
    chronicle $E = [s_1,\ldots,s_k]$ with $\hcv E = \left[n_1,\ldots,n_k\right]
    = \left[\hcv{s_1},\ldots,\hcv{s_k}\right]$.

\smallskip
\noindent\textbf{The context calculus \ctxc} is defined in \figref{fig:contextcalculus}, where we assume that $\star \# C$,
    \begin{figure}[!t]
     \begin{center}
      \AxiomC{$\etw{C}{\ne_1+\ne_2}{E}$}
      \RightLabel{($\cplus_1$)}
      \UnaryInfC{$\etw{C}{\ne_1}{E}$}
      \DisplayProof
      \hfil
      \AxiomC{$\etw{C}{\ne_1+\ne_2}{E}$}
      \RightLabel{($\cplus_2$)}
      \UnaryInfC{$\etw{C}{\ne_2}{E}$}
      \DisplayProof
      \\[7pt]
      \rootAtTop
      \AxiomC{$\etw{C}{\ne_1}{\nec{C}}$}
      \AxiomC{$\etw{C}{\ne_2}{E}$}
      \RightLabel{($\ccirc$)}
      \BinaryInfC{$\etw{C}{\ne_1\circ\ne_2}{E}$}
      \DisplayProof
      \hfil
      \AxiomC{$\etw{C}{\ne^\ast}{E}$}
      \RightLabel{($\cks$)}
      \UnaryInfC{$\etw{C}{\underbrace{\ne\circ\cdots\circ\ne}_{h \text{
      times}}}{E}$}
      \DisplayProof
      \\[5pt]
      \AxiomC{$\etw{C}{\pexp{n}{\ne}{m}}{E} \qquad m = n \qquad C' = C+[\star]
      \qquad E' = (E@\star) + [\star]$}
      \RightLabel{($\cdmd_{=}$)}
      \UnaryInfC{$\etw{C'}{\left(n\ \star\right)\cdot\ne}{\ihsv{E'}{\hcv{E'}}}$}
      \DisplayProof
      \\[7pt]
      \AxiomC{$\etw{C}{\pexp{n}{\ne}{m}}{E} \quad m \not= n \quad C' = C+
      [\star] \quad E' = (E@\star) + [\star m] \quad \hcv{E'} =
      \hcv{E}+[\star]$}
      \RightLabel{($\cdmd_{\not=}$)}
      \UnaryInfC{$\etw{C'}{\left(n\
      \star\right)\cdot\ne}{\ihsv{E'}{\left(\left(m\
      \star\right)\cdot\hcv{E'}\right)}}$}
      \DisplayProof
     \end{center}
     \caption{CTXC: Rules for contextualised expressions\label{fig:contextcalculus}}
    \end{figure}
    $C = [n_1,\ldots,n_k]$ is a list of pairwise distinct names, and $\nec{C} =
    [\ihsv{s_1}{n_1},\ldots,\ihsv{s_k}{n_k}]$ is the extant chronicle where for
    each $\inleq 1 i k$, $s_i = n_i \cdots n_k$.
    As in \cite{KurzST12tcs}, the rules in
    \figref{fig:contextcalculus} propagate pre- and post-contexts to
    all subexpressions of $\ne$ with a top-down visit of the abstract
    syntax tree of $\ne$.
    The rule $(\cks)$ unfolds the Kleene star an arbitrary but bound
    number of times $h$; later (c.f. \lngc\ in ~\figref{fig:languagecalculus})
    we will take the union of the languages computed for each unfolding.
    By rules ($\cdmd_=$) and ($\cdmd_{\neq}$) it is clear that in pre-contexts
    it is necessary just to record the names already used from the root to the
    current node of the tree.
    Instead, in the post-context it is necessary to keep track of the ``names
    created'' in the current subexpression for relative global freshness.
    We will see that this is crucial for computing the local freshness and the
    concatenation of the languages of NREs.
    Notice that in rule ($\cdmd_{\neq}$), for an extant chronicle
    $E=[s_1,\ldots,s_k]$, the notation $\ihsv E {\hcv E}$ abbreviates
    $[\ihsv{\hcv{s_1}}{s_1},\ldots,\ihsv{\hcv{s_k}}{s_k}]$.

    \begin{remark}
      For NREs without permutations, we do not need to consider the rule
      ($\cdmd_{\neq}$) in Fig.~\ref{fig:contextcalculus}. For NREs
      without underlines, only current values of registers matter.
    \end{remark}
    Note how ($\cdmd_{\neq}$) deals with permutations: the expression
    $\ne$ in the binder is contextualised by a local renaming of the
    (content corresponding to) $n$ with $\star$ which ``after'' $\ne$
    (namely in the post-context) is also replaced for $m$ in the
    current values, while $m$ is added to the chronicle corresponding
    to the new register.
    \begin{example}\it
     \label{ex:ctx}
     For the up-NRE $\npexp{n}{\cof{n}\npexp{m}{\pexp{l}{m}{m}}\cof{n}}$,
     contexts are computed as
     {\small
     \[
     \begin{minipage}{\linewidth}\centering
      \rootAtTop
      \AxiomC{$\etw{[a]}{\cof{a}}{[\ihsv{a}{a}]}$}
      \AxiomC{$\etw{[a, b, c]}{b}{[\ihsv{abc}{a},\ihsv{bc}{c},\ihsv{cb}{b}]}$}
      \RightLabel{($\cdmd_=$)}
      \UnaryInfC{$\etw{[a, b]}{\pexp{l}{b}{b}}{[\ihsv{ab}{a},\ihsv{b}{b}]}$}
      \RightLabel{($\cdmd_=$)}
      \UnaryInfC{$\etw{[a]}{\npexp{m}{\pexp{l}{m}{m}}}{[\ihsv{a}{a}]}$}
      \AxiomC{$\etw{[a]}{\cof{a}}{[\ihsv{a}{a}]}$}
      \RightLabel{($\ccirc$)}
      \BinaryInfC{$\etw{[a]}{\npexp{m}{\pexp{l}{m}{m}}\cof{a}}{[\ihsv{a}{a}]}$}
      \RightLabel{($\ccirc$)}
      \BinaryInfC{$\etw{[a]}{\cof{a}\npexp{m}{\pexp{l}{m}{m}}\cof{a}}{[\ihsv{a}{a}]}$}
      \RightLabel{($\cdmd_{\not=}$)}
      \UnaryInfC{$\etw{[]}{\npexp{n}{\cof{n}\npexp{m}{\pexp{l}{m}{m}}\cof{n}}}{[]}$}
      \DisplayProof
     \end{minipage}
     \]
     }
     Note that the up-NRE is quite similar to a simple trace of $\moret$ in \secref{sec:intro}.
    \end{example}

    \smallskip\noindent\textbf{The language calculus} (\lngc\ for short) is given in
    Fig.~\ref{fig:languagecalculus} (the new notations are explained below in
    the comment of the rules).
    \begin{figure}[!t]
     \begin{center}
      \AxiomC{$\etw{C}{\nex{1}}{E}$}
      \RightLabel{($\nex{1}$)}
      \UnaryInfC{$\etw{C}{\phw{\emptystr}{}}{E}$}
      \DisplayProof \hfil
      \AxiomC{$\etw{C}{\nex{0}}{E}$}
      \RightLabel{($\nex{0}$)}
      \UnaryInfC{$\etw{C}{\phw{\emptyset}{\bot}}{E}$}
      \DisplayProof
      \hfil
      \AxiomC{$\etw{C}{\nex{s}}{E}$}
      \RightLabel{($\nex{s}$)}
      \UnaryInfC{$\etw{C}{\phw{s}{}}{E}$}
      \DisplayProof
      \hfil
      \AxiomC{$\etw{C}{\nex{n}}{E}$}
      \RightLabel{($\nex{n}$)}
      \UnaryInfC{$\etw{C}{\phw{n}{}}{E}$}
      \DisplayProof
      \\[7pt]
      \AxiomC{$\etw{C}{\cof{\nex{n}}}{E}$}
      \AxiomC{$\el{C}{i} = n$}
      \RightLabel{($\cof{\nex{n}}$)}
      \BinaryInfC{$\etw{C}{\phw{\phl}{\phl \lfresh C, \phl \gfresh^i
      \el{\nec{C}}{i}}}{\ihsv{\left(E@\phl\right)}{\left(\left(n\
      \phl\right)\cdot\hcv{E}\right)}}$}
      \DisplayProof
      \\[7pt]
      \AxiomC{$\etw{C}{\phw{\bullet_1^1 \cdots \bullet^1_{k_1}}{\phi_1}}{E_1}$}
      \AxiomC{$\etw{C}{\phw{\bullet^2_1 \cdots \bullet^2_{k_2}}{\phi_2}}{E_2}$}
      \RightLabel{($\lcirc$)}
      \BinaryInfC{$\etw{C}{\phw{\bullet^1_1 \cdots \bullet^1_{k_1} @
      \pi_{\left[{C}{\rhd}{E_1}\right]} \cdot \left(\bullet^2_1 \cdots
      \bullet^2_{k_2}\right)}{\phi_1, \left(\pi_{\left[{C}{\rhd}{E_1}\right]}
      \cdot \phi_2 \right)}}{E_1 @ \left(\pi_{\left[{C}{\rhd}{E_1}\right]} \cdot
      E_2 \right)}$}
      \DisplayProof
      \\[7pt]
      \AxiomC{$\etw{C + [n]}{\phw{\bullet_1 \cdots \bullet_k}{\phi}}{E +
      [\ihsv{t}{m}]}$}
      \RightLabel{($\ldmd$)}
      \UnaryInfC{$\etw{C}{\phw{\left(n\ \phl\right)\cdot \bullet_1 \cdots
      \bullet_k}{\phl \lfresh C, \left(\left(n\ \phl\right) \cdot
      \phi\right)}}{\left(n\ \phl\right) \cdot E}$}
      \DisplayProof
     \end{center}
     \caption{LNGC: Rules for computing schematic words\label{fig:languagecalculus}}
    \end{figure}
    Given an NRE $\ne$, the rules in \figref{fig:languagecalculus} are meant to
    be applied ``bottom up'' to the proof trees computed by the \ctxc\ starting
    from the contextualised expressions $\etw{[]} \ne {[]}$.
    Also, in each instance of the rules $\ldmd$ and $\cof{\nex n}$
    a completely fresh placeholder has to be introduced.
    Therefore, in the \lngc, we do not have rules for $\_^\ast$ nor $+$.
    Instead, we take unions for $\_^\ast$ and $+$ after computing languages of
    each tree.
    On each application of the rules ($\cof{n}$) and ($\ldmd$), we
    must use a new $\phl$, see eg the occurrence of $\phl_1$ and
    $\phl_4$ in Ex.\ \ref{ex:langex}.
    We use $\bullet$'s to range over $\names \cup \{\phl, \phl_1,
    \ldots \}$.

\begin{remark}
  We can use the same pre-context $C$ in both premises of rule
  ($\lcirc$) in \figref{fig:languagecalculus} because \ctxc\
  duplicates the pre-contexts when decomposing the expression while
  \lngc\ recovers the same pre-contexts when visiting the tree.
\end{remark}

\medskip\noindent Rules ($\nex{1}$), ($\nex{0}$), ($\nex{n}$) and ($\nex{s}$) yield the natural
 interpretation for the corresponding elementary expressions.
 Note that $\phw{\emptyset}{\bot}$ in rule ($\nex{0}$) generates nothing by concatenations of
 schematic words.
 The remaining rules are more delicate since our NREs encompass both freshness
 ``with respect to current values'' (local freshness) and freshness ``with
 respect to chronicles'' (relative global freshness).
 Therefore, it is important to identify which names have to be locally and which
 relative globally fresh.
 Formally, this is done by noticing that each language obtained by \lngc\ from
 $+$-free NREs can be expressed as a finite conjunction of freshness conditions
 that we write as $\phl_1 \lfresh S_1 \ \mathsf{and} \cdots \mathsf{and}\
 \phl_r \gfresh^i S'_r$ where $S_j, S'_{j'}$ are lists of names in $C$ or some
 placeholders $\phl$ (note that $\lfresh$ is not the $\#$ operation, it is just
 a syntactic device to mark the type of freshness required on placeholders; similarly,
 $\gfresh^i$ represent marks that relative global freshness with respect to the
 chronicles of the $i$-th regist is required).
 ).
 This presentation of a language can be obtained by inspecting the corresponding
 NRE and noting the conditions on names that occur in the scopes of binders.

 Rule ($\cof{\nex n}$) is for relative global freshness; $\phl \lfresh C$ means
 that $\phl$ is locally fresh for current values $C$ and $\phl \gfresh^i
 \el{\nec{C}}{i}$ for $\el{C}{i} = n$ does that $\phl$ is fresh with respect to
 the $i$-th chronicle.

 Rule ($\lcirc$) deals with the concatenation of two languages by attaching each
 schematic word $v$ of second language to each schematic word $w$ of the
 first; note that, since permutations and underlines may change the
 post-contexts, it is necessary to use a permutation
 $\pi_{\left[{C}{\rhd}{E_1}\right]}$ to rename everything and update global
 fresh information before concatenating schematic words.
 Given two schematic words $\phw{\bullet^1_1 \cdots \bullet^1_{k_1}}{\phi_1}$
 and $\phw{\bullet^2_1 \cdots \bullet^2_{k_2}}{\phi_2}$, we append the first
 schematic word with the second schematic word permuted by
 $\pi_{\left[{C}{\rhd}{E_1}\right]}$.
 In addition, we also update the freshness condition $\phi_2$.
 There are two types of ``updates'' in $\pi_{\left[{C}{\rhd}{E_1}\right]} \cdot
 \phi_2$: update for local freshness $\lfresh$ and relative global freshness
 $\gfresh^i$: see rules ($\cof{n}$) and ($\ldmd$).
 For the local freshness $\lfresh$ in $\phi_2$, say $\bullet_j \lfresh
 \bullet_{j_1} \cdots \bullet_{j_l}$, we just replace this condition as
 $\pi_{\left[{C}{\rhd}{E_1}\right]} \cdot \bullet_j \lfresh
 \pi_{\left[{C}{\rhd}{E_1}\right]} \cdot \left(\bullet_{j_1} \cdots
 \bullet_{j_l} \right)$ in $\pi_{\left[{C}{\rhd}{E_1}\right]} \cdot \phi_2$.
 And for the global freshness with respect to the register $i$ $\gfresh^i$ in
 $\phi_2$, say $\bullet_j \gfresh^i \bullet_{j_1} \cdots \bullet_{j_l}$, we
 replace this condition ``with appending'' the register $i$'s chronicle
 freshness information as $\pi_{\left[{C}{\rhd}{E_1}\right]} \cdot \bullet_j
 \gfresh^i \el{E_1}{i} @ \left(\pi_{\left[{C}{\rhd}{E_1}\right]} \cdot
 \left(\bullet_{j_1} \cdots \bullet_{j_l} \right)\right)$.
 Therefore, for $\gfresh^i$, we not only permute the freshness conditions but
 also update the corresponding previous chronicle information in $E_1$ in
 $\pi_{\left[{C}{\rhd}{E_1}\right]} \cdot \phi_2$.
 We also note that, for $E_1 @ \left( \pi_{\left[{C}{\rhd}{E_1}\right]} \cdot
 E_2 \right)$ in the same rule, we keep the current values of the whole extant
 chronicle $\hcv{\pi_{\left[{C}{\rhd}{E_1}\right]} \cdot E_2}$,
 i.e.~$\pi_{\left[{C}{\rhd}{E_1}\right]} \cdot \hcv{E_2}$.

 Rule ($\ldmd$) deallocates the last name $n$ in the pre-context and the last
 chronicle $t$ in the post-context, i.e~$\ihsv{t}{m}$.
 Accordingly, we abstract the name $n$ to a placeholder $\phl$.

\exref{ex:langex} below shows an application of the rules; the reader
is referred to Appendix \ref{sec:complex} for an example of a more
complex language.
 \begin{example}\it
  \label{ex:langex}
  To show the difference between local and global freshness we compute
  the language considered in \exref{ex:ctx}:
  {\small\[
  \begin{minipage}{\linewidth}\centering
   \AxiomC{$\etw{[a]}{\cof{a}}{[\ihsv{a}{a}]}$}
   \RightLabel{($\cof{a}$)}
   \UnaryInfC{$\etw{[{a}]}{\phw{\phl_1}{\phl_1 \lgfresh
   a}}{[\ihsv{a\phl_1}{\phl_1}]}$}
   \dashedLine
   \UnaryInfC{$\mathfrak{L}_l$}
   \DisplayProof
   \hfil
   \AxiomC{$\etw{[a]}{\cof{a}}{[\ihsv{a}{a}]}$}
   \RightLabel{($\cof{a}$)}
   \UnaryInfC{$\etw{[a]}{\phw{\phl_4}{\phl_4 \lgfresh
   a}}{[\ihsv{a\phl_4}{\phl_4}]}$}
   \dashedLine
   \UnaryInfC{$\mathfrak{L}_r$}
   \DisplayProof

   \bigskip

   \AxiomC{$\mathfrak{L}_l$}
   \AxiomC{$\etw{[a, b, c]}{b}{[\ihsv{abc}{a},\ihsv{bc}{c},\ihsv{cb}{b}]}$}
   \RightLabel{($b$)}
   \UnaryInfC{$\etw{[a, b, c]}{\phw{b}{}}{[\ihsv{abc}{a},\ihsv{bc}{c},\ihsv{cb}{b}]}$}
   \RightLabel{($\ldmd$)}
   \UnaryInfC{$\etw{[a, b]}{\phw{b}{\phl_3 \lfresh ab}}{[\ihsv{ab
   \phl_3}{a},\ihsv{b \phl_3}{\phl_3}]}$}
   \RightLabel{($\ldmd$)}
   \UnaryInfC{$\etw{[a]}{\phw{\phl_2}{\phl_2 \lfresh a, \phl_3 \lfresh a
   \phl_2}}{[\ihsv{a \phl_2 \phl_3}{a}]}$}
   \AxiomC{$\mathfrak{L}_r$}
   \RightLabel{($\ccirc$)}
   \BinaryInfC{$\etw{[a]}{\phw{\phl_2 \phl_4}{\phl_2 \lfresh a, \phl_3 \lfresh a
   \phl_2, \phl_4 \lgfresh a \phl_2 \phl_3}}{[\ihsv{a \phl_2 \phl_3 a
   \phl_4}{\phl_4}]}$}
   \RightLabel{($\ccirc$)}
   \BinaryInfC{$\etw{[a]}{\phw{\phl_1 \phl_2 \phl_4}{\phl_1 \lgfresh a, \phl_2
   \lfresh \phl_1, \phl_3 \lfresh \phl_1 \phl_2, \phl_4 \lgfresh a \phl_1 \phl_2
   \phl_3}}{[\ihsv{a \phl_1 \phl_1 \phl_2 \phl_3 \phl_1 \phl_4}{\phl_4}]}$}
   \RightLabel{($\cdmd$)}
   \UnaryInfC{$\etw{[]}{\phw{\phl_1 \phl_2 \phl_4}{\phl_1 \lgfresh \phl, \phl_2
   \lfresh \phl_1, \phl_3 \lfresh \phl_1 \phl_2, \phl_4 \lgfresh \phl \phl_1
   \phl_2 \phl_3}}{[]}$}
   \DisplayProof
  \end{minipage}
	 \]}
 where, for compactness, conditions of the form $\phl \lfresh S \ \mathsf{and}\
 \phl \gfresh^1 S'$ are abbreviated as $\phl \lgfresh S'$ provided that $S
 \subset S'$).
 Note how $\gfresh^1$ (and $\lgfresh$) differs from $\lfresh$.
 Since the expression's depth is at most three, only by local freshness
 $\lfresh$, we cannot encounter freshness conditions with respect to more than
 four other placeholders, as $\phl_4 \lgfresh \phl \phl_1 \phl_2 \phl_3$.
 \end{example}

The language of an up-NRE $\ne$ with no free names is obtained by three steps:
1. compute schematic words with \lngc\ on all the proof-trees generated by
\ctxc\ starting with $\etw{[]}{\ne}{[]}$, 2. interpret all schematic words
naturally into languages over infinite alphabets, and 3. take the union of all
the languages.
We denote the language obtained from $\ne$ by $\langof(\ne)$.
 \begin{definition}
  A language over infinite alphabets is \emph{nominal regular} if there is an
  up-NREs $\ne$ such that the language is $\langof(\ne)$.
 \end{definition}
\begin{proposition}
 \label{prop:alpha}
 All nominal regular expressions are closed under $\alpha$-equivalence.
\end{proposition}


\renewcommand{\lfresh}{\,\#\,}
\renewcommand{\gfresh}{\,\underline{\#}\,}

\section{Kleene theorems}
\label{sec:kleene}
We now give our main results.
\begin{theorem}
 \label{thm:nre2noma}
 Nominal regular languages are accepted by \cda\ as follows:
 
 \begin{tabular}{ll}
  \begin{minipage}[l]{.43\linewidth}\hspace{-.53cm}
   $
   \resizebox{6cm}{!}{
   \xymatrix{
   {} & \text{up-NREs} \leftrightarrow \text{CDA$^\sharp$} & {}\\
   \text{u-NREs} \leftrightarrow \text{CA$^\sharp$} \ar[ru]^{\text{permutation}}
   & {} & \text{p-NREs}\leftrightarrow \text{DA$^\sharp$}
   \ar[lu]_{\text{underline}}\\
   {} & \text{b-NREs}\leftrightarrow \text{A$^\sharp$}
   \ar[lu]^{\text{underline}} \ar[ru]_{\text{permutation}} & {} }
   }
   $
  \end{minipage}
	  &
  \begin{minipage}[l]{.5\linewidth}\scriptsize
   \begin{itemize}
    \item languages described by up-NREs are accepted by CDA$^\sharp$
    \item languages described by u-NREs are accepted by CA$^\sharp$
    \item languages described by p-NREs are accepted by DA$^\sharp$
    \item languages described by b-NREs are accepted by A$^\sharp$.
   \end{itemize}
  \end{minipage}
 \end{tabular}
\end{theorem}

For up-NREs, we inductively construct the corresponding $\cda$.
To do so, we extend the notion of $\cda$
to $\cda$ \emph{in-contexts},
$\etw{C}{\hdof{\ne}}{E}$, languages to languages \emph{in-contexts},
$\etw{C}{\langof{\left(\ne\right)}}{E}$, and up-NREs to up-NREs
\emph{in-contexts}, $\etw{C}{\ne}{E}$.
For automata-in-contexts $\etw{C}{\hdof{\ne}}{E}$ with $\hdof{\ne} =
\tuple{Q,q_0,\trans,F}$, we let
 \begin{itemize}
  \item $\card{q} \geq \length{C}$ for each $q \in Q$, especially, $\card{q_0} =
	\length{C}$ and $\card{q} = \length{C}$ for $q \in F$,
  \item the initial configuration is $\tuple{q_0,w,\nec{C}}$ and final
	configurations $\tuple{q,\emptystr,E'}$ for some $q \in F$ and some
	extant chronicle $E'$ ($\length{E'} = \length{C}$),
  \item for each free name $n$ in $\ne$, there is a unique index $i \in
	\{1,\ldots,\length{C}\}$ with $\el{{C}}{i} = n$.
 \end{itemize}

\subsection{From NRE to $\cda$}

We show the inductive construction of \cda\ for up-NREs; the
construction is similar to the one in \cite{KurzST12fossacs}.
The inductive construction is informally depicted in Fig.~\ref{fig:indsteps}
where the CDA$^\sharp$
\begin{equation}
 \label{eq:indaut}
\hdof{\ne} = \tuple{Q,q_0,\trans,F}
\qquad\text{and}\qquad
\hdof{\ne_h} = \tuple{Q_h,q_{(h,0)},\trans_h, F_h},
\ \ h=1,2
\end{equation}
in the inductive steps are instances of the generic automaton of Fig.~\ref{fig:indhdof} (which, for
readability, encompasses only two final states, but in general could have more)
\begin{figure}[!t]
 \begin{minipage}{\linewidth}\centering
  \includegraphics[scale=.4]{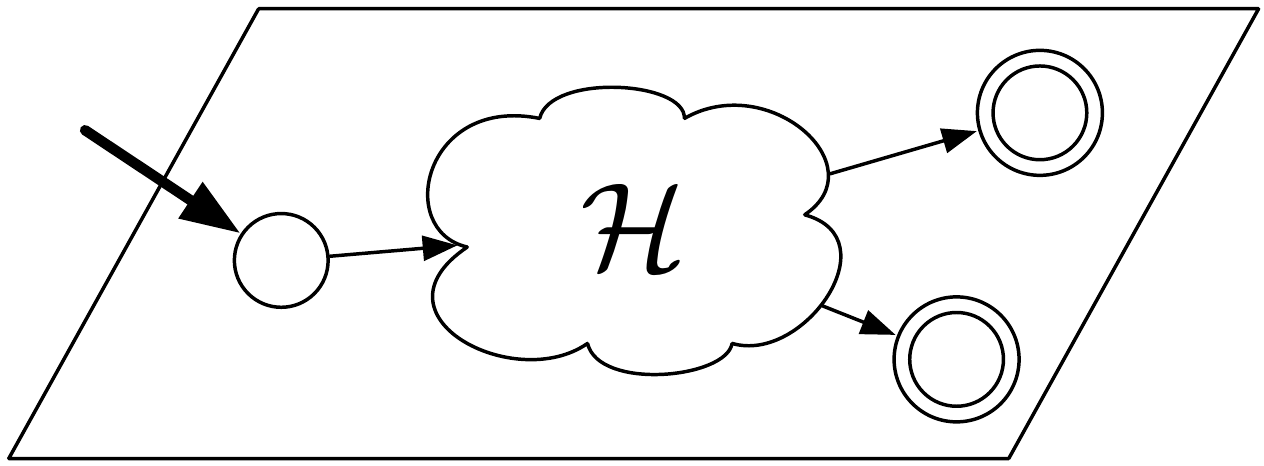}
 \end{minipage}
 \caption{The CDA$^\sharp$ $\hdof{\ne}$.\label{fig:indhdof}}
\end{figure}
and respectively correspond to the following up-NREs in-contexts
\[
\etw{C}{\ne}{E}
\qquad
\etw{C}{\ne_1}{E}
\qquad
\etw{C}{\ne_2}{E}
\]
In the last case of \figref{fig:indhdof} for the NRE
$\pexp{n}{\ne}{m}$, we let the contexts for $\ne$ be $C+[n]$ and
$E+[\ihsv{t}{m}]$,
i.e.~$\etw{C+[n]}{\hdof{\ne}}{E+\left[\ihsv{t}{m}\right]}$.
Note also that the automata in~\eqref{eq:indaut}
used in the inductive cases may generate states at higher levels, see
the last case of \figref{fig:indsteps}.

Notice that, for simplicity, we assume that NREs have no free
names.
Hence, $n$ or $\cof{n}$ are local names which must be stored in a
unique register $i$ and pre-contexts $C$.
In addition, post-contexts $E$ do not coincide with the ``real''
post-contexts when each word is accepted.
This is because, for example, if a NRE has a Kleene star, the real
post-contexts may change from time to time depending on how many times
we make loops to accept words.

\noindent\textbf{Base cases.}
    If $\ne = \nex{1}$, we let the corresponding CDA$^\sharp$ in-contexts
    $\etw{C}{\hdof{\nex{1}}}{E}$ with $\hdof{\nex{1}} = \tuple{Q,q_0,\trans,F}$
    where $Q = \{q_0\}$, $\trans = \emptyset$ and $F = \{q_0\}$.
    Note that the number of registers in $q_0$ is determined by $\length{C}$,
    i.e.~$\card{q_0} = \length{C}$.
    \begin{figure}[htbp]
     \begin{minipage}{\linewidth}\centering
      \includegraphics[scale=.5]{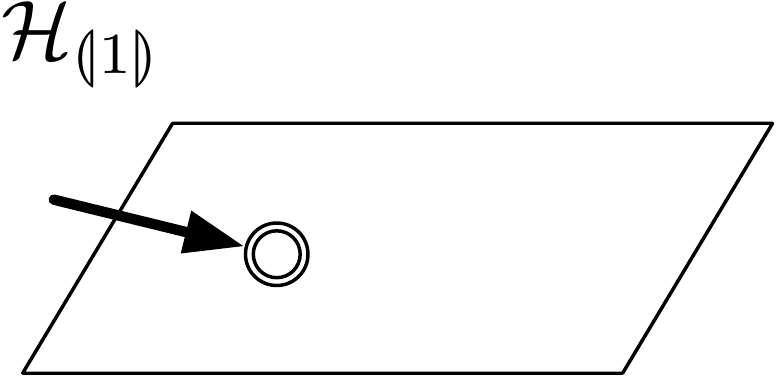}
     \end{minipage}
    \end{figure}

    If $\ne = \nex{0}$, we let the corresponding CDA$^\sharp$ in-contexts
    $\etw{C}{\hdof{\nex{0}}}{E}$ with $\hdof{\nex{0}} = \tuple{Q,q_0,\trans,F}$
    where $Q = \{q_0\}$, $\trans = \emptyset$ and $F = \emptyset$.
    Note that the number of registers in $q_0$ is determined by $\length{C}$.
    \begin{center}
     \begin{minipage}{\linewidth}\centering
      \includegraphics[scale=.5]{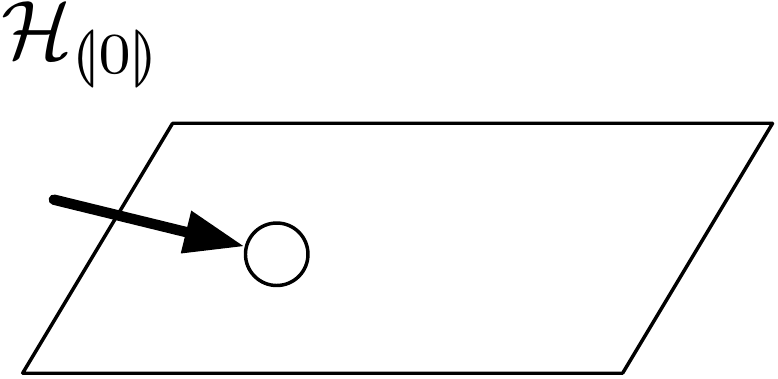}
     \end{minipage}
    \end{center}

    If $\ne = \nex{s}$, we let the corresponding CDA$^\sharp$ in-contexts
    $\etw{C}{\hdof{\nex{s}}}{E}$ with $\hdof{\nex{s}} = \tuple{Q,q_0,\trans,F}$
    where $Q = \{q_0,q_1\}$, $F = \{q_1\}$ and
    \[
    \begin{cases}
     \trans(q,\alpha) = \{q_1\} & \text{if } q = q_0 \text{ and } \alpha = s\\
     \trans(q,\alpha) = \emptyset & \text{otherwise}
    \end{cases}
    \]
    Note that the number of registers in $q_0$ and $q_1$ is the same as
    $\length{C}$, i.e.~$\card{q_0} = \card{q_1} = \length{C}$.
    \begin{center}
     \begin{minipage}{\linewidth}\centering
      \includegraphics[scale=.5]{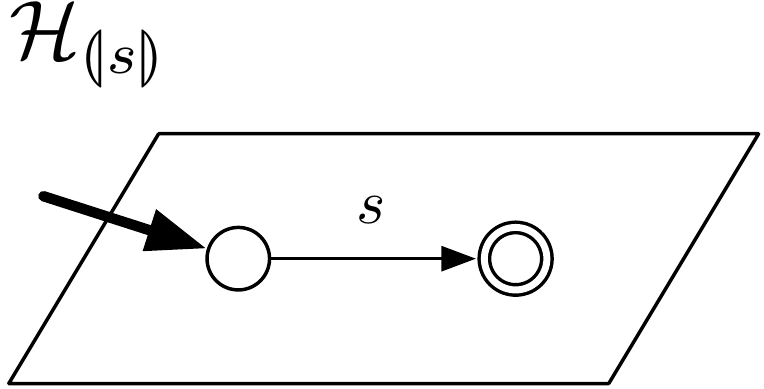}
     \end{minipage}
    \end{center}

    If $\ne = \nex{n}$, we let the corresponding CDA$^\sharp$ in-contexts
    $\etw{C}{\hdof{\nex{n}}}{E}$ with $\hdof{\nex{n}} = \tuple{Q,q_0,\trans,F}$
    where $Q = \{q_0,q_1\}$, $F = \{q_1\}$ and
    \[
    \begin{cases}
     \trans(q,\alpha) = \{q_1\} & \text{if } q = q_0 \text{ and } \el{C}{\alpha}
     = n\\
     \trans(q,\alpha) = \emptyset & \text{otherwise}
    \end{cases}
    \]
    Note that the number of registers in $q_0$ and $q_1$ is the same as
    $\length{C}$, i.e.~$\card{q_0} = \card{q_1} = \length{C}$.
    \begin{center}
     \begin{minipage}{\linewidth}\centering
      \includegraphics[scale=.5]{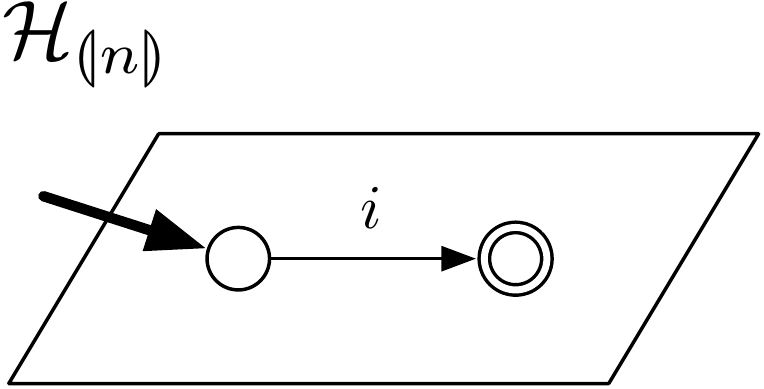}
     \end{minipage}
    \end{center}
    Remember that, since we are only considering closed up-NREs (no free name),
    this $n$ must be local (stored in a unique register in $q_0$) and appears in
    the pre-context $C$.
    This unique register is identified by $\el{C}{\alpha} = n$, and the above
    picture is assuming $\el{C}{i} = n$.

    If $\ne = \cof{n}$, we let the corresponding CDA$^\sharp$ in-contexts
    $\etw{C}{\hdof{\cof{n}}}{E}$ with $\hdof{\cof{n}} = \tuple{Q,q_0,\trans,F}$
    where $Q = \{q_0,q_1\}$, $F = \{q_1\}$ and
    \[
    \begin{cases}
     \trans(q,\alpha) = \{q_1\} & \text{if } q = q_0 \text{ and } \el{C}{\alpha}
     = n\\
     \trans(q,\alpha) = \emptyset & \text{otherwise}
    \end{cases}
    \]
    Note that the number of registers in $q_0$ and $q_1$ is the same as
    $\length{C}$, i.e.~$\card{q_0} = \card{q_1} = \length{C}$.
    \begin{center}
     \begin{minipage}{\linewidth}\centering
      \includegraphics[scale=.5]{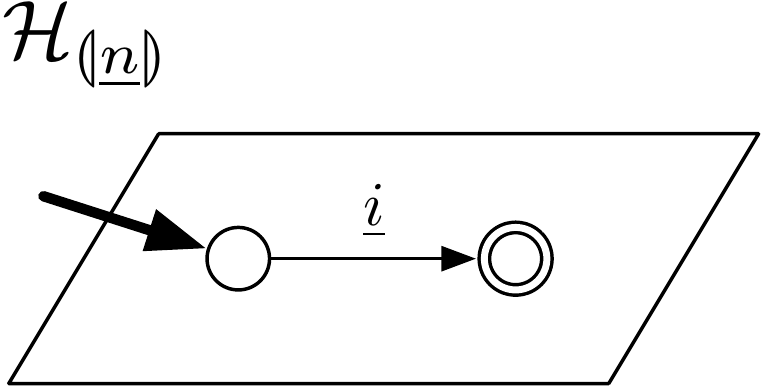}
     \end{minipage}
    \end{center}
    By our assumption, the underline $\cof{n}$ can be added only for local
    names.
    So, $n$ is local (stored in a unique register in $q_0$) and appears in the
    pre-context $C$.
    The register number is identified $\el{C}{\alpha} = n$ as above, and the
    picture is assuming $\el{C}{i} = n$.

\noindent\textbf{Inductive cases.}
    For $\ne_1 + \ne_2$, the corresponding CDA$^\sharp$ in-contexts
    $\etw{C}{\hdof{\ne_1+\ne_2}}{E}$ is $\hdof{\ne_1+\ne_2} =
    \tuple{Q^+,q_0,\trans^+,F^+}$ where
     \begin{itemize}
      \item $q_0$ is a new initial state with $\card{q_0} = \length{C}$
      \item $Q^+ = \{q_0\} \cup Q_1 \cup Q_2$
	 \item $\begin{cases}
		 \trans^+(q,\alpha) = \{q_{\left(1,0\right)},
		 q_{\left(2,0\right)}\} & \text{if } q = q_0 \text{ and } \alpha
		 = \emptystr\\
		 \trans^+(q,\alpha) = \trans_1(q,\alpha) & \text{if } q \in
		 Q_1\\
		 \trans^+(q,\alpha) = \trans_2(q,\alpha) & \text{if } q \in Q_2
		\end{cases}$
      \item $F^+ = F_1 \cup F_2$
     \end{itemize}
     Notice that the previous initial states $q_{(1,0)}$ and $q_{(2,0)}$ have
     the same amount of registers as $q_0$, by the inductive hypothesis,
     i.e.~$\card{q_{(1,0)}} = \card{q_{(2,0)}} = \card{q_0} = \length{C}$.
     \begin{center}
      \begin{minipage}{\linewidth}\centering
       \includegraphics[scale=.3]{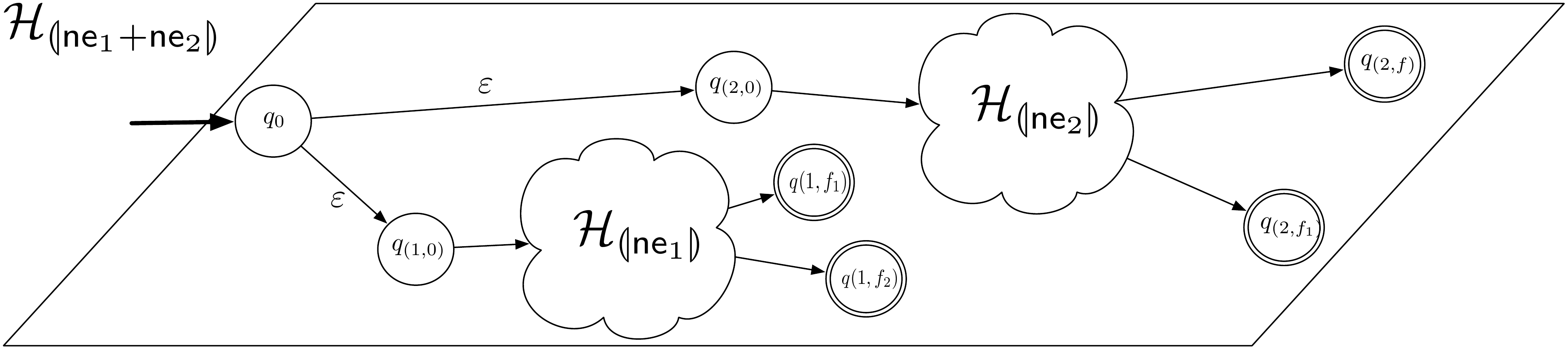}
      \end{minipage}
     \end{center}

     For $\ne_1 \circ \ne_2$, the corresponding CDA$^\sharp$ in-contexts
     $\etw{C}{\hdof{\ne_1\circ\ne_2}}{E}$ is $\hdof{\ne_1\circ\ne_2} =
     \tuple{Q^\circ,q_{\left(1,0\right)},\trans^\circ,F_2}$ where
      \begin{itemize}
       \item $Q^\circ = Q_1 \cup Q_2$
       \item $\begin{cases}
	       \trans(q,\alpha) = \trans_1(q,\alpha) \cup
	       \{q_{\left(2,0\right)}\} & \text{if } q \in F_1 \text{ and }
	       \alpha = \emptystr\\
	       \trans(q,\alpha) = \trans_1(q,\alpha) & \text{if } q \in Q_1
	       \setminus F_1 \text{ or } \alpha \not= \emptystr\\
	       \trans(q,\alpha) = \trans_2(q,\alpha) & \text{if } q \in Q_2
	      \end{cases}$
      \end{itemize}
      By our construction, we easily check that, in each automaton in-contexts,
      the initial state has the same number of registers as those of all final
      states.
      By the induction hypothesis, notice that the final states in
      $\hdof{\ne_1}$ has the same amount of registers as that of the initial
      state $q_{(2,0)}$ in $\hdof{\ne_2}$.
      \begin{center}
       \begin{minipage}{\linewidth}\centering
	\includegraphics[scale=.3]{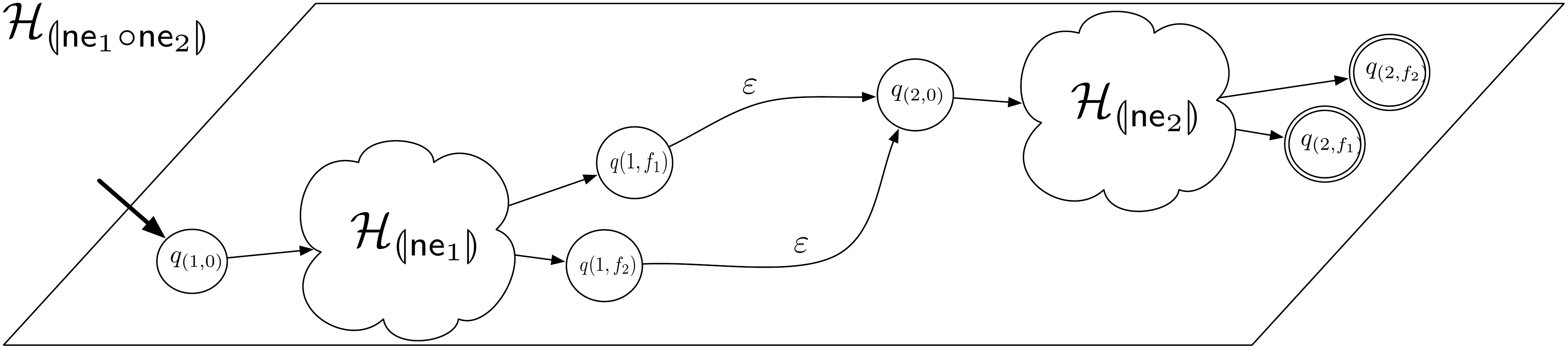}
       \end{minipage}
      \end{center}

      For $\ne^\ast$, the corresponding CDA$^\sharp$ in-contexts
      $\etw{C}{\hdof{\ne^\ast}}{E}$ is $\hdof{\ne^\ast} =
      \tuple{Q,q_0,\trans^\ast,F^\ast}$ where
       \begin{itemize}
	 \item $\begin{cases}
		 \trans^\ast(q,\alpha) = \trans(q,\alpha) \cup \{q_0\} &
		 \text{if } q \in F \text{ and } \alpha = \emptystr\\
		 \trans^\ast(q,\alpha) = \trans(q,\alpha) & \text{otherwise}
		\end{cases}$
	\item $F^\ast = \{q_0\}$
       \end{itemize}
       Notice that the initial state has the same number of registers as those
       of all the final states.
       \begin{center}
	\begin{minipage}{\linewidth}\centering
	 \includegraphics[scale=.3]{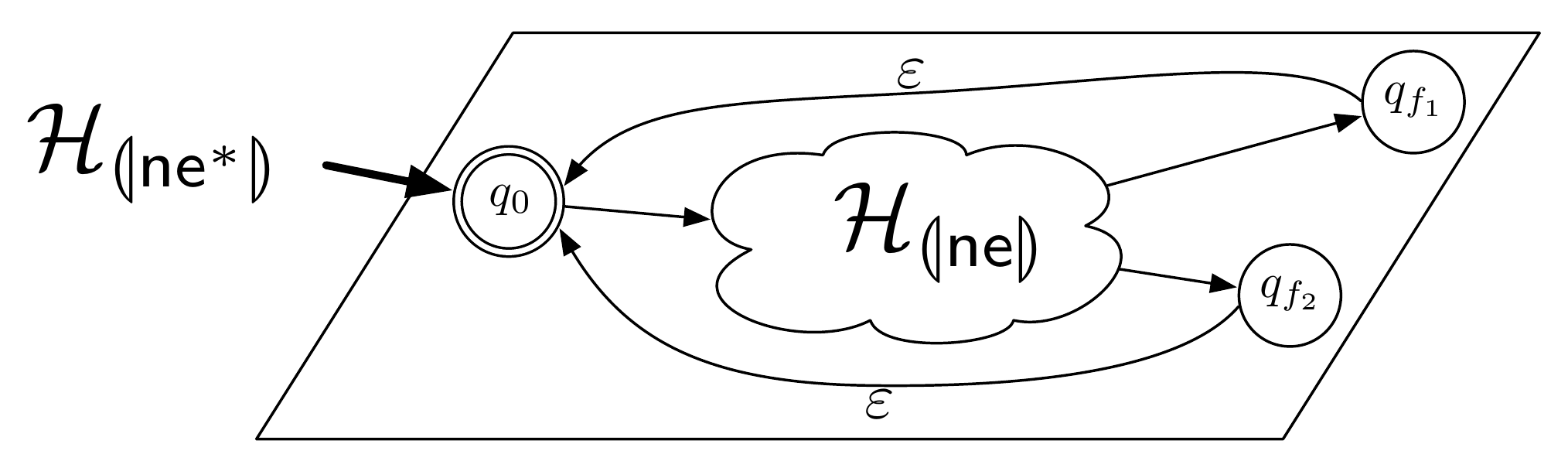}
	\end{minipage}
       \end{center}

       For $\pexp{n}{\ne}{m}$, the corresponding CDA$^\sharp$ in-contexts
       $\etw{C}{\hdof{\pexp{n}{\ne}{m}}}{E}$ is $\hdof{\pexp{n}{\ne}{m}} =
       \tuple{Q^\Diamond,q_s,\trans^\Diamond,F^\Diamond}$ where
       \begin{itemize}
	\item $q_s$ and $q_t$ are new states with $\card{q_s} = \card{q_t} =
	      \length{C}$ (remember $\hdof{\ne}$ is in-contexts $C + [n]$ and
	      $E+[\ihsv{t}{m}]$)
	\item $Q^\Diamond = \{q_s, q_t\} \cup Q$
	\item $q_s$ is the initial state
	 \item $\begin{cases}
		 \trans^\Diamond(q,\alpha) = \{q_0\} & \text{if } q = q_s \text{
		 and } \alpha = \star\\
		 \trans^\Diamond(q,\alpha) = \emptyset & \text{if } \left(q =
		 q_s \text{ and } \alpha \not= \star \right) \text{ or } q =
		 q_t\\
		 \trans^\Diamond(q,\alpha) = \{q_t\} & \text{if } q \in F \text{
		 and } \alpha = \pclose{i} \text{ with }
		 \el{\hcv{E+[\ihsv{t}{m}]}}{i} = n\\
		 \trans^\Diamond(q,\alpha) = \trans(q,\alpha) & \text{otherwise}
		\end{cases}$
	\item $F^\Diamond = \{q_t\}$
       \end{itemize}
       Remember that, for this case, we are assuming that $\hdof{\ne}$ are
       in-contexts between $C + [n]$ and $E + [\ihsv{t}{m}]$.
       Namely, $\card{q_s} = \card{q_t} = \length{C}$ and $\card{q_0} =
       \length{C} + 1$.
       Hence, $q_s$ can take a $\star$-transition to $q_0$ and final states in
       $\hdof{\ne}$ can take $\pclose{i}$-transitions to $q_t$.
       \begin{center}
	\begin{minipage}{\linewidth}\centering
	 \includegraphics[scale=.3]{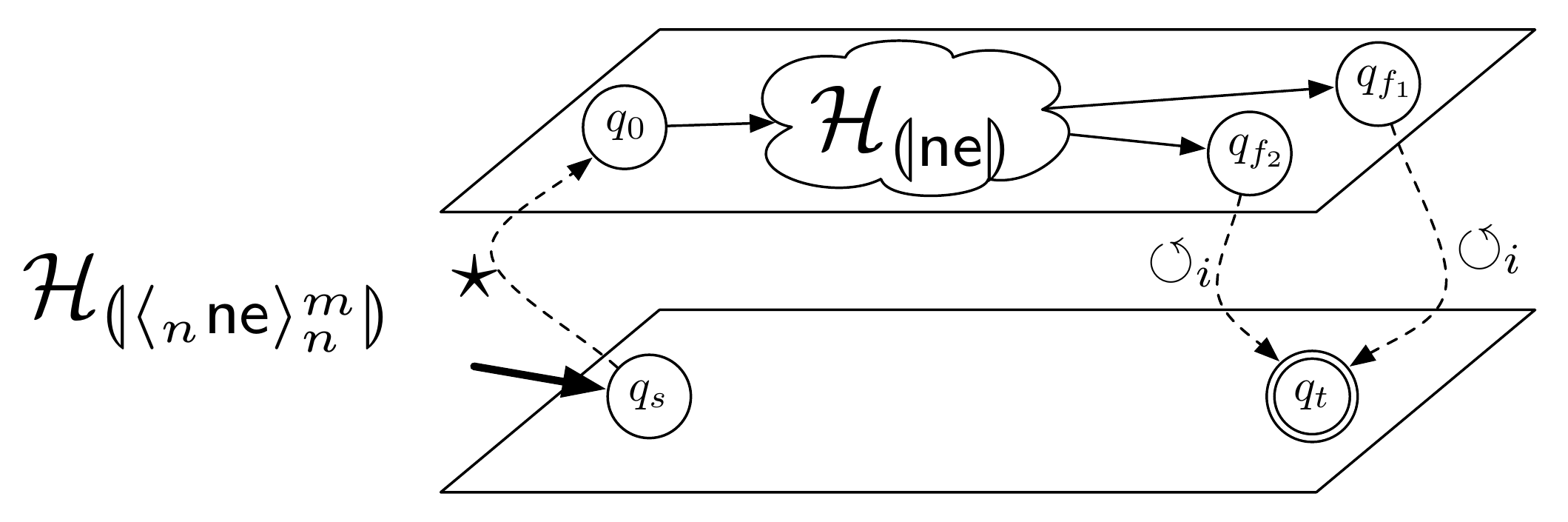}
	\end{minipage}
       \end{center}
       Moreover, notice that, by our assumption for the superscript names on
       closing brackets, $m$ must be in a unique register in the pre-context
       $C$, which can be identified with $\el{\hcv{E+[\ihsv{t}{m}]}}{i} = n$ and
       $\el{C}{i} = m$.
       Recall the rules ($\cdmd_=$) and ($\cdmd_{\not=}$) in
       Fig.~\ref{fig:contextcalculus}.

\subsection{\cda\ accepts NREs}
\begin{proposition}
 \label{prop:base}
 For up-NREs in-contexts $\etw{C}{\nex{1}}{E}$, $\etw{C}{\nex{0}}{E}$,
 $\etw{C}{s}{E}$, $\etw{C}{\nex{n}}{E}$ and $\etw{C}{\cof{n}}{E}$, the
 corresponding CDA$^\sharp$ in-contexts $\etw{C}{\hdof{1}}{E}$,
 $\etw{C}{\hdof{0}}{E}$, $\etw{C}{\hdof{s}}{E}$, $\etw{C}{\hdof{n}}{E}$ and
 $\etw{C}{\hdof{\cof{n}}}{E}$ accept the language in-contexts
 $\etw{C}{\langof{\left(\nex{1}\right)}}{E}$,
 $\etw{C}{\langof{\left(\nex{0}\right)}}{E}$,
 $\etw{C}{\langof{\left(\nex{s}\right)}}{E}$,
 $\etw{C}{\langof{\left(\nex{n}\right)}}{E}$ and
 $\etw{C}{\langof{\left(\cof{n}\right)}}{E}$, respectively.
\end{proposition}
\begin{proof}
    The only non-trivial case is $\etw{C}{\cof{n}}{E}$.
    In this case, the language in-contexts is obtained by the language calculus
    as follows:
    \[
    \left\{\star \in \names \sep \star \lfresh C \text{ and } \star \gfresh^i
    \el{\nec{C}}{i}\right\}
    \]
    with the post context $\ihsv{\left(E@\star\right)}{\left(\left(n\
    \star\right)\cdot\hcv{E}\right)}$, where $i$ is the register number whose
    current value is $n$, i.e.~$\el{\hcv{C}}{i} = n$.
    The corresponding CDA$^\sharp$ is given as follows:
    \begin{center}
     \begin{minipage}{\linewidth}\centering
      \includegraphics[scale=.4]{hdofcofn}
     \end{minipage}
    \end{center}
    Hence, the initial configuration $\tuple{q_0,\star,\nec{C}}$ can reach the
    final state $q_1$ if and only if $\star \lfresh C$ and $\star \gfresh^i
    \el{\nec{C}}{i}$.
    This is because $\star$ must be fresh for all the current names $\star
    \lfresh C$ also for the chronicle $i$, i.e.~$\el{\nec{C}}{i}$.
    Notice that, later on, when we concatenate with other languages, we may
    change, by permuting names and appending chronicles, the local and global
    freshness conditions (recall rule ($\lcirc$) in
    Fig.~\ref{fig:languagecalculus} and see how it works in Appendix
    \ref{sec:complex}).
    Also, the post-context must correspond to the post-context given by the
    language calculus, i.e.~$\ihsv{\left(E@\star\right)}{\left(\left(n\
    \star\right)\cdot\hcv{E}\right)}$, by the definition of the movement of
    CDA$^\sharp$.
\qed
\end{proof}

\begin{figure}\centering
  \resizebox{!}{4cm}{
  \begin{minipage}{\linewidth}\centering
   \includegraphics[scale=.2]{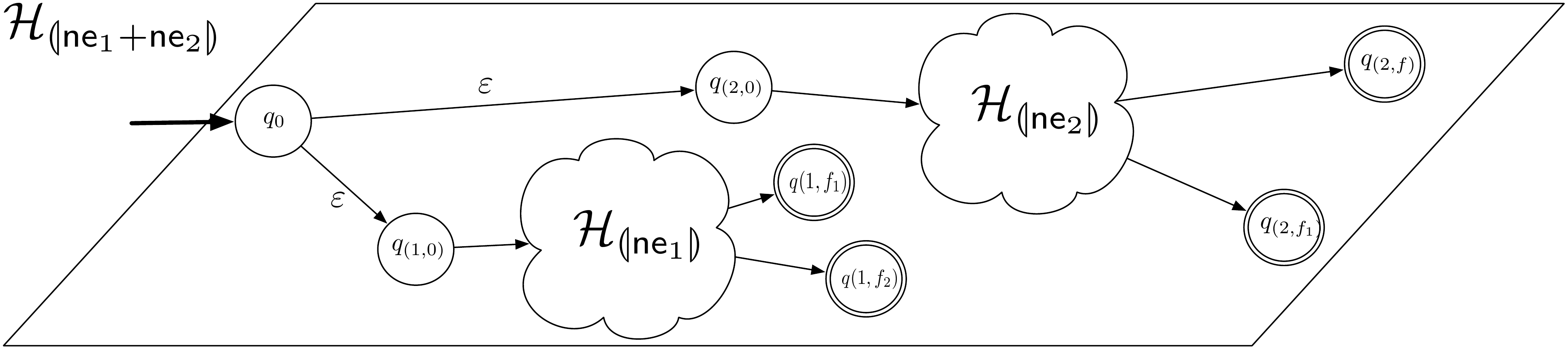}
   \\[1cm]
   \includegraphics[scale=.2]{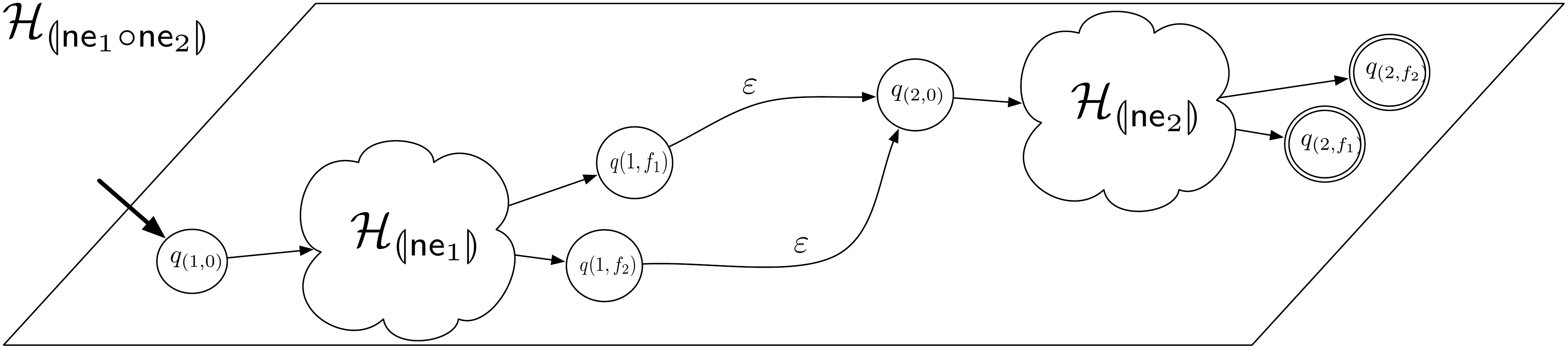}
   \\[1cm]
   \includegraphics[scale=.2]{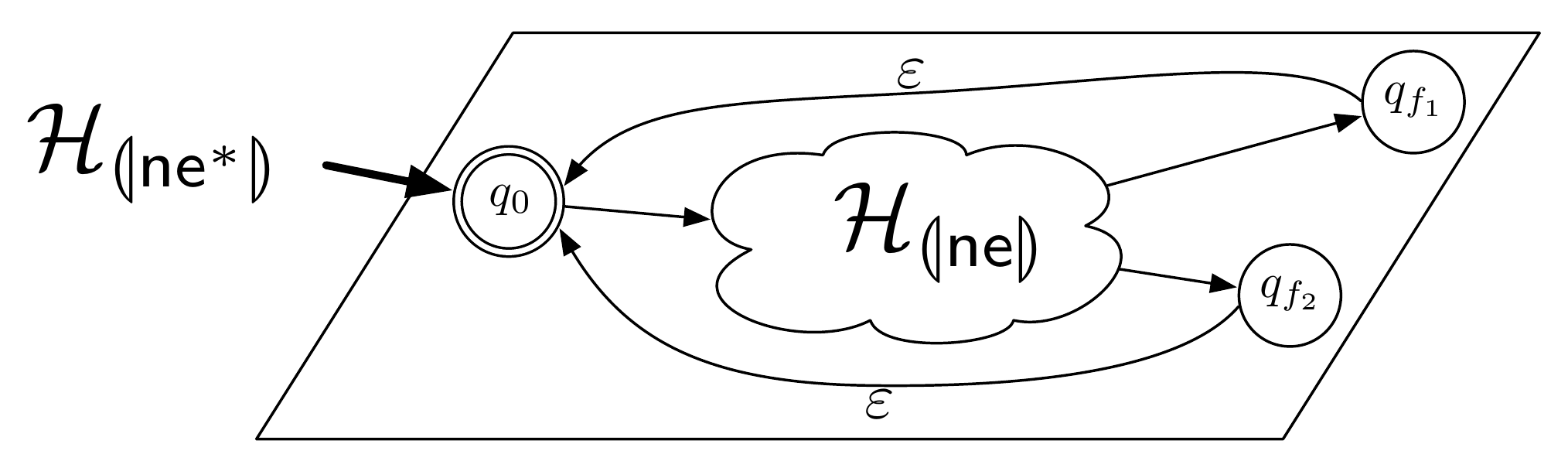}
   \qquad
   \includegraphics[scale=.2]{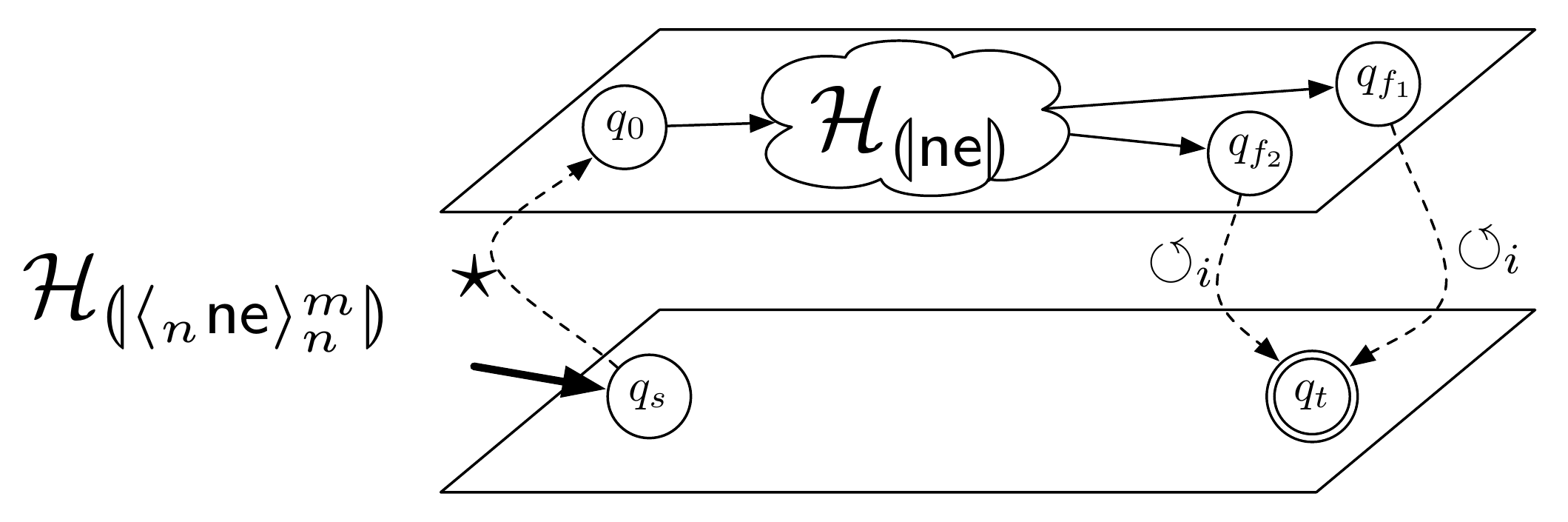}
  \end{minipage}
  }
 \caption{The inductive constructions for $\hdof{\ne_1 + \ne_2}$,
 $\hdof{\ne_1\circ\ne_2}$, $\hdof{\ne^\ast}$ and
 $\hdof{\pexp{n}{\ne}{m}}$.\label{fig:indsteps}}
\end{figure}

The construction from NREs to automata is summarised by the next two
propositions.
\begin{proposition}
 Given two NREs $\ne_1$ and $\ne_2$, a pre-context $C$, and a post-context $E$,
 the CDA$^\sharp$ in-contexts $\etw{C}{\hdof{\ne_1+\ne_2}}{E}$ recognises the
 language in-contexts $\etw{C}{\langof{\left(\ne_1+\ne_2\right)}}{E}$ while the
 CDA$^\sharp$ in-contexts $\etw{C}{\hdof{\ne_1\circ\ne_2}}{E}$ recognises the
 language in-contexts $\etw{C}{\langof{\left(\ne_1\circ\ne_2\right)}}{E}$.
\end{proposition}
\begin{proof}
    Because of the context calculus, for an NRE in-contexts
    $\etw{C}{\ne_1\circ\ne_2}{E}$, we assume, as the inductive hypothesis, the
    languages in-contexts $\etw{C}{\langof{(\ne_1)}}{\nec{C}}$ and
    $\etw{C}{\langof{(\ne_2)}}{E}$ obtained by $\etw{C}{\ne_1}{\nec{C}}$ and
    $\etw{C}{\ne_2}{E}$ are accepted by automata in-contexts
    $\etw{C}{\hdof{\ne_1}}{\nec{C}}$ and $\etw{C}{\hdof{\ne_2}}{E}$,
    respectively.
    Here we let the schematic words for $\langof(\ne_1)$ and $\langof(\ne_2)$
    be $\phw{\bullet^1_1 \ldots \bullet^1_{k_1}}{\phi_1}$ and $\phw{\bullet^2_1
    \cdots \bullet^2_{k_2}}{\phi_2}$, respectively.
    Note that, the post-contexts in languages in-contexts are not necessarily
    reflecting the real extent chronicles in their final states.
    This is because it may have loops or unions in $\hdof{\ne_1}$ and
    $\hdof{\ne_2}$, then the extant chronicles may change depending on how many
    times each word makes loops until it is recognised, etc.
    However, when we consider each path (without unions and Kleene stars),
    i.e.~schematic words in Fig.~\ref{fig:languagecalculus}, each
    post-contexts reflects the 'real' extant chronicle in the final
    configuration in the corresponding \cda: see Appendix \ref{sec:complex}.

    By the language calculus, we obtain the following schematic word for each
    pair of schematic words $\langof_1$ and $\langof_2$:
    \[
    \phw{\bullet^1_1 \cdots \bullet^1_{k_1} \circ
    \pi_{\left[{C}{\rhd}{E_1}\right]} \cdot \left(\bullet^2_1 \cdots
    \bullet^2_{k_2}\right)}{\phi_1, \left(\pi_{\left[{C}{\rhd}{E_1}\right]}
    \cdot \phi_2 \right)}
    \]
    Since $\ne_1$ and $\ne_2$ may contain $+$ or $\_^\ast$, the post-contexts
    $\nec{C}$ and $E$ obtained in Fig.~\ref{fig:contextcalculus} may change to
    some other extant chronicles depending on which path we take or how many
    time we make loops etc during \lngc.
    Hence, we assume for the current schematic words in-contexts that they
    have $E_1$ and $E_2$ as their post-contexts.
    Notice that, the post-chronicles $E_1$ and $E_2$ reflect the extant
    chronicles in their final configurations.
    As the permutation action $\pi_{\left[{C}{\rhd}{E_1}\right]}$, by
    definition, permutes the current values to start $\hdof{\ne_2}$ to the
    current values of $E_1$.
    Not only that, it appends the chronicles in $E_1$ to the initial
    configuration of the second \cda.
    Accordingly, it updates the local freshness conditions and the relative
    global freshness conditions in $\phi_2$ to the appropriate one: also see how
    it works in Appendix \ref{sec:complex}.
    Hence, the construction of $\etw{C}{\hdof{\ne_1\circ\ne_2}}{E}$ works.
    \begin{figure*}[h!]\centering
      \includegraphics[scale=.25]{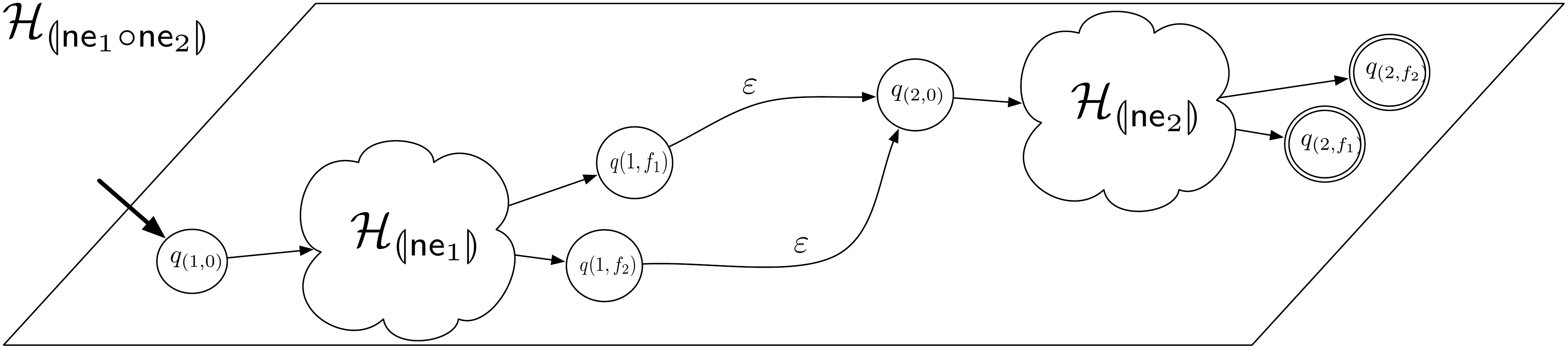}
    \end{figure*}  
\qed
\end{proof}

\begin{proposition}
 Given an NRE $\ne$, a pre-context $C$, and a post-context $E$, the CDA$^\sharp$
 in-contexts the CDA$^\sharp$ in-contexts $\etw{C}{\hdof{\ne^\ast}}{E}$
 recognises the language in-contexts
 $\etw{C}{\langof{\left(\ne^\ast\right)}}{E}$ while the CDA$^\sharp$ in-contexts
 $\etw{C}{\hdof{\pexp{n}{\ne}{m}}}{E}$ recognises the language in-contexts
 $\etw{C}{\langof{\left(\pexp{n}{\ne}{m}\right)}}{E}$.
\end{proposition}
\begin{proof}
 Let a language in-contexts
 $\etw{C+[n]}{\langof{(\ne)}}{E+[\ihsv{t}{m}]}$ be recognised by the
 CDA$^\sharp$ in-contexts $\etw{C+[n]}{\hdof{\ne}}{E+[\ihsv{t}{m}]}$,
 with the schematic word for $\langof{(\ne)}$ being $\phw{\bullet_1
   \cdots \bullet_k}{\phi}$.
 By the rules of \lngc, a schematic word for
 $\langof{\pexp{n}{\ne}{m}}$ is

 \begin{equation}\label{eq:sw}
   \phw{\left(n\ \phl\right) \cdot
   \bullet_1 \cdots \bullet_k}{\phl \lfresh \hcv{C},
   \left(n\ \phl\right) \cdot \phi}
 \end{equation}

 The corresponding CDA$^\sharp$ in-contexts can store any name $\star$
 locally fresh wrt $\hcv{C}$.
 By induction hypothesis, for each instance $\left(n\ \star\right)
 \cdot \bullet_1 \cdots \bullet_k$ of~\eqref{eq:sw} such that $\star
 \lfresh \hcv{C}$ and $\left(n\ \star\right) \circ \phi$ holds, it is
 the case that a state of $\hdof{\pexp{n}{\ne}{m}}$ corresponding to a
 final state of $\hdof{\ne}$ is reached (now on the $1$-st layer of
 $\hdof{\pexp{n}{\ne}{m}}$).
 To help the intuition, consider the following figure
  \begin{figure*}[h!]\centering
   \includegraphics[scale=.25]{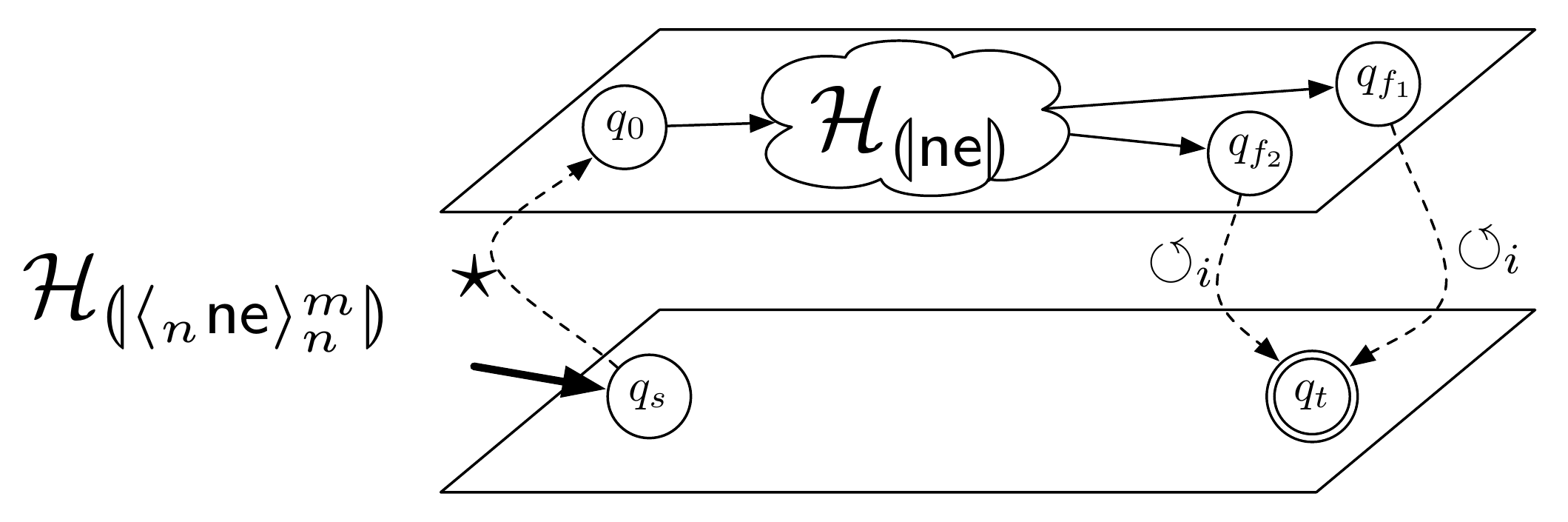}
 \end{figure*}
  (where $q_{f_1}$ and $q_{f_2}$ are final states of $\hdof \ne$).

 Now, to remove an appropriate current value and the last chronicle from the
 final extant chronicle, we have to choose the corresponding $\pclose{i}$ for
 some $i$.
 Thanks to the rules of \ctxc, we can choose $i$ such that $m = \el C
 i$ (recall rule ($\cdmd_{\not=}$) in \figref{fig:contextcalculus} and
 note that the existence of $i$ is guaranteed by our constrains:
 $\pexp{n}{\ne}{m}$ must appear in a scope of $m$).
 Hence, the automaton stops with accounting of the corresponding
 permutations on $\rangle^m_n$.

 The proof of the other cases is similar.
\qed
\end{proof}

\subsection{Each \cda\ has an NRE}
\thmref{thm:noma2nre} shows that each language accepted by a \cda\
can be described by an NRE.
\begin{theorem}
 \label{thm:noma2nre}
 Each language accepted by an CDA$^\sharp$, CA$^\sharp$, DA$^\sharp$ or
 A$^\sharp$ is nominal regular. That is, there exists an up-NRE, u-NRE, p-NRE or
 b-NRE which generates the same language.
\end{theorem}
 \begin{proof}
  This is almost the same as the proof in \cite{KurzST12tcs}. The only
  difference is that we have $\gft{i}$ transitions. For those transitions, we
  just take the corresponding names with underlines.

    Let $\mathcal{H}$ be a CDA$^\sharp$.
    Since each layer, if we ignore $\star$-transitions and
    $\pclose{i}$-transitions, is a classical automaton.
    Hence, by the well known method ($\emptystr$-closure and the powerset
    construction), we make each layer deterministic.
    For all $\star$-transitions and $\pclose{i}$-transitions, we make another
    powerset construction to connect each layer as follows: for each state $Q_j
    = \left\{q^j_1,\ldots,q^j_k\right\}$ on the $j$-th layer, remember that each
    state is a subset of states because of the first powerset construction, we
    let
    \begin{align*}
     \trans'(Q_j,\star) &\mmdef \left\{q^{j+1} \sep \exists q^j \in Q_j.\
     q^{j+1} \in \trans(q^j,\star)\right\}\\
     \trans'(Q_j,\pclose{i}) &\mmdef \left\{q^{j-1} \sep \exists q^j \in Q_j.\
     q^{j-1} \in \trans(q^j,\pclose{i})\right\}
    \end{align*}
    where $\trans$ and $\trans'$ are transitions after the first powerset
    construction and the second one, respectively.
    So the CDA$^\sharp$ is now deterministic.

    For the obtained automaton, as in the case of the classical language theory,
    we calculate paths inductively.
    But, in our case, the inductive steps are also separated into two steps.
    Namely, the first step is on the highest layer of the automaton, which is
    almost the same as the classical method.
    The only difference is that in our automaton, names are labeled by natural
    numbers.
    After that, we make another induction on layers, see also
    \cite{KurzST12fossacs}).
    Notice that, to bind names, we use a canonical naming,
    i.e.~$\left[n_1,\ldots,n_h\right]$ ($n_i$ is allocated to the label $i$).
    Hence the translation from accepted paths to expressions are
    straightforward.

    Finally, the definition of subclasses of CDA$^\sharp$ tells their
    corresponding types of nominal regular expressions (CDA$^\sharp$,
    CA$^\sharp$, DA$^\sharp$ and A$^\sharp$ corresponding to up-NREs, u-NREs,
    p-NREs and NREs, respectively).
    \qed
 \end{proof}
 \begin{corollary}
  Nominal regular languages are closed under union, concatenation and Kleene
  star.
 \end{corollary}

 For the languages with explicit binders considered in~\cite{KurzST12fossacs} it
 is possible to define a notion of \emph{resource-sensitive complementation} and
 prove that such languages are closed under resource-sensitive complementation.
 This is not possible when considering languages over infinite alphabets
 \emph{without} explicit binders.

 As a corollary of our theory, we describe how to define nominal regular
 expressions for fresh-register automata and register automata.
 Consider the following subclass of up-NREs (that we call \emph{first-degree}
 up-NREs):
 \[
 \fne ::= 1 \sep 0 \sep n_i \sep \underline{n_i} \sep s \sep
 \pexp{n_{h+1}}{n_{h+1}}{n_{i}} \sep \fne + \fne \sep \fne \circ \fne \sep
 \fne^\ast
 \]
 where $n_1, \ldots, n_{h+1}$ are pairwise distinct names and $s \in \letters$.
 Furthermore, an \emph{h-prefixed first-degree} up-NREs is a first-degree up-NRE
 of the form $\npexp{n_1}{\cdots\npexp{n_h}{\fne}\cdots}$ where $\fne$ is a
 binder-free first-degree up-NRE.
 Then we can prove the following result:
 \begin{theorem}
  For every FRA (RA), there is an up-NRE (p-NRE) which generates the accepted
  language. More precisely, the up-NRE (p-NRE) is \emph{h-prefixed
  first-degree}. Hence every FRA (RA) is expressible by an h-prefixed
  first-degree up-NRE (p-NRE) $\npexp{n_1}{\cdots\npexp{n_h}{\fne}\cdots}$.
 \end{theorem}


\section{Conclusion}

We studied different types of automata and languages over
  infinite alphabets and gave Kleene type theorems characterising them
  by regular expressions.  On the one hand, this extends the work on
  automata over infinite alphabets begun in \cite{KaminskiF94}, on the
  other hand the automata we propose are variations on the HD-automata
  of \cite{MontanariP97,MontanariP05} (in particular, our
transitions allocating fresh-names and  permuting names are borrowed
and adapted from HDA).
 As HDA are automata
  internal in the category of named sets, this also means, see
  \cite{GadducciMM06},  that our work can be seen in the context of
  nominal sets \cite{GabbayP99} and the more recent line of research
  on nominal automata \cite{BojanczykKL11}.

  Regular expressions for register automata were investigated in
  \cite{KaminskiT06,KaminskiZ10}. A difference is that the NREs of
  this paper have primitives for allocation and deallocation and
  permutations. Moreover, we also introduced NREs for relative global
  freshness.

The novel notion of relative global freshness is closely
related to the recent \cite{GrigoreDPT13,TzevelekosG13}. Whereas we
are interested in choreographies, \cite{GrigoreDPT13} use register
automata to monitor the execution of Java programs that generate a
potentially unbounded number of names, albeit without using global
freshness or histories.
The history register automata (HRA) of \cite{TzevelekosG13} share with
\cda\ the ability to "forget" names since reset transitions can modify
histories.
We observe that~\cite{TzevelekosG13} makes no attempt at finding a
class of corresponding regular expressions. A detailed comparison as
well as the definition of NREs for HRA have to be left as future work.


\bibliographystyle{abbrv}
\bibliography{stringdef-short,all}

\begin{thebibliography}{10}

\bibitem{BojanczykKL11}
M.~Bojanczyk, B.~Klin, and S.~Lasota.
\newblock Automata with group actions.
\newblock In {\em LICS'11}.

\bibitem{GabbayP99}
M.~Gabbay and A.~M. Pitts.
\newblock A new approach to abstract syntax involving binders.
\newblock In {\em LICS'99}.

\bibitem{GadducciMM06}
F.~Gadducci, M.~Miculan, and U.~Montanari.
\newblock About permutation algebras, (pre)sheaves and named sets.
\newblock {\em Higher-Order and Symbolic Computation}, 19(2-3), 2006.

\bibitem{GrigoreDPT13}
R.~Grigore, D.~Distefano, R.~L. Petersen, and N.~Tzevelekos.
\newblock Runtime verification based on register automata.
\newblock In {\em TACAS'13}.

\bibitem{hvk98}
K.~Honda, V.~Vasconcelos, and M.~Kubo.
\newblock Language primitives and type discipline for structured
  communication-based programming.
\newblock In {\em ESOP'13}.

\bibitem{KaminskiF94}
M.~Kaminski and N.~Francez.
\newblock Finite-memory automata.
\newblock {\em Theoret.\ Comput.\ Sci.}, 134(2), 1994.

\bibitem{KaminskiT06}
M.~Kaminski and T.~Tan.
\newblock Regular expressions for languages over infinite alphabets.
\newblock {\em Fundam. Inform.}, 69(3), 2006.

\bibitem{KaminskiZ10}
M.~Kaminski and D.~Zeitlin.
\newblock Finite-memory automata with non-deterministic reassignment.
\newblock {\em Int. J. Found. Comput. Sci.}, 21(5), 2010.

\bibitem{w3c:cho}
N.~Kavantzas, D.~Burdett, G.~Ritzinger, T.~Fletcher, and Y.~Lafon.
\newblock \url{http://www.w3.org/TR/2004/WD-ws-cdl-10-20041217}.
\newblock Working Draft 17 December 2004.

\bibitem{KurzST12tcs}
A.~Kurz, T.~Suzuki, and E.~Tuosto.
\newblock A characterisation of languages on infinite alphabets with nominal
  regular expressions.
\newblock In {\em IFIP TCS'12}.

\bibitem{KurzST12fossacs}
A.~Kurz, T.~Suzuki, and E.~Tuosto.
\newblock On nominal regular languages with binders.
\newblock In {\em FoSSaCS'12}.

\bibitem{MontanariP97}
U.~Montanari and M.~Pistore.
\newblock An introduction to history dependent automata.
\newblock {\em Electr. Notes Theor. Comput. Sci.}, 10, 1997.

\bibitem{MontanariP05}
U.~Montanari and M.~Pistore.
\newblock Structured coalgebras and minimal {HD}-automata for the {\it
  pi}-calculus.
\newblock {\em Theor. Comput. Sci.}, 340(3), 2005.

\bibitem{Tzevelekos11}
N.~Tzevelekos.
\newblock Fresh-register automata.
\newblock In {\em POPL'11}.

\bibitem{TzevelekosG13}
N.~Tzevelekos and R.~Grigore.
\newblock History-register automata.
\newblock In {\em FoSSaCS'13}.

\end{thebibliography}

\appendix
\newpage
\section{An example}
\label{sec:complex}

As a more complex example, we consider the language and the automaton for the
following up-NRE:
\[
 \npexp{n}{n \npexp{m}{m \pexp{l}{l}{m} m \npexp{l}{\cof{n} l \cof{m}} } }
\]

By \ctxc\ in Fig.~\ref{fig:contextcalculus}, we obtain the proof tree in
Fig.~\ref{ex:apndctxt}.
In the tree, it is not necessary that $c$ is different from $d$.
It is also not important ow to choose their names.
The only thing we have to care is to keep pre-contexts $C$ \emph{pairwise
distinct}.
Note that, in Fig.~\ref{ex:apndctxt}, Fig.~\ref{ex:apndlngc1} and
Fig.~\ref{ex:apndlngc2}, capital alphabets A, B, C, D, E, F are added to
elementary NREs in-contexts, and round-bracketed numbers (1) - (10) are added to
point out NREs on each inductive step.
Also, double-dashed lines are used for simplifications (in particular, to remove
repeating names in chronicles).

\begin{figure}[p]
{\small
 \begin{minipage}{\linewidth}
  \rootAtTop
  \AxiomC{$\mathcal{C}_1$}
  \AxiomC{B}
  \noLine
  \UnaryInfC{$\vdots$}
  \noLine
  \UnaryInfC{$\etw{[a, b]}{b}{[\ihsv{ab}{a}, \ihsv{b}{b}]}$}
  \AxiomC{C}
  \noLine
  \UnaryInfC{$\vdots$}
  \noLine
  \UnaryInfC{$\etw{[a, b, c]}{c}{[\ihsv{abc}{a}, \ihsv{bc}{c}, \ihsv{cb}{b}]}$}
  \RightLabel{($\cdmd_{\not=}$)}
  \UnaryInfC{(5): $\etw{[a, b]}{\pexp{l}{l}{b}}{[\ihsv{ab}{a}, \ihsv{b}{b}]}$}
  \RightLabel{($\ccirc$)}
  \BinaryInfC{(6): $\etw{[a, b]}{b \pexp{l}{l}{b}}{[\ihsv{ab}{a},
  \ihsv{b}{b}]}$}
  \AxiomC{$\mathcal{C}_2$}
  \RightLabel{($\ccirc$)}
  \BinaryInfC{(7): $\etw{[a, b]}{b \pexp{l}{l}{b} b \npexp{l}{\cof{a} l
  \cof{b}}}{[\ihsv{ab}{a}, \ihsv{b}{b}]}$}
  \RightLabel{($\cdmd_=$)}
  \UnaryInfC{(8): $\etw{[a]}{\npexp{m}{m \pexp{l}{l}{m} m \npexp{l}{\cof{a} l
  \cof{m}}}}{[\ihsv{a}{a}]}$}
  \RightLabel{($\ccirc$)}
  \BinaryInfC{(9): $\etw{[a]}{a \npexp{m}{m \pexp{l}{l}{m} m \npexp{l}{\cof{a} l
  \cof{m}}}}{[\ihsv{a}{a}]}$}
  \RightLabel{($\cdmd_=$)}
  \UnaryInfC{(10): $\etw{[]}{\npexp{n}{n \npexp{m}{m \pexp{l}{l}{m} m
  \npexp{l}{\cof{n} l \cof{m}}}}}{[]}$}
  \DisplayProof

  \vspace*{10pt}

  \rootAtTop
  \AxiomC{A}
  \noLine
  \UnaryInfC{$\vdots$}
  \noLine
  \UnaryInfC{$\etw{[a]}{a}{[\ihsv{a}{a}]}$}
  \dashedLine
  \UnaryInfC{$\mathcal{C}_1$}
  \DisplayProof

 \vspace*{10pt}

  \rootAtTop
  \AxiomC{B}
  \noLine
  \UnaryInfC{$\vdots$}
  \noLine
  \UnaryInfC{$\etw{[a, b]}{b}{[\ihsv{ab}{a}, \ihsv{b}{b}]}$}
  \AxiomC{$\mathcal{C}_3$}
  \AxiomC{F}
  \noLine
  \UnaryInfC{$\vdots$}
  \noLine
  \UnaryInfC{$\etw{[a, b, d]}{\cof{b}}{[\ihsv{abd}{a}, \ihsv{bd}{b},
  \ihsv{d}{d}]}$}
  \RightLabel{($\ccirc$)}
  \BinaryInfC{(2): $\etw{[a, b, d]}{\cof{a} d \cof{b}}{[\ihsv{abd}{a},
  \ihsv{bd}{b}, \ihsv{d}{d}]}$}
  \RightLabel{($\cdmd_=$)}
  \UnaryInfC{(3): $\etw{[a, b]}{\npexp{l}{\cof{a} l \cof{b}}}{[\ihsv{ab}{a},
  \ihsv{b}{b}]}$}
  \RightLabel{($\ccirc$)}
  \BinaryInfC{(4): $\etw{[a, b]}{b \npexp{l}{\cof{a} l \cof{b}}}{[\ihsv{ab}{a},
  \ihsv{b}{b}]}$}
  \dashedLine
  \UnaryInfC{$\mathcal{C}_2$}
  \DisplayProof

  \vspace*{10pt}

  \rootAtTop
  \AxiomC{D}
  \noLine
  \UnaryInfC{$\vdots$}
  \noLine
  \UnaryInfC{$\etw{[a, b, d]}{\cof{a}}{[\ihsv{abd}{a}, \ihsv{bd}{b},
  \ihsv{d}{d}]}$}
  \AxiomC{E}
  \noLine
  \UnaryInfC{$\vdots$}
  \noLine
  \UnaryInfC{$\etw{[a, b, d]}{d}{[\ihsv{abd}{a}, \ihsv{bd}{b}, \ihsv{d}{d}]}$}
  \RightLabel{($\ccirc$)}
  \BinaryInfC{(1): $\etw{[a, b, d]}{\cof{a}d}{[\ihsv{abd}{a}, \ihsv{bd}{b},
  \ihsv{d}{d}]}$}
  \dashedLine
  \UnaryInfC{$\mathcal{C}_3$}
  \DisplayProof
 \end{minipage}
 }
 \caption{\ctxc\ for $\npexp{n}{n \npexp{m}{m \pexp{l}{l}{m} m \npexp{l}{\cof{n}
 l \cof{m}} } }$\label{ex:apndctxt}}
\end{figure}

 For the derivation tree, by \lngc\ in Fig.~\ref{fig:languagecalculus}, we
 compute the schematic word from the backward direction (i.e.~from leaves to
 the root) as in Fig.~\ref{ex:apndlngc1} and Fig.~\ref{ex:apndlngc2}.
 The schematic word we obtain is
 \[
  \phw{\phl \phl_5 \phl_4 \phl_4 \phl_1 \phl_3\phl_2}{
  \left(
  \begin{array}{c}
   \phl \not= \phl_1, \phl \not= \phl_3, \phl \not= \phl_4, \phl \not= \phl_5,
    \phl_1 \not= \phl_2,\\
   \phl_1 \not= \phl_3, \phl_1 \not= \phl_4, \phl_1 \not= \phl_5, \phl_2 \not=
    \phl_3,\\
   \phl_2 \not= \phl_4, \phl_2 \not= \phl_5, \phl_3 \not= \phl_4, \phl_4 \not=
    \phl_5
  \end{array}
  \right)}
 \]
 so the nominal regular language is
 \begin{align*}
  \bigl\{
  abccdef \in \names^\ast \mid a \not= b, &a \not= c, a \not= d, a \not= e, b
  \not= c, b \not= d, b \not= f,\\
  &c \not= d, c \not= e, c \not= f, d \not= e, d \not= f, e \not= f
  \bigr\}.
 \end{align*}
 Notice that $a$ and $b$ can appear as $f$ and $e$, respectively, in this
 language.

\begin{figure}[p]
{\small
 \vspace*{50pt}
\begin{sideways}
 \begin{minipage}{\textheight}
  \AxiomC{D}
  \noLine
  \UnaryInfC{$\vdots$}
  \noLine
  \UnaryInfC{$\etw{[a, b, d]}{\cof{a}}{[\ihsv{abd}{a}, \ihsv{bd}{b},
  \ihsv{d}{d}]}$}
  \RightLabel{($\cof{a}$)}
  \UnaryInfC{$\etw{[a, b, d]}{\phw{\phl_1}{\phl_1 \lfresh abd, \phl_1 \gfresh^1
  abd}}{[\ihsv{abd \phl_1}{\phl_1}, \ihsv{bd \phl_1}{b}, \ihsv{d \phl_1}{d}]}$}
  \AxiomC{E}
  \noLine
  \UnaryInfC{$\vdots$}
  \noLine
  \UnaryInfC{$\etw{[a, b, d]}{d}{[\ihsv{abd}{a}, \ihsv{bd}{b}, \ihsv{d}{d}]}$}
  \RightLabel{($d$)}
  \UnaryInfC{$\etw{[a, b, d]}{\phw{d}{}}{[\ihsv{abd}{a}, \ihsv{bd}{b},
  \ihsv{d}{d}]}$}
  \RightLabel{($\lcirc$)}
  \BinaryInfC{$\etw{[a, b, d]}{\phw{\phl_1 d}{\phl_1 \lfresh abd, \phl_1
  \gfresh^1 abd}}{[\ihsv{abd \phl_1 \phl_1 bd}{\phl_1}, \ihsv{b}{bd \phl_1 bd},
  \ihsv{d \phl_1 d}{d}]}$}
  \dashedLine
  \doubleLine
  \UnaryInfC{$\etw{[a, b, d]}{\phw{\phl_1 d}{\phl_1 \lfresh abd, \phl_1
  \gfresh^1 abd}}{[\ihsv{abd \phl_1}{\phl_1}, \ihsv{bd \phl_1}{b}, \ihsv{d
  \phl_1}{d}]}$}
  \dashedLine
  \UnaryInfC{$\mathcal{L}_3$}
  \DisplayProof

  \vspace*{10pt}

  \AxiomC{B}
  \noLine
  \UnaryInfC{$\vdots$}
  \noLine
  \UnaryInfC{$\etw{[a, b]}{b}{[\ihsv{ab}{a}, \ihsv{b}{b}]}$}
  \RightLabel{($b$)}
  \UnaryInfC{$\etw{[a, b]}{\phw{b}{}}{[\ihsv{ab}{a}, \ihsv{b}{b}]}$}
  \AxiomC{$\mathcal{L}_3$}
  \AxiomC{F}
  \noLine
  \UnaryInfC{$\vdots$}
  \noLine
  \UnaryInfC{$\etw{[a, b, d]}{\cof{b}}{[\ihsv{abd}{a}, \ihsv{bd}{b},
  \ihsv{d}{d}]}$}
  \RightLabel{($\cof{b}$)}
  \UnaryInfC{$\etw{[a, b, d]}{\phw{\phl_2}{\phl_2 \lfresh abd, \phl_2 \gfresh^2
  bd}}{[\ihsv{abd \phl_2}{a}, \ihsv{bd \phl_2}{\phl_2}, \ihsv{d \phl_2}{d}]}$}
  \RightLabel{($\lcirc$)}
  \BinaryInfC{$\etw{[a, b, d]}{\phw{\phl_1 d \phl_2}{\phl_1 \lfresh abd, \phl_1
  \gfresh^1 abd, \phl_2 \lfresh \phl_1 bd, \phl_2 \gfresh^2 bd \phl_1
  bdbd}}{[\ihsv{abd \phl_1 \phl_1 bd \phl_2}{\phl_1}, \ihsv{bd \phl_1 bd
  \phl_2}{\phl_2}, \ihsv{d \phl_1 d \phl_2}{d}]}$}
  \doubleLine
  \dashedLine
  \UnaryInfC{$\etw{[a, b, d]}{\phw{\phl_1 d \phl_2}{\phl_1 \lfresh abd, \phl_1
  \gfresh^1 abd, \phl_2 \lfresh \phl_1 bd, \phl_2 \gfresh^2 bd
  \phl_1}}{[\ihsv{abd \phl_1 \phl_2}{\phl_1}, \ihsv{bd \phl_1 \phl_2}{\phl_2},
  \ihsv{d \phl_1 \phl_2}{d}]}$}
  \RightLabel{($\ldmd$)}
  \UnaryInfC{$\etw{[a, b]}{\phw{\phl_1 \phl_3 \phl_2}{\phl_3 \lfresh ab, \phl_1
  \lfresh ab \phl_3, \phl_1 \gfresh^1 ab \phl_3, \phl_2 \lfresh \phl_1 b \phl_3,
  \phl_2 \gfresh^2 b \phl_3 \phl_1}}{[\ihsv{ab \phl_3 \phl_1 \phl_2}{\phl_1},
  \ihsv{b \phl_3 \phl_1 \phl_2}{\phl_2}]}$}
  \RightLabel{($\lcirc$)}
  \BinaryInfC{$\etw{[a, b]}{\phw{b \phl_1 \phl_3 \phl_2}{\phl_3 \lfresh ab,
  \phl_1 \lfresh ab \phl_3, \phl_1 \gfresh^1 abab \phl_3, \phl_2 \lfresh \phl_1
  b \phl_3, \phl_2 \gfresh^2 bb \phl_3 \phl_1}}{[\ihsv{abab \phl_3 \phl_1
  \phl_2}{\phl_1}, \ihsv{bb \phl_3 \phl_1 \phl_2}{\phl_2}]}$}
  \doubleLine
  \dashedLine
  \UnaryInfC{$\etw{[a, b]}{\phw{b \phl_1 \phl_3 \phl_2}{\phl_3 \lfresh ab,
  \phl_1 \lfresh ab \phl_3, \phl_1 \gfresh^1 ab \phl_3, \phl_2 \lfresh \phl_1 b
  \phl_3, \phl_2 \gfresh^2 b \phl_3 \phl_1}}{[\ihsv{ab \phl_3 \phl_1
  \phl_2}{\phl_1}, \ihsv{b \phl_3 \phl_1 \phl_2}{\phl_2}]}$}
  \dashedLine
  \UnaryInfC{$\mathcal{L}_2$}
  \DisplayProof

  \vspace*{10pt}

  \AxiomC{A}
  \noLine
  \UnaryInfC{$\vdots$}
  \noLine
  \UnaryInfC{$\etw{[a]}{a}{[\ihsv{a}{a}]}$}
  \RightLabel{($a$)}
  \UnaryInfC{$\etw{[a]}{\phw{a}{}}{[\ihsv{a}{a}]}$}
  \dashedLine
  \UnaryInfC{$\mathcal{L}_1$}
  \DisplayProof
 \end{minipage}
\end{sideways}
}
 \caption{First half of \lngc\ for $\npexp{n}{n \npexp{m}{m \pexp{l}{l}{m} m
 \npexp{l}{\cof{n} l \cof{m}} } }$\label{ex:apndlngc1}}
\end{figure}

\begin{figure}[p]
{\small
 \vspace*{70pt}
\begin{sideways}
 \begin{minipage}{\textheight}\centering
  \AxiomC{$\mathcal{L}_1$}
  \AxiomC{B}
  \noLine
  \UnaryInfC{$\vdots$}
  \noLine
  \UnaryInfC{$\etw{[a, b]}{b}{[\ihsv{ab}{a}, \ihsv{b}{b}]}$}
  \RightLabel{($b$)}
  \UnaryInfC{$\etw{[a, b]}{\phw{b}{}}{[\ihsv{ab}{a}, \ihsv{b}{b}]}$}
  \AxiomC{C}
  \noLine
  \UnaryInfC{$\vdots$}
  \noLine
  \UnaryInfC{$\etw{[a, b, c]}{c}{[\ihsv{abc}{a}, \ihsv{bc}{c}, \ihsv{cb}{b}]}$}
  \RightLabel{($c$)}
  \UnaryInfC{$\etw{[a, b, c]}{\phw{c}{}}{[\ihsv{abc}{a}, \ihsv{bc}{c},
  \ihsv{cb}{b}]}$}
  \RightLabel{($\ldmd$)}
  \UnaryInfC{$\etw{[a, b]}{\phw{\phl_4}{\phl_4 \lfresh ab}}{[\ihsv{ab
  \phl_4}{a}, \ihsv{b \phl_4}{\phl_4}]}$}
  \RightLabel{($\lcirc$)}
  \BinaryInfC{$\etw{[a, b]}{\phw{b \phl_4}{\phl_4 \lfresh ab}}{[\ihsv{abab
  \phl_4}{a}, \ihsv{bb \phl_4}{\phl_4}]}$}
  \doubleLine
  \dashedLine
  \UnaryInfC{$\etw{[a, b]}{\phw{b \phl_4}{\phl_4 \lfresh ab}}{[\ihsv{ab
  \phl_4}{a}, \ihsv{b \phl_4}{\phl_4}]}$}
  \AxiomC{$\mathcal{L}_2$}
  \RightLabel{($\lcirc$)}
  \BinaryInfC{$\etw{[a, b]}{\phw{b \phl_4 \phl_4 \phl_1 \phl_3 \phl_2}{\phl_4
  \lfresh ab, \phl_3 \lfresh a \phl_4, \phl_1 \lfresh a \phl_4 \phl_3, \phl_1
  \gfresh^1 ab \phl_4 ab \phl_3, \phl_2 \lfresh \phl_1 \phl_4 \phl_3, \phl_2
  \gfresh^2 b \phl_4 b \phl_3 \phl_1}}{[\ihsv{ab \phl_4 a \phl_4 \phl_3 \phl_1
  \phl_2}{\phl_1}, \ihsv{b \phl_4 \phl_4 \phl_3 \phl_1 \phl_2}{\phl_2}]}$}
  \doubleLine
  \dashedLine
  \UnaryInfC{$\etw{[a, b]}{\phw{b \phl_4 \phl_4 \phl_1 \phl_3 \phl_2}{\phl_4
  \lfresh ab, \phl_3 \lfresh a \phl_4, \phl_1 \lfresh a \phl_4 \phl_3, \phl_1
  \gfresh^1 ab \phl_4 \phl_3, \phl_2 \lfresh \phl_1 \phl_4 \phl_3, \phl_2
  \gfresh^2 b \phl_4 \phl_3 \phl_1}}{[\ihsv{ab \phl_4 \phl_3 \phl_1
  \phl_2}{\phl_1}, \ihsv{b \phl_4 \phl_3 \phl_1 \phl_2}{\phl_2}]}$}
  \RightLabel{($\ldmd$)}
  \UnaryInfC{$\etw{[a]}{\phw{\phl_5 \phl_4 \phl_4 \phl_1 \phl_3 \phl_2}{\phl_5
  \lfresh a, \phl_4 \lfresh a \phl_5, \phl_3 \lfresh a \phl_4, \phl_1 \lfresh a
  \phl_4 \phl_3, \phl_1 \gfresh^1 a \phl_5 \phl_4 \phl_3, \phl_2 \lfresh \phl_1
  \phl_4 \phl_3, \phl_2 \gfresh^2 \phl_5 \phl_4 \phl_3 \phl_1}}{[\ihsv{a \phl_5
  \phl_4 \phl_3 \phl_1 \phl_2}{\phl_1}]}$}
  \RightLabel{($\lcirc$)}
  \BinaryInfC{$\etw{[a]}{\phw{a \phl_5 \phl_4 \phl_4 \phl_1 \phl_3
  \phl_2}{\phl_5 \lfresh a, \phl_4 \lfresh a \phl_5, \phl_3 \lfresh a \phl_4,
  \phl_1 \lfresh a \phl_4 \phl_3, \phl_1 \gfresh^1 aa \phl_5 \phl_4 \phl_3,
  \phl_2 \lfresh \phl_1 \phl_4 \phl_3, \phl_2 \gfresh^2 \phl_5 \phl_4 \phl_3
  \phl_1}}{[\ihsv{aa \phl_5 \phl_4 \phl_3 \phl_1 \phl_2}{\phl_1}]}$}
  \doubleLine
  \dashedLine
  \UnaryInfC{$\etw{[a]}{\phw{a \phl_5 \phl_4 \phl_4 \phl_1 \phl_3 \phl_2}{\phl_5
  \lfresh a, \phl_4 \lfresh a \phl_5, \phl_3 \lfresh a \phl_4, \phl_1 \lfresh a
  \phl_4 \phl_3, \phl_1 \gfresh^1 a \phl_5 \phl_4 \phl_3, \phl_2 \lfresh \phl_1
  \phl_4 \phl_3, \phl_2 \gfresh^2 \phl_5 \phl_4 \phl_3 \phl_1}}{[\ihsv{a \phl_5
  \phl_4 \phl_3 \phl_1 \phl_2}{\phl_1}]}$}
  \RightLabel{($\ldmd$)}
  \UnaryInfC{$\etw{[]}{\phw{\phl \phl_5 \phl_4 \phl_4 \phl_1 \phl_3
  \phl_2}{\phl_5 \lfresh \phl, \phl_4 \lfresh \phl \phl_5, \phl_3 \lfresh \phl
  \phl_4, \phl_1 \lfresh \phl \phl_4 \phl_3, \phl_1 \gfresh^1 \phl \phl_5 \phl_4
  \phl_3, \phl_2 \lfresh \phl_1 \phl_4 \phl_3, \phl_2 \gfresh^2 \phl_5 \phl_4
  \phl_3 \phl_1}}{[]}$}
  \doubleLine
  \dashedLine
  \UnaryInfC{$\etw{[]}{\phw{\phl \phl_5 \phl_4 \phl_4 \phl_1 \phl_3 \phl_2}{\phl
  \not= \phl_1, \phl \not= \phl_3, \phl \not= \phl_4, \phl \not= \phl_5, \phl_1
  \not= \phl_2, \phl_1 \not= \phl_3, \phl_1 \not= \phl_4, \phl_1 \not= \phl_5,
  \phl_2 \not= \phl_3, \phl_2 \not= \phl_4, \phl_2 \not= \phl_5, \phl_3 \not=
  \phl_4, \phl_4 \not= \phl_5}}{[]}$}
  \DisplayProof
 \end{minipage}
\end{sideways}
}
 \caption{Second half of \lngc\ for $\npexp{n}{n \npexp{m}{m \pexp{l}{l}{m} m
 \npexp{l}{\cof{n} l \cof{m}} } }$\label{ex:apndlngc2}}
\end{figure}
The languages in-contexts and automata in-contexts for base cases are considered
as follows:
Note that, as the example does not possess unions nor Kleene stars, we can
denote the ``real'' extant chronicles in the post-contexts (hence we do so
below, instead of showing the same post-contests as NREs in-contexts and
automata in-contexts).
\begin{enumerate}[A:]
 \item For the NRE in-contexts $\etw{[a]}{a}{[\ihsv{a}{a}]}$, the language
       in-contexts and the \cda\ are
       \begin{align*}
	\etw{[a]}{\left\{a\right\}}{[\ihsv{a}{a}]}
       \end{align*}
       \begin{minipage}[c]{\linewidth}\centering
	\includegraphics[scale=.4]{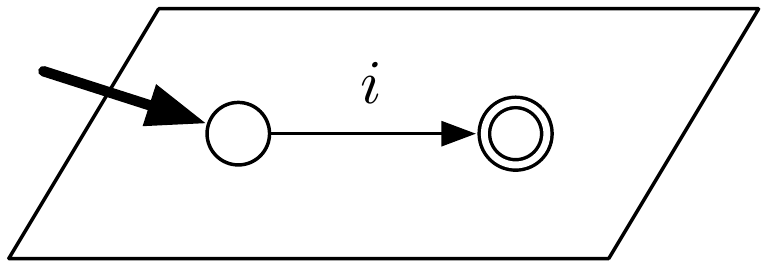}
       \end{minipage}
       \\[.5pc]
       In this case, since the length of $[a]$ is $1$, the automaton is on a
       first layer, hence each state has only one register.
       In the picture, $i$ should be $1$ and the initial assignment of the
       register is: $1 \mapsto a$ with the natural extant chronicle $[a]$.
       The run is: $\tuple{q_0, a, [\ihsv{a}{a}]} \xrightarrow{1} \tuple{q_1,
       \epsilon, [\ihsv{a}{a}]}$.
       So, the \cda\ recognises the language in-contexts.

 \item For the NRE in-contexts $\etw{[a, b]}{b}{[\ihsv{ab}{a}, \ihsv{b}{b}]}$,
       the language in-contexts and the \cda\ are
       \begin{align*}
	\etw{[a, b]}{\left\{b\right\}}{[\ihsv{ab}{a}, \ihsv{b}{b}]}
       \end{align*}
       \begin{minipage}[c]{\linewidth}\centering
	\includegraphics[scale=.4]{hdexn}
       \end{minipage}
       \\[.5pc]
       In this case, since the length of $[a, b]$ is $2$, the automaton is on a
       second layer, hence each state has two registers.
       In the picture, $i$ should be $2$ and the initial assignment of the
       registers is: $1 \mapsto a$ and $2 \mapsto b$ with the natural extant
       chronicle $[ab, b]$.
       The run is: $\tuple{q_0, b, [\ihsv{ab}{a}, \ihsv{b}{b}]} \xrightarrow{2}
       \tuple{q_1, \epsilon, [\ihsv{ab}{a}, \ihsv{b}{b}]}$.
       So, the \cda\ recognises the language in-contexts.

 \item For the NRE in-contexts $\etw{[a, b, c]}{c}{[\ihsv{abc}{a}, \ihsv{bc}{c},
       \ihsv{cb}{b}]}$, the language in-contexts and the \cda\ are
       \begin{align*}
	\etw{[a, b, c]}{\left\{c\right\}}{[\ihsv{abc}{a}, \ihsv{bc}{c},
	\ihsv{cb}{b}]}
       \end{align*}
       \begin{minipage}[c]{\linewidth}\centering
	\includegraphics[scale=.4]{hdexn}
       \end{minipage}
       \\[.5pc]
       In this case, since the length of $[a, b, c]$ is $3$, the automaton is on
       a third layer, hence each state has three registers.
       In the picture, $i$ should be $3$ and the initial assignment of the
       registers is: $1 \mapsto a$, $2 \mapsto b$ and $3 \mapsto c$ with the
       natural extant chronicle $[abc, bc, c]$.
       The run is: $\tuple{q_0, c, [\ihsv{abc}{a}, \ihsv{bc}{b}, \ihsv{c}{c}]}
       \xrightarrow{3} \tuple{q_1, \epsilon, [\ihsv{abc}{a}, \ihsv{bc}{b},
       \ihsv{c}{c}]}$.
       So, the \cda\ recognises the language in-contexts.

 \item For the NRE in-contexts $\etw{[a, b, d]}{\cof{a}}{[\ihsv{abd}{a},
       \ihsv{bd}{b}, \ihsv{d}{d}]}$, the language in-contexts and the \cda\ are
       \begin{align*}
	\etw{[a, b, d]}{\left\{n \mid n \lfresh abd, (n \gfresh^1 abd)
	\right\}}{[\ihsv{abdn}{n}, \ihsv{bdn}{b}, \ihsv{dn}{d}]}
       \end{align*}
       \begin{minipage}{\linewidth}\centering
	\includegraphics[scale=.4]{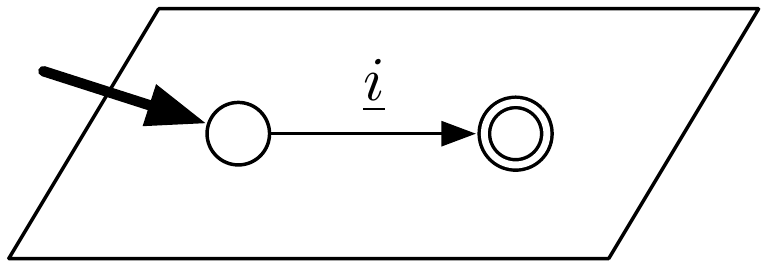}
       \end{minipage}
       \\[.5pc]
       In this case, since the length of $[a, b, d]$ is $3$, the automaton is on
       a third layer, hence each state has three registers.
       In the picture, $i$ should be $1$ and the initial assignment of the
       registers is: $1 \mapsto a$, $2 \mapsto b$ and $3 \mapsto d$ with the
       natural extant chronicle $[abd, bd, d]$.
       The run is:
       \[
	\tuple{q_0, n, [\ihsv{abd}{a}, \ihsv{bd}{b},
       \ihsv{d}{d}]} \xrightarrow{\underline{1}} \tuple{q_1, \epsilon,
       [\ihsv{abdn}{n}, \ihsv{bdn}{b}, \ihsv{dn}{d}]}
       \]
       where any name $n \lfresh abd$.
       So, the \cda\ recognises the language in-contexts.
       One may feel that there seems no difference between $\lfresh$ and
       $\gfresh^1$.
       However, the difference appear when we concatenate it with other
       languages: recall rule ($\lcirc$) in Fig.~\ref{fig:languagecalculus}.
       That is, the natural chronicle is used as a bookmark for the later use
       here, for permutations and concatenations.

 \item For the NRE in-contexts $\etw{[a, b, d]}{d}{[\ihsv{abd}{a}, \ihsv{b}{bd},
       \ihsv{d}{d}]}$, the language in-contexts and the \cda\ are
       \begin{align*}
	\etw{[a, b, d]}{\left\{d\right\}}{[\ihsv{abd}{a}, \ihsv{bd}{b},
	\ihsv{d}{d}]}
       \end{align*}
       \begin{minipage}{\linewidth}\centering
	\includegraphics[scale=.4]{hdexn}
       \end{minipage}
       \\[.5pc]
       In this case, since the length of $[a, b, d]$ is $3$, the automaton is on
       a third layer, hence each state has three registers.
       In this picture, $i$ should be $3$ and the initial assignment of the
       registers is: $1 \mapsto a$, $2 \mapsto b$ and $3 \mapsto d$ with the
       natural extant chronicle $[abd, bd, d]$.
       The run is: $\tuple{q_0, d, [\ihsv{abd}{a}, \ihsv{bd}{b}, \ihsv{d}{d}]}
       \xrightarrow{3} \tuple{q_1, \epsilon, [\ihsv{abd}{a}, \ihsv{bd}{b},
       \ihsv{d}{d}]}$.
       So, the \cda\ recognises the language in-contexts.

 \item For the NRE in-contexts $\etw{[a, b, d]}{\cof{b}}{[\ihsv{abd}{a},
       \ihsv{bd}{b}, \ihsv{d}{d}]}$, the language in-contexts and the \cda\ are
       \begin{align*}
	\etw{[a, b, d]}{\left\{n \mid n \lfresh abd, (n \gfresh^2
	bd)\right\}}{[\ihsv{abdn}{a}, \ihsv{bdn}{n}, \ihsv{dn}{d}]}
       \end{align*}
       \begin{minipage}{\linewidth}\centering
	\includegraphics[scale=.4]{hdexcofn}
       \end{minipage}
       \\[.5pc]
       In this case, since the length of $[a, b, d]$ is $3$, the automaton is on
       a third layer, hence each state has three registers.
       In this picture, $i$ should be $2$ and the initial assignment of the
       registers is: $1 \mapsto a$, $2 \mapsto b$ and $3 \mapsto d$ with the
       natural extant chronicle $[abd, bd, d]$.
       The run is: $\tuple{q_0, n, [\ihsv{abd}{a}, \ihsv{bd}{b}, \ihsv{d}{d}]}
       \xrightarrow{\underline{2}} \tuple{q_1, \epsilon, [\ihsv{abdn}{a},
       \ihsv{bdn}{n}, \ihsv{dn}{d}]}$, where any name $n \lfresh abd$.
       Therefore, the \cda\ recognises the language in-contexts.
       As $\cof{b}$ is a relative global fresh transition with respect to the
       second chronicle, we take $\star_2 \gfresh^2 bd$ by means of the natural
       chronicle for the register $2$ as a bookmark.
\end{enumerate}

The languages in-contexts and the automata in-contests for inductive steps are
as follows (note that we simplify chronicles or remove some
$\epsilon$-transitions):
\begin{enumerate}[(1)]
 \item For the NRE $\etw{[a, b, d]}{\cof{a}d}{[\ihsv{abd}{a}, \ihsv{bd}{b},
       \ihsv{d}{d}]}$, the language in-contexts and the \cda\ are
       \begin{align*}
	\etw{C}{\langof(\ne)}{E} = \etw{[a, b, d]}{\left\{nd \mid n \lfresh abd,
	(n \gfresh^1 abd) \right\}}{[\ihsv{abdn}{n}, \ihsv{bdn}{b},
	\ihsv{dn}{d}]}
       \end{align*}
       \begin{minipage}{\linewidth}\centering
	\includegraphics[scale=.4]{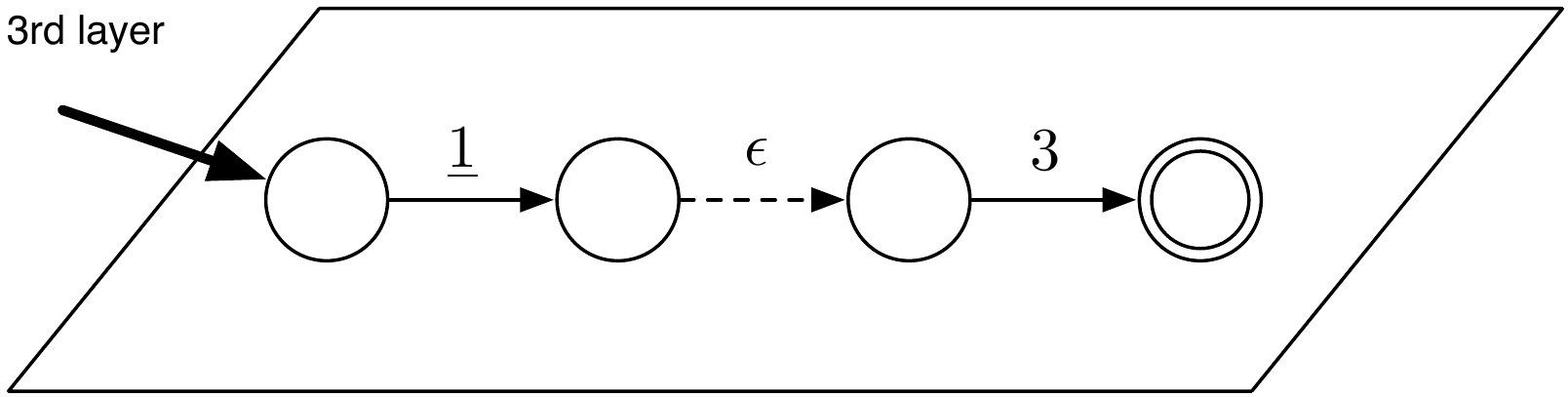}
       \end{minipage}
       \\[.5pc]
       In the language, we have simplified chronicles.
       This step is a concatenation of cases D and E on a third layer.
       The initial assignment of the registers is: $1 \mapsto a$, $2 \mapsto b$
       and $3 \mapsto d$ with the natural extant chronicle $[abd, bd, d]$.
       The run is:
       \begin{align*}
	\tuple{q_0, nd, [\ihsv{abd}{a}, \ihsv{bd}{b}, \ihsv{d}{d}]}
	&\xrightarrow{\cof{1}} \tuple{q_1, d, [\ihsv{abdn}{n}, \ihsv{bdn}{b},
	\ihsv{dn}{d}]}\\
	&\xrightarrow{3} \tuple{q_2, \epsilon, [\ihsv{abdn}{n}, \ihsv{bdn}{b},
	\ihsv{dn}{d}]}
       \end{align*}
       where any name $n \lfresh abd$.
       So, it is easy to see that the \cda\ accepts the language in-contexts.

 \item For the NRE $\etw{[a, b, d]}{\cof{a} d \cof{b}}{[\ihsv{abd}{a},
       \ihsv{bd}{b}, \ihsv{d}{d}]}$, the language in-contexts and the \cda\ are
       \begin{align*}
	C:& [a, b, d]\\
	\langof(\ne):& \left\{ndm \mid n \lfresh abd, (n \gfresh^1 abd), m
	\lfresh nbd, (m \gfresh^2 bdn) \right\}\\
	E:& [\ihsv{abdnm}{n}, \ihsv{bdnm}{m}, \ihsv{dnm}{d}]
       \end{align*}
       \begin{minipage}{\linewidth}\centering
	\includegraphics[scale=.4]{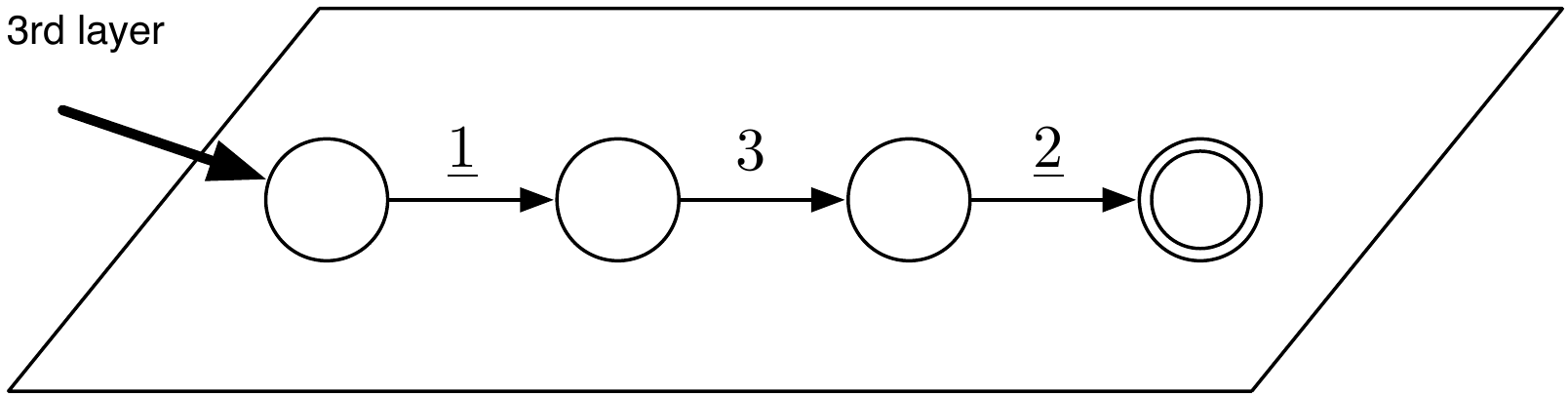}
       \end{minipage}
       \\[.5pc]
       This step is the concatenation of cases (1) and F on a third layer.
       The initial assignment of the registers is: $1 \mapsto a$, $2 \mapsto b$
       and $3 \mapsto d$ with the natural extant chronicle $[abd, bd, d]$.
       The run is:
       \begin{align*}
	\tuple{q_0, ndm, [\ihsv{abd}{a}, \ihsv{bd}{b}, \ihsv{d}{d}]}
	&\xrightarrow{\cof{1}} \tuple{q_1, dm, [\ihsv{abdn}{n}, \ihsv{bdn}{b},
	\ihsv{dn}{d}]}\\
	&\xrightarrow{3} \tuple{q_2, m, [\ihsv{abdn}{n}, \ihsv{bdn}{b},
	\ihsv{dn}{d}]}\\
	&\xrightarrow{\cof{2}} \tuple{q_3, \epsilon, [\ihsv{abdnm}{n},
	\ihsv{bdnm}{m}, \ihsv{dnm}{d}]}
       \end{align*}
       where any names $n \lfresh abd$ and $m \lfresh bdn$.
       Hence, the \cda\ accepts the language in-contexts.
       Notice that, when we concatenate the languages, the latter words are
       permute $a$ with $\phl_1$ and updated chronicles (and relative-global
       freshness).
       One may find that $m$ can be $a$, because of $m \lfresh bdn$.
       The fact reflects the relative global freshness (with respect to the
       chronicle $2$).

 \item For the NRE $\etw{[a, b]}{\npexp{l}{\cof{a} l \cof{b}}}{[\ihsv{ab}{a},
       \ihsv{b}{b}]}$, the language in-contexts and the \cda\ are
       \begin{align*}
	C:& [a, b]\\
	\langof(\ne):& \left\{nlm \mid l \lfresh ab, n \lfresh abl, (n \gfresh^1
	abl), m \lfresh nbl, (m \gfresh^2 bln) \right\}\\
	E:& [\ihsv{ablnm}{n}, \ihsv{blnm}{m}]
       \end{align*}
       \begin{minipage}{\linewidth}\centering
	\includegraphics[scale=.4]{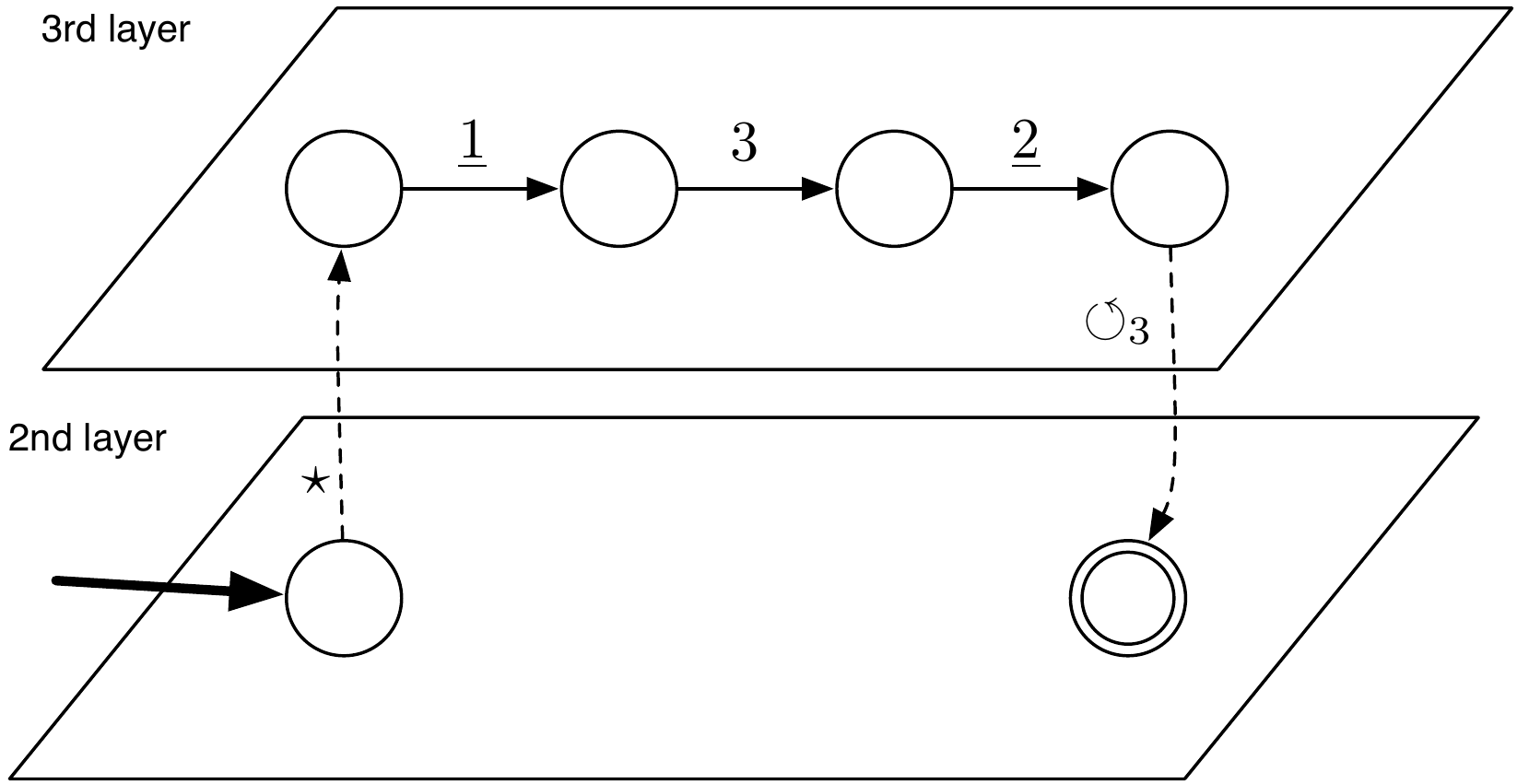}
       \end{minipage}
       \\[.5pc]
       This step abstracts case (2).
       The initial assignment of the registers is: $1 \mapsto a$ and $2 \mapsto
       b$ with the natural extant chronicle $[ab, b]$.
       The run is:
       \begin{align*}
	\tuple{q_0, nlm, [\ihsv{ab}{a}, \ihsv{b}{b}]} &\xrightarrow{\star}
	\tuple{q_1, nlm, [\ihsv{abl}{a}, \ihsv{bl}{b}, \ihsv{l}{l}]}\\
	&\xrightarrow{\cof{1}} \tuple{q_2, lm, [\ihsv{abln}{n}, \ihsv{bln}{b},
	\ihsv{ln}{l}]}\\
	&\xrightarrow{3} \tuple{q_3, m, [\ihsv{abln}{n}, \ihsv{bln}{b},
	\ihsv{ln}{l}]}\\
	&\xrightarrow{\cof{2}} \tuple{q_4, \epsilon, [\ihsv{ablnm}{n},
	\ihsv{blnm}{m}, \ihsv{lnm}{l}]}\\
	&\xrightarrow{\circlearrowleft{3}} \tuple{q_5, \epsilon,
	[\ihsv{ablnm}{n}, \ihsv{blnm}{m}]}
       \end{align*}
       where any names $l \lfresh ab$, $n \lfresh abl$ and $m \lfresh bln$.
       Hence, the \cda\ accepts the language in-contexts.

 \item For the NRE $\etw{[a, b]}{b \npexp{l}{\cof{a} l \cof{b}}}{[\ihsv{ab}{a},
       \ihsv{b}{b}]}$, the language in-contexts and the \cda\ are
       \begin{align*}
	C:& [a, b]\\
	\langof(\ne):& \left\{bnlm \mid l \lfresh ab, n \lfresh abl, (n
	\gfresh^1 abl), m \lfresh nbl, (m \gfresh^2 bln) \right\}\\
	E:& [\ihsv{ablnm}{n}, \ihsv{blnm}{m}]
       \end{align*}
       \begin{minipage}{\linewidth}\centering
	\includegraphics[scale=.4]{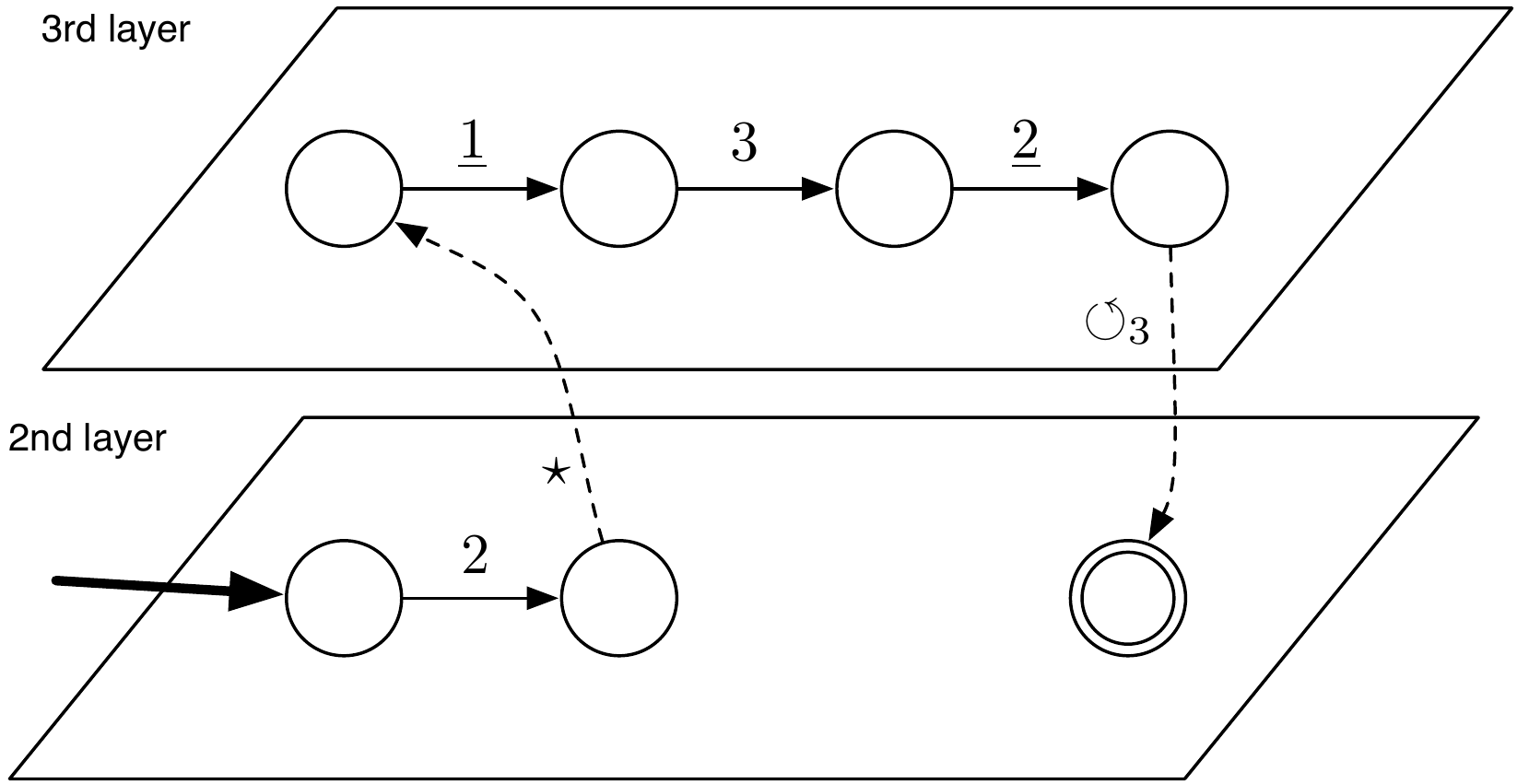}
       \end{minipage}
       \\[.5pc]
       This step concatenates case B with case (3).
       The initial assignment of the registers is: $1 \mapsto a$ and $2 \mapsto
       b$ with the natural extant chronicle $[ab, b]$.
       The run is:
       \begin{align*}
	\tuple{q_0, bnlm, [\ihsv{ab}{a}, \ihsv{b}{b}]} &\xrightarrow{2}
	\tuple{q_1, nlm, [\ihsv{ab}{a}, \ihsv{b}{b}]}\\
	&\xrightarrow{\star} \tuple{q_2, nlm, [\ihsv{abl}{a}, \ihsv{bl}{b},
	\ihsv{l}{l}]}\\
	&\xrightarrow{\cof{1}} \tuple{q_3, lm, [\ihsv{abln}{n}, \ihsv{bln}{b},
	\ihsv{ln}{l}]}\\
	&\xrightarrow{3} \tuple{q_4, m, [\ihsv{abln}{n}, \ihsv{bln}{b},
	\ihsv{ln}{l}]}\\
	&\xrightarrow{\cof{2}} \tuple{q_5, \epsilon, [\ihsv{ablnm}{n},
	\ihsv{blnm}{m}, \ihsv{lnm}{l}]}\\
	&\xrightarrow{\pclose{3}} \tuple{q_6, \epsilon, [\ihsv{ablnm}{n},
	\ihsv{blnm}{m}]}
       \end{align*}
       where any names $l \lfresh ab$, $n \lfresh abl$ and $m \lfresh bln$.
       Hence, the \cda\ accepts the language in-contexts.

 \item For the NRE $\etw{[a, b]}{\pexp{l}{l}{b}}{[\ihsv{ab}{a}, \ihsv{b}{b}]}$,
       the language in-contexts and the \cda\ are
       \begin{align*}
	\etw{C}{\langof(\ne)}{E} = \etw{[a, b]}{\left\{n \mid n \lfresh
	ab\right\}}{[\ihsv{abn}{a}, \ihsv{bn}{n}]}
       \end{align*}
       \begin{minipage}{\linewidth}\centering
	\includegraphics[scale=.4]{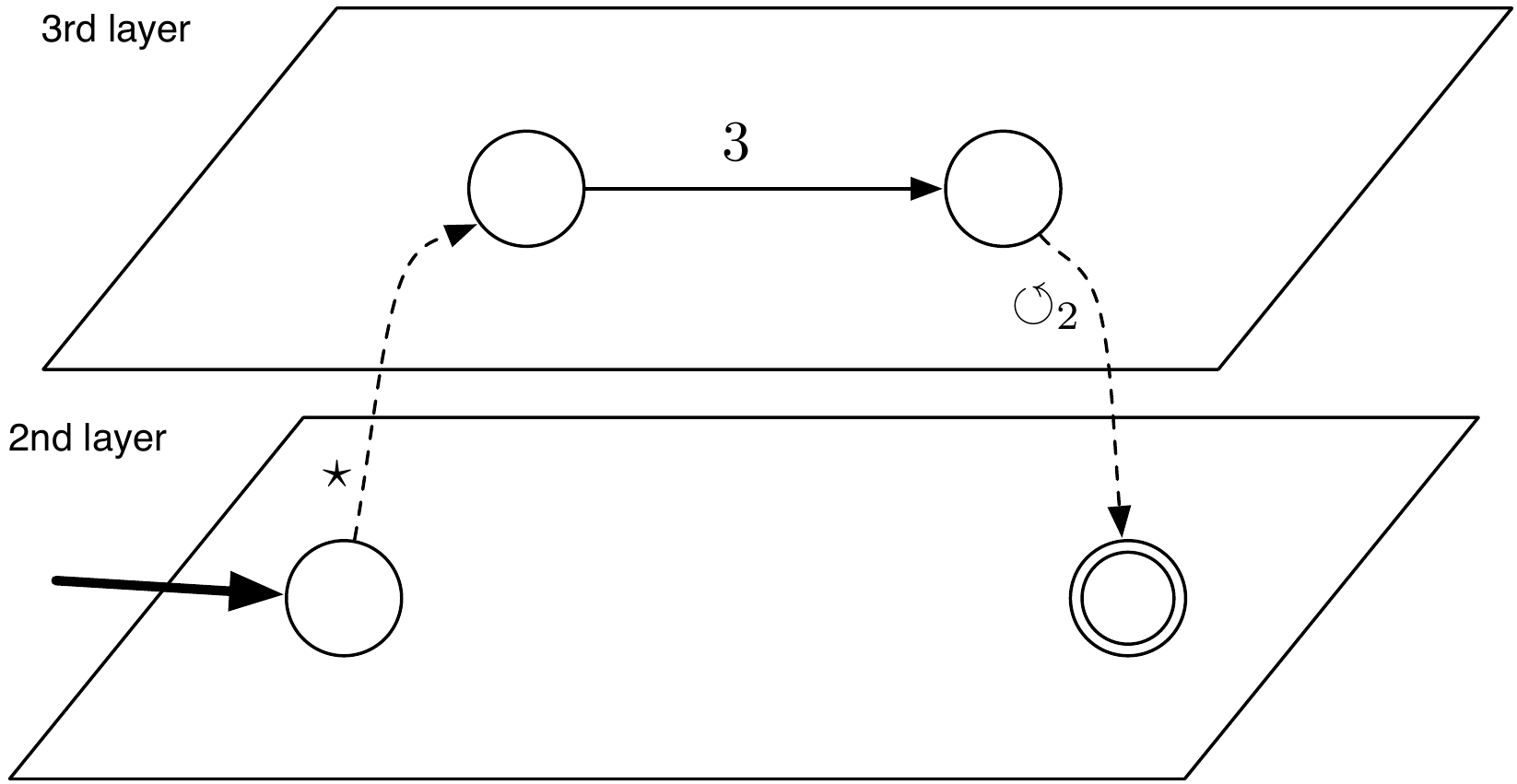}
       \end{minipage}
       \\[.5pc]
       This step abstracts case C.
       Since the post extant chronicle of case C remembers the permutation, the
       transition to the final state labelled with $\pclose{2}$ not with
       $\pclose{3}$.
       Notice that the permutation effect is left in the post extant chronicle,
       see the current value of the register 2.
       The initial assignment of the registers is: $1 \mapsto a$ and $2 \mapsto
       b$ with the natural extant chronicle $[ab, b]$.
       The run is:
       \begin{align*}
	\tuple{q_0, n, [\ihsv{ab}{a}, \ihsv{a}]} &\xrightarrow{\star}
	\tuple{q_1, n, \ihsv{abn}{a}, \ihsv{bn}{b}, \ihsv{n}{n}}\\
	&\xrightarrow{3} \tuple{q_2, \epsilon, [\ihsv{abn}{a}, \ihsv{bn}{b},
	\ihsv{n}{n}]}\\
	&\xrightarrow{\pclose{2}} \tuple{q_3, \epsilon, [\ihsv{abn}{a}, \ihsv{bn}{n}]}
       \end{align*}
       where any name $n \lfresh ab$.
       Therefore, the \cda\ accepts the language in-contexts with keeping the
       permutation action on the post-context.
       But, in the final state, the configuration of the registers turns to be
       on the level of a schematic word: $1 \mapsto a$ and $2 \mapsto \phl_4$
       (not $b$) with the extant chronicle $[\ihsv{ab \phl_4}{a}, \ihsv{b
       \phl_4}{\phl_4}]$.

 \item For the NRE $\etw{[a, b]}{b \pexp{l}{l}{b}}{[\ihsv{ab}{a},
       \ihsv{b}{b}]}$, the language in-contexts and the \cda\ are
       \begin{align*}
	\etw{C}{\langof(\ne)}{E} = \etw{[a,b]}{\left\{bn \mid n \lfresh ab
	\right\}}{[\ihsv{abn}{a}, \ihsv{bn}{n}]}
       \end{align*}
       \begin{minipage}{\linewidth}\centering
	\includegraphics[scale=.4]{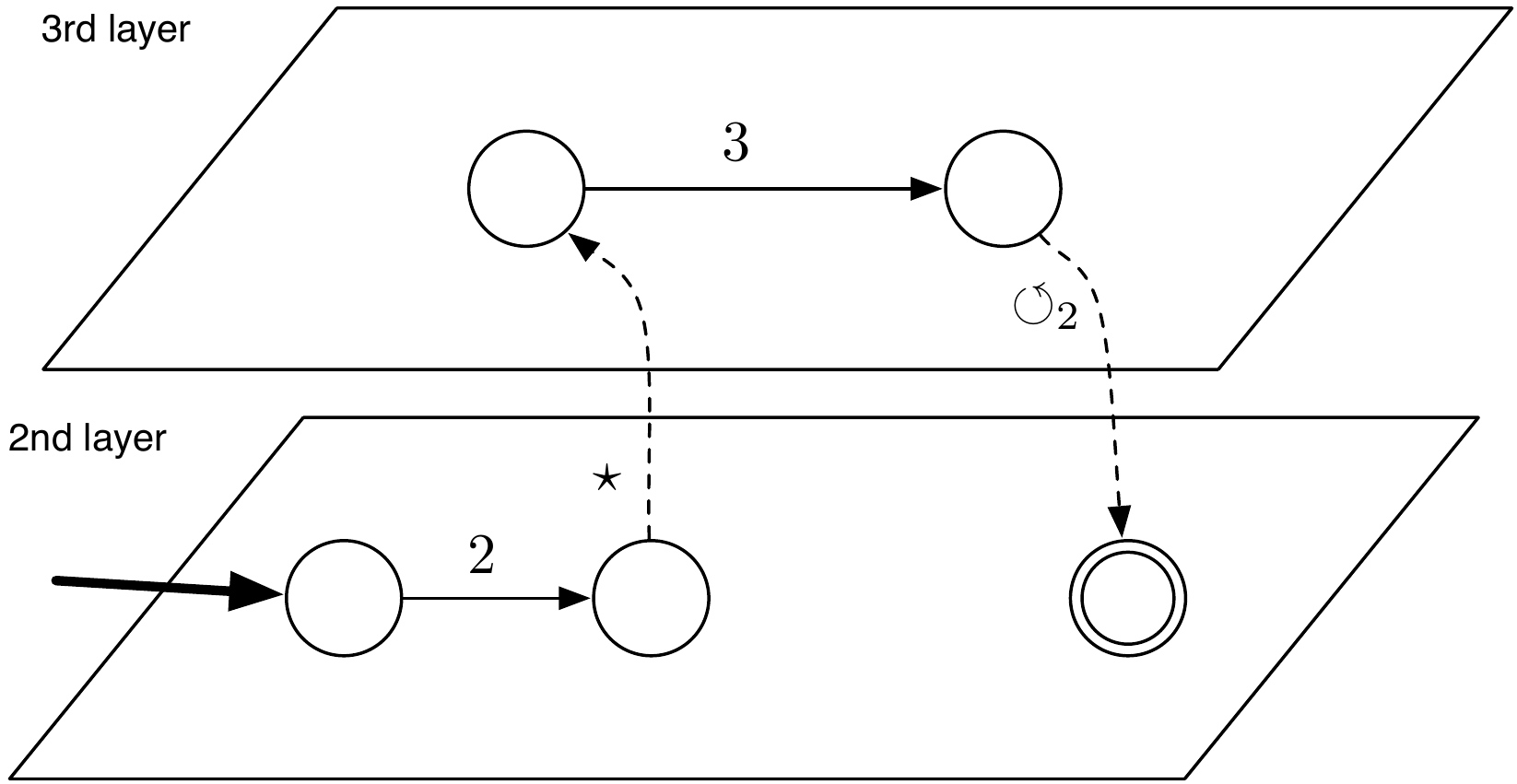}
       \end{minipage}
       \\[.5pc]
       This step concatenates case B with case (5).
       The initial assignment of the registers is: $1 \mapsto a$ and $2 \mapsto
       b$ with the natural extant chronicle $[ab, b]$.
       The run is:
       \begin{align*}
	\tuple{q_0, bn, [\ihsv{ab}{a}, \ihsv{b}{b}]} &\xrightarrow{2}
	\tuple{q_1, n, [\ihsv{ab}{a}, \ihsv{b}{b}]}\\
	&\xrightarrow{\star} \tuple{q_2, n, [\ihsv{abn}{a}, \ihsv{bn}{b},
	\ihsv{n}{n}]}\\
	&\xrightarrow{3} \tuple{q_3, \epsilon, [\ihsv{abn}{a}, \ihsv{bn}{b},
	\ihsv{n}{n}]}\\
	&\xrightarrow{\pclose{2}} \tuple{q_4, \epsilon, [\ihsv{abn}{a},
	\ihsv{bn}{n}]}
       \end{align*}
       Hence, the \cda\ accepts the language in-contexts.
       Note that the permutation action in $\pclose{2}$ is still preserved in
       safe in the post-context as we expect, i.e.~the extant chronicle is
       $[\ihsv{ab \phl_4}{a}, \ihsv{b \phl_4}{\phl_4}]$ on the level of a
       schematic word.

 \item For the NRE $\etw{[a, b]}{b \pexp{l}{l}{b} b \npexp{l}{\cof{a} l
       \cof{b}}}{}$, the language in-contexts and the \cda\ are
       \begin{align*}
	C:& [a, b]\\
	\langof(\ne):& \bigl\{bnnn'lm \mid n \lfresh ab, l \lfresh an, n'
	\lfresh anl, n' \gfresh^1 abnl, m \lfresh nn'l, m \gfresh^2 b nln'
	\bigr\}\\
	E:& [\ihsv{abnln'm}{n'}, \ihsv{bnln'm}{m}]
       \end{align*}
       \begin{minipage}{\linewidth}\centering
	\includegraphics[scale=.4]{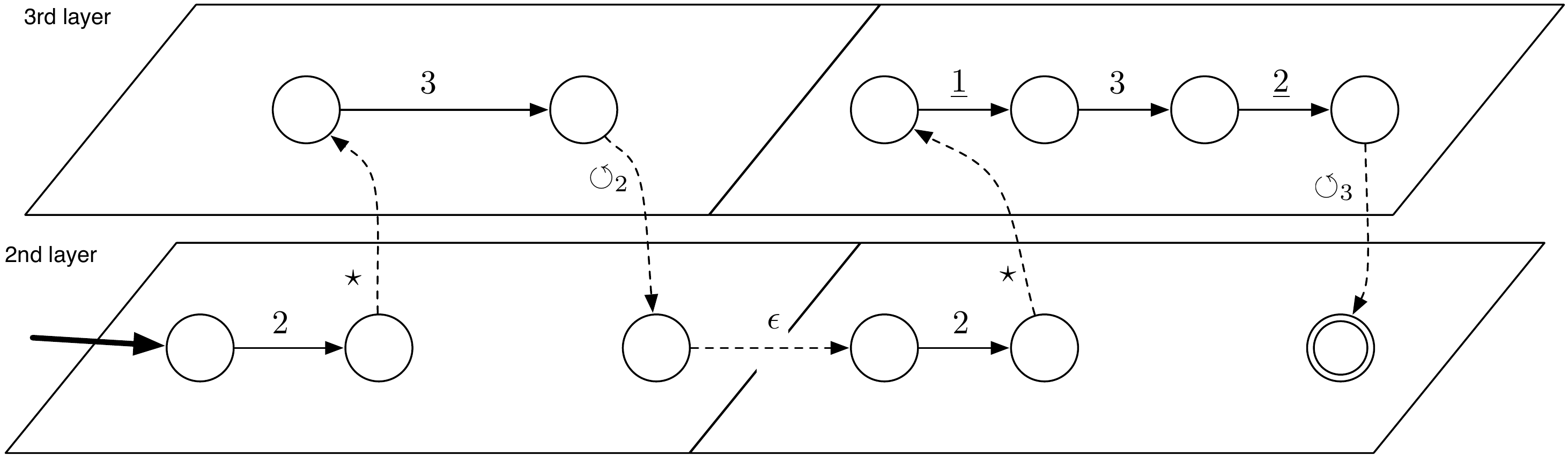}
       \end{minipage}
       \\[.5pc]
       This step concatenates cases (4) and (6).
       The initial assignment of the registers is: $1 \mapsto a$ and $2 \mapsto
       b$ with the natural extant chronicle $[ab, b]$.
       The run is:
       \begin{align*}
	\tuple{q_0, bnnn'lm, [\ihsv{ab}{a}, \ihsv{b}{b}]} &\xrightarrow{2}
	\tuple{q_1, nnn'lm, [\ihsv{ab}{a}, \ihsv{b}{b}]}\\
	&\xrightarrow{\star} \tuple{q_2, nnn'lm, [\ihsv{abn}{a}, \ihsv{bn}{b},
	\ihsv{n}{n}]}\\
	&\xrightarrow{3} \tuple{q_3, nn'lm, [\ihsv{abn}{a}, \ihsv{bn}{b},
	\ihsv{n}{n}]}\\
	&\xrightarrow{\pclose{2}} \tuple{q_4, nn'lm, [\ihsv{abn}{a},
	\ihsv{bn}{n}]}\\
	&\xrightarrow{2} \tuple{q_5, n'lm, [\ihsv{abn}{a}, \ihsv{bn}{n}]}\\
	&\xrightarrow{\star} \tuple{q_6, n'lm, [\ihsv{abnl}{a}, \ihsv{bnl}{n},
	\ihsv{l}{l}]}\\
	&\xrightarrow{\cof{1}} \tuple{q_7, lm, [\ihsv{abnln'}{n'},
	\ihsv{bnln'}{n}, \ihsv{ln'}{l}]}\\
	&\xrightarrow{3} \tuple{q_8, m, [\ihsv{abnln'}{n'}, \ihsv{bnln'}{n},
	\ihsv{ln'}{l}]}\\
	&\xrightarrow{\cof{2}} \tuple{q_9, \epsilon, [\ihsv{abnln'm}{n'},
	\ihsv{bnln'm}{m}, \ihsv{ln'm}{l}]}\\
	&\xrightarrow{\pclose{3}} \tuple{q_{10}, \epsilon, [\ihsv{abnln'm}{n'},
	\ihsv{bnln'm}{m}]}
       \end{align*}
       where any names $n \lfresh ab$, $l \lfresh an$, $n' \lfresh abnl$ and $m
       \lfresh bnln'$.
       The important things are: 1. when we take the $\epsilon$-transition, the
       name of the register $2$ is $n$ ($\phl_4$ on the level of a schematic
       word).
       However, the concatenation of languages rule ($\lcirc$) reflects the fact
       as the permutation $\pi_{[{C}{\rhd}{E_1}]}$ and the relative global
       freshness by appending corresponding chronicles in $E_1$.
       Therefore, the \cda\ accepts the language in-contexts without any
       problem.
       Also, all the information are kept in the post-context in safe again.

 \item For the NRE $\etw{[a]}{\npexp{m}{m \pexp{l}{l}{m} m \npexp{l}{\cof{a} l
       \cof{m}}}}{[\ihsv{a}{a}]}$, the language in-contexts and the \cda\ are
       \begin{align*}
	C:& [a]\\
	\langof(\ne):& \bigl\{mnnn'lm' \mid m \lfresh a, n \lfresh am, l \lfresh
	an, n' \lfresh anl, n' \gfresh^1 amnl, m' \lfresh n'nl, m' \gfresh^2
	mnln' \bigr\}\\
	E:& [\ihsv{amnln'm'}{n'}]
       \end{align*}
       \begin{minipage}{\linewidth}\centering
	\includegraphics[scale=.4]{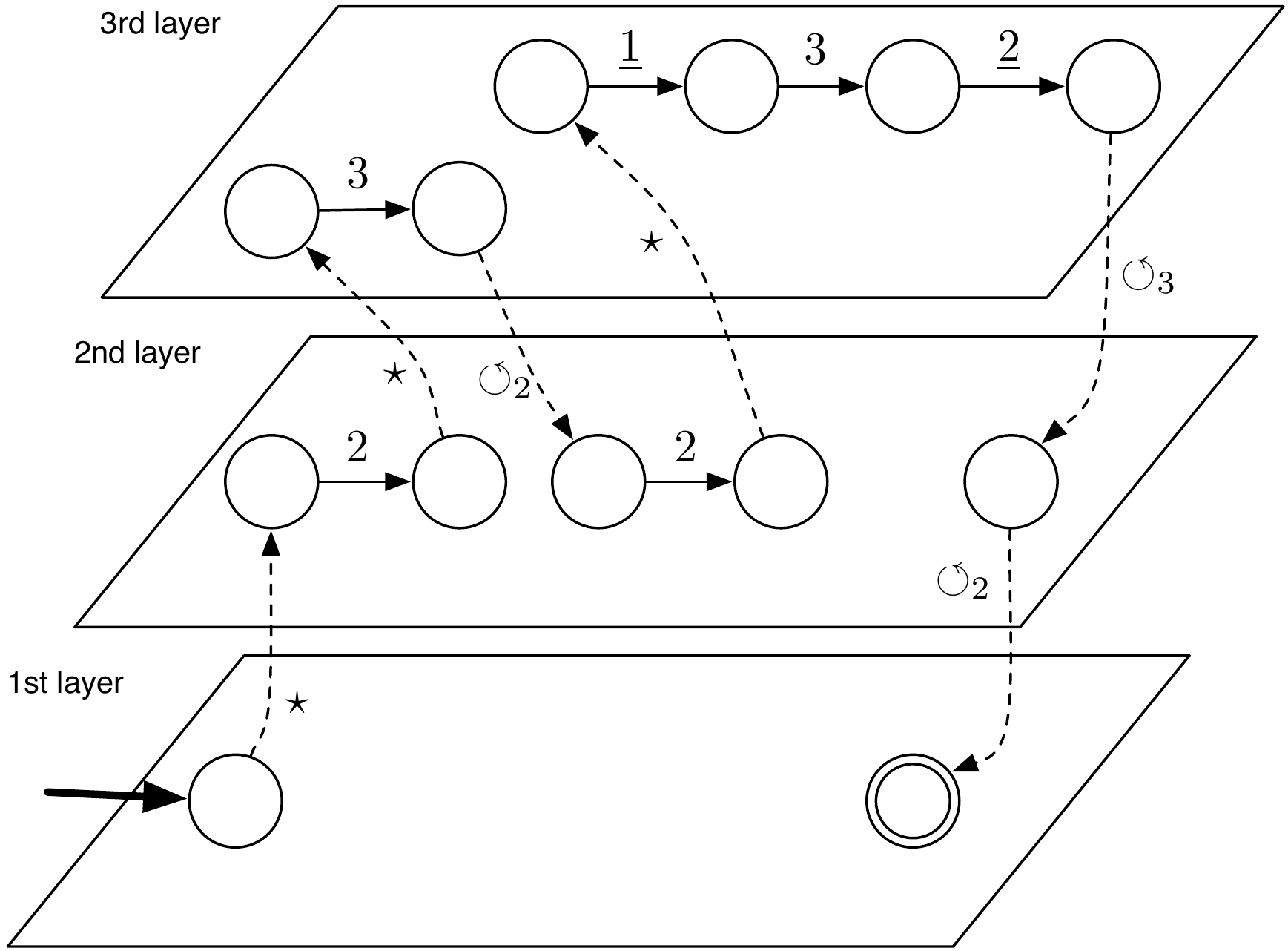}
       \end{minipage}
       \\[.5pc]
       This step abstracts the second register of case (7).
       The initial assignment of the register is: $1 \mapsto a$ with the natural
       extant chronicle $[a]$.
       The run is:
       \begin{align*}
	\tuple{q_0, mnnn'lm', [\ihsv{a}{a}]} &\xrightarrow{\star} \tuple{q_1,
	mnnn'lm', [\ihsv{am}{a}, \ihsv{m}{m}]}\\
	&\xrightarrow{2} \tuple{q_2, nnn'lm', [\ihsv{am}{a}, \ihsv{m}{m}]}\\
	&\xrightarrow{\star} \tuple{q_3, nnn'lm', [\ihsv{amn}{a}, \ihsv{mn}{m},
	\ihsv{n}{n}]}\\
	&\xrightarrow{3} \tuple{q_4, nn'lm', [\ihsv{amn}{a}, \ihsv{mn}{m},
	\ihsv{n}{n}]}\\
	&\xrightarrow{\pclose{2}} \tuple{q_5, nn'lm', [\ihsv{amn}{a},
	\ihsv{mn}{n}]}\\
	&\xrightarrow{2} \tuple{q_6, n'lm', [\ihsv{amn}{a}, \ihsv{mn}{n}]}\\
	&\xrightarrow{\star} \tuple{q_7, n'lm', [\ihsv{amnl}{a}, \ihsv{mnl}{n},
	\ihsv{l}{l}]}\\
	&\xrightarrow{\cof{1}} \tuple{q_8, lm', [\ihsv{amnln'}{n'},
	\ihsv{mnln'}{n}, \ihsv{ln'}{l}]}\\
	&\xrightarrow{3} \tuple{q_9, m', [\ihsv{amnln'}{n'}, \ihsv{mnln'}{n},
	\ihsv{ln'}{l}]}\\
	&\xrightarrow{\cof{2}} \tuple{q_{10}, \epsilon, [\ihsv{amnln'm'}{n'},
	\ihsv{mnln'm'}{m'}, \ihsv{ln'm'}{l}]}\\
	&\xrightarrow{\pclose{3}} \tuple{q_{11}, \epsilon, [\ihsv{amnln'm'}{n'},
	\ihsv{mnln'm'}{m'}]}\\
	&\xrightarrow{\pclose{2}} \tuple{q_{12}, \epsilon,
	[\ihsv{amnln'm'}{n'}]}
       \end{align*}
       where any names $m \lfresh a$, $n \lfresh am$, $l \lfresh an$, $n'
       \lfresh amnl$ and $m' \lfresh mnln'$.
       So, the abstracted $m$ ($\phl_5$ on a schematic word) can be any name
       except $a$.
       Accordingly, we replace $b$ by $\phl_5$.
       Hence, the \cda\ accepts the language in-contexts.

 \item For the NRE $\etw{[a]}{a \npexp{m}{m \pexp{l}{l}{m} m \npexp{\cof{a} l
       \cof{m}}}}{[\ihsv{a}{a}]}$, the language in-contexts and the \cda\ are
       \begin{align*}
	C:& [a]\\
	\langof(\ne):& \bigl\{amnnn'lm' \mid m \lfresh a, n \lfresh am, l
	\lfresh an, n' \lfresh anl, n' \gfresh^1 amnl, m' \lfresh n'nl, m' \gfresh^2 mnln'
	\bigr\}\\
	E:& [\ihsv{amnln'm'}{n'}]
       \end{align*}
       \begin{minipage}{\linewidth}\centering
	\includegraphics[scale=.4]{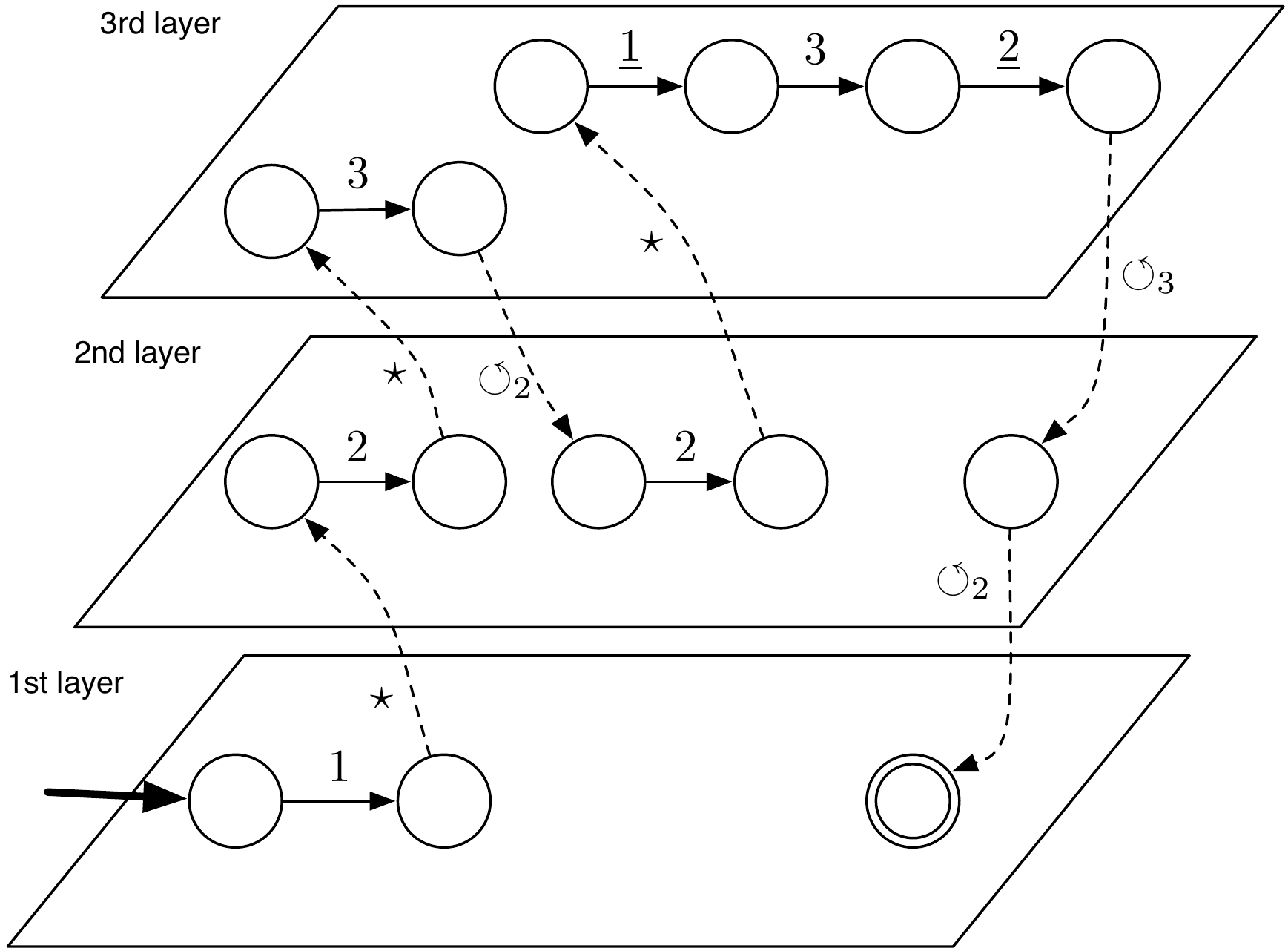}
       \end{minipage}
       \\[.5pc]
       This step concatenates case A with case (8).
       The initial assignment of the register is: $1 \mapsto a$ with the natural
       extant chronicle $[a]$.
       The run is:
       \begin{align*}
	\tuple{q_0, amnnn'lm', [\ihsv{a}{a}]} &\xrightarrow{1} \tuple{q_1,
	mnnn'lm', [\ihsv{a}{a}]}\\
	&\xrightarrow{\star} \tuple{q_2, mnnn'lm', [\ihsv{am}{a},
	\ihsv{m}{m}]}\\
	&\xrightarrow{2} \tuple{q_3, nnn'lm', [\ihsv{am}{a}, \ihsv{m}{m}]}\\
	&\xrightarrow{\star} \tuple{q_4, nnn'lm', [\ihsv{amn}{a}, \ihsv{mn}{m},
	\ihsv{n}{n}]}\\
	&\xrightarrow{3} \tuple{q_5, nn'lm', [\ihsv{amn}{a}, \ihsv{mn}{m},
	\ihsv{n}{n}]}\\
	&\xrightarrow{\pclose{2}} \tuple{q_6, nn'lm', [\ihsv{amn}{a},
	\ihsv{mn}{n}]}\\
	&\xrightarrow{2} \tuple{q_7, n'lm', [\ihsv{amn}{a}, \ihsv{mn}{n}]}\\
	&\xrightarrow{\star} \tuple{q_8, n'lm', [\ihsv{amnl}{a}, \ihsv{mnl}{n},
	\ihsv{l}{l}]}\\
	&\xrightarrow{\cof{1}} \tuple{q_9, lm', [\ihsv{amnln'}{n'},
	\ihsv{mnln'}{n}, \ihsv{ln'}{l}]}\\
	&\xrightarrow{3} \tuple{q_{10}, m', [\ihsv{amnln'}{n'}, \ihsv{mnln'}{n},
	\ihsv{ln'}{l}]}\\
	&\xrightarrow{\cof{2}} \tuple{q_{11}, \epsilon, [\ihsv{amnln'm'}{n'},
	\ihsv{mnln'm'}{m'}, \ihsv{ln'm'}{l}]}\\
	&\xrightarrow{\pclose{3}} \tuple{q_{12}, \epsilon, [\ihsv{amnln'm'}{n'},
	\ihsv{mnln'm'}{m'}]}\\
	&\xrightarrow{\pclose{2}} \tuple{q_{13}, \epsilon,
	[\ihsv{amnln'm'}{n'}]}
       \end{align*}
       where any names $m \lfresh a$, $n \lfresh am$, $l \lfresh an$, $n'
       \lfresh amnl$ and $m' \lfresh mnln'$.
       Hence, the \cda\ accepts the language in-contexts.

 \item For the NRE $\etw{[]}{\npexp{n}{n \npexp{m}{m \pexp{l}{l}{m} m
       \npexp{l}{\cof{n} l \cof{m}}}}}{[]}$, the language in-contexts and the
       \cda\ are
       \begin{align*}
	C:& []\\
	\langof(\ne):& \bigl\{lmnnn'l'm' \mid m \lfresh l, n \lfresh lm,
	l' \lfresh ln, n' \lfresh lnl', n' \gfresh^1 lmnl', m' \lfresh n'nl', m' \gfresh^2
	mnl'n' \bigr\}\\
	E:& []
       \end{align*}
       \begin{minipage}{\linewidth}\centering
	\includegraphics[scale=.4]{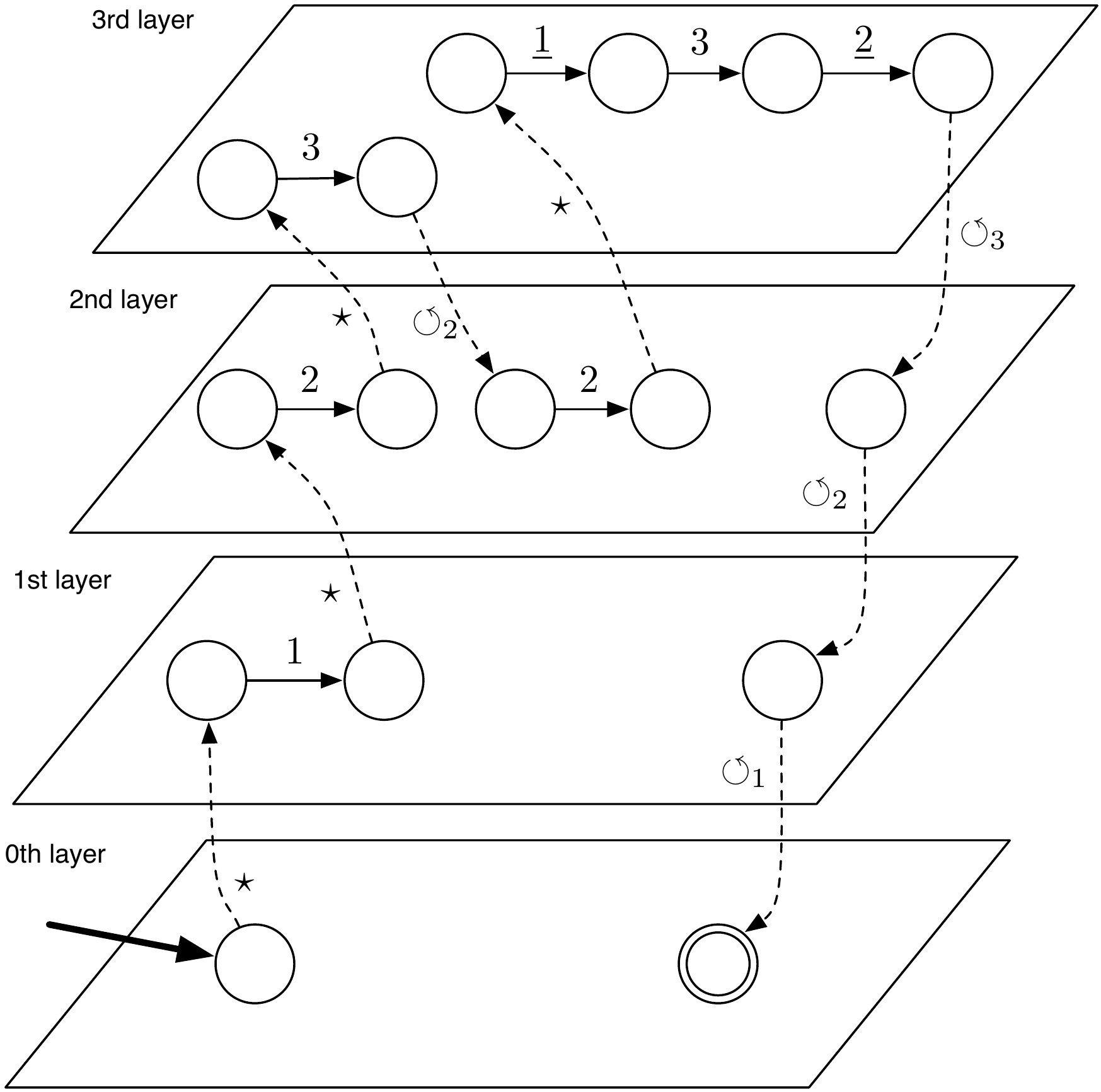}
       \end{minipage}
       \\[.5pc]
       This step abstracts the first register in case (9).
       The initial assignment is empty.
       The run is:
       \begin{align*}
	\tuple{q_0, lmnnn'l'm', []} &\xrightarrow{\star} \tuple{q_1, lmnnn'l'm',
	[\ihsv{l}{l}]}\\
	&\xrightarrow{1} \tuple{q_2, mnnn'l'm', [\ihsv{l}{l}]}\\
	&\xrightarrow{\star} \tuple{q_3, mnnn'l'm', [\ihsv{lm}{l},
	\ihsv{m}{m}]}\\
	&\xrightarrow{2} \tuple{q_4, nnn'l'm', [\ihsv{lm}{l}, \ihsv{m}{m}]}\\
	&\xrightarrow{\star} \tuple{q_5, nnn'l'm', [\ihsv{lmn}{l}, \ihsv{mn}{m},
	\ihsv{n}{n}]}\\
	&\xrightarrow{3} \tuple{q_6, nn'l'm', [\ihsv{lmn}{l}, \ihsv{mn}{m},
	\ihsv{n}{n}]}\\
	&\xrightarrow{\pclose{2}} \tuple{q_7, nn'l'm', [\ihsv{lmn}{l},
	\ihsv{mn}{n}]}\\
	&\xrightarrow{2} \tuple{q_8, n'l'm', [\ihsv{lmn}{l}, \ihsv{mn}{n}]}\\
	&\xrightarrow{\star} \tuple{q_9, n'l'm', [\ihsv{lmnl'}{l},
	\ihsv{mnl'}{n}, \ihsv{l'}{l'}]}\\
	&\xrightarrow{\cof{1}} \tuple{q_{10}, l'm', [\ihsv{lmnl'n'}{n'},
	\ihsv{mnl'n'}{n}, \ihsv{l'n'}{l'}]}\\
	&\xrightarrow{3} \tuple{q_{11}, m', [\ihsv{lmnl'n'}{n'},
	\ihsv{mnl'n'}{n}, \ihsv{l'n'}{l'}]}\\
	&\xrightarrow{\cof{2}} \tuple{q_{12}, \epsilon, [\ihsv{lmnl'n'm'}{n'},
	\ihsv{mnl'n'm'}{m'}, \ihsv{l'n'm'}{l'}]}\\
	&\xrightarrow{\pclose{3}} \tuple{q_{13}, \epsilon,
	[\ihsv{lmnl'n'm'}{n'}, \ihsv{mnl'n'm'}{m'}]}\\
	&\xrightarrow{\pclose{2}} \tuple{q_{14}, \epsilon,
	[\ihsv{lmnl'n'm'}{n'}]}\\
	&\xrightarrow{\pclose{1}} \tuple{q_{15}, \epsilon, []}
       \end{align*}
       where any names $l$, $m \lfresh l$, $n \lfresh lm$, $l' \lfresh ln$, $n'
       \lfresh lmnl'$ and $m' \lfresh mnl'n'$.
       Therefore, the \cda\ accepts the language for the NRE $\npexp{n}{n
       \npexp{m}{m \pexp{l}{l}{m} m \npexp{l}{\cof{n} l \cof{m}}}}$.
\end{enumerate}


\end{document}